\documentclass[12pt]{article}

\usepackage{amsmath}
\usepackage{amssymb}
\usepackage{graphicx}
\usepackage{layout}
\usepackage[table]{xcolor}
\RequirePackage[round]{natbib}
\RequirePackage[colorlinks,citecolor=blue,urlcolor=blue]{hyperref}
\usepackage{theoremref}
\usepackage{amsthm}
\usepackage{bm}
\usepackage{float}
\usepackage{caption}
\usepackage{bbm}
\usepackage{mathabx}
\let\hat\widehat
\let\tilde\widetilde
\usepackage{ulem}
\usepackage{subcaption}
\usepackage{lineno}
\usepackage{docmute}

\newtheorem{theorem}{Theorem}[section]
\newtheorem{lemma}[theorem]{Lemma}

\newtheorem{proposition}[theorem]{Proposition}
\newtheorem{remark}[theorem]{Remark}

\theoremstyle{definition}

\theoremstyle{definition}

\theoremstyle{definition}

\newcommand{\indep}{\perp \!\!\! \perp}

\makeatletter

\newcommand{\Rmnum}[1]{\expandafter\@slowromancap\romannumeral #1@}
\makeatother

\renewenvironment{proof}{{\bf Proof.}}{$\Box$}

\newcommand{\yingqi}{\textcolor{black}}
\renewcommand{\P}{\mbox{$\mathbb{P}$}}
\newcommand{\E}{\mbox{$\mathbb{E}$}}
\newcommand{\R}{\mbox{$\mathbb{R}$}}
\newcommand{\G}{\mathbb{G}}

\newcommand{\Var}{\mathrm{Var}}

\newcommand{\bbeta}{\boldsymbol\beta}
\newcommand{\bX}{\mathbf{X}}
\newcommand{\bS}{\mathbf{S}}
\newcommand{\bs}{\mathbf{s}}
\newcommand{\bgamma}{\boldsymbol\gamma}
\newcommand{\tX}{\tilde{\mathbf{X}}}
\newcommand{\bU}{\mathbf{U}}
\newcommand{\bSigma}{\mathbf{\Sigma}}
\newcommand{\bx}{\mathbf{x}}
\newcommand{\tx}{\tilde{\mathbf{x}}}
\newcommand{\bphi}{\bm{\phi}}
\newcommand{\bQ}{\mathbf{Q}}
\newcommand{\bA}{\mathbf{A}}
\newcommand{\bZ}{\mathbf{Z}}
\newcommand{\tZ}{\tilde{\mathbf{Z}}}
\newcommand{\bg}{\mathbf{g}}
\newcommand{\gc}{\textcolor{black}}
\newcommand{\blind}{1}

\catcode`@=11
\newskip\beforeproofvskip
\newskip\afterproofvskip
\beforeproofvskip=\medskipamount
\afterproofvskip=\bigskipamount

\def\prooftag{Proof}
\def\proofskip{\enspace}
\def\proof{\@ifnextchar[{\@@proof}{\@proof}}
\def\@startproof{\par\vskip\beforeproofvskip\leavevmode}
\def\@proof{\@startproof{\scshape\prooftag.}\proofskip}
\def\@@proof[#1]{\@startproof {\scshape\prooftag #1.}\proofskip}

\catcode`@=12

\begin{document}

\date{}
\let\combinedbibliographystyle\bibliographystyle
\let\combinedbibliography\bibliography
\renewcommand{\bibliographystyle}[1]{}
\renewcommand{\bibliography}[1]{}

\def\spacingset#1{\renewcommand{\baselinestretch}%
{#1}\small\normalsize} \spacingset{1}
\if1\blind
{
  \title{\bf Long-term effect estimation when combining clinical trial and observational follow-up datasets}
  \author{Gang Cheng
    \\
    Department of Statistics, University of Washington\\
    and \\
    Yen-Chi Chen \\
    Department of Statistics, University of Washington \\
    and \\
    Joseph M. Unger \\
    Public Health Sciences Division, Fred Hutchinson Cancer Research Center \\
    and \\
    Cathee Till \\
    Public Health Sciences Division, Fred Hutchinson Cancer Research Center \\
    and \\ 
    Ying-Qi Zhao\\
     Public Health Sciences Division, Fred Hutchinson Cancer Research Center}
  \maketitle
} \fi
\if0\blind
{
  \bigskip
  \bigskip
  \bigskip
  \begin{center} 
    {\LARGE\bf Long-term effect estimation when combining clinical trial and observational follow-up datasets}
\end{center}
  \medskip
} \fi

\bigskip
\begin{abstract}
Combining experimental and observational follow-up datasets has received much
attention lately. In a survival setting, recent work has used Medicare claims to
extend the follow-up period for participants in a prostate cancer clinical trial. This
allows the estimation of the long-term effect that cannot be estimated by the
trial data alone. In this paper, we study the estimation of long-term effect when
participants in a clinical trial are linked to an observational follow-up dataset. Such linkages are often incomplete and we formulate incomplete linkages as a missing data problem.  We
use the popular Cox model to define the long-term effect and we propose two approaches to deal with the missing data problem. The first approach, termed non-linked-as-censored (NLAC), is a simple approach that works when Cox model is correctly specified and linkage satisfies a conditionally independent assumption.  To gain robustness against model mis-specification, \yingqi{we propose an inverse probability of linkage weighted approach, along with the augmented inverse probability of weighted method,} based on a novel conditional linking at random (CLAR) assumption.  We further extend our approach to incorporate time-dependent covariates. Simulation results confirm the validity of our method and we apply our methods to the SWOG study.
\end{abstract}


\noindent%
\vfill

\newpage
\spacingset{1.9} 

\section{Introduction} 

With the increasing availability of electronic health data, combining experimental and observational datasets has been widely applied in public health research \citep{warren2002overview,gilbert2018guild}.
In a time-to-event setting,  we consider the setup when data from a clinical trial is combined with an observational follow-up dataset, such as electronic health records or administrative claims.  Clinical trials often study the effect of a particular treatment for a fixed period of time and it might not be long enough to determine the maximum benefit of the treatment.  
In contrast,  an observational dataset such as Medicare claims naturally extends the follow-up period for clinical trial participants at minimal cost. This enables the estimation of the long-term effect for the treatment after the clinical trial. 
To combine the observational follow-up dataset with the clinical trial data, records belonging to the same individual can be linked with unique identifiers from both datasets. 
In this paper,  we use Cox model \citep{cox1972regression} to define the long-term effect as the parameter for treatment when participants from the clinical trial are linked to an observational follow-up dataset. 

For a real data example, the Prostate Cancer Prevention Trial (PCPT) was previously launched to examine whether finasteride\footnote{a treatment that inhibits the development potent androgen that fuels the malignancy of prostate cancer}
 could prevent the development of prostate cancer (PC). 
PCPT 
showed that seven years of finasteride reduced PC risk by 25\% \citep{thompson2003influence}. 
However, it was unclear if seven years' of trial follow-up sufficed to determine the maximum benefit of the treatment. Further,  the reduced risk of prostate cancer for subjects receiving finasteride might not be maintained after finasteride discontinuation \citep{unger2018using}. 
A later 
 study linked Medicare claims to the clinical records for participants in PCPT with their social security numbers (SSN) \citep{unger2018using} to estimate the long-term effect of finasteride on prostate cancer (PC) development. In this example, PCPT is the clinical trial and Medicare claim is the observational follow-up dataset.
Medicare claims extend the follow-up periods up to a maximum of 20 years compared to 7 years by the PCPT. Thus, we can observe more diagnosis times of PC within the Medicare claims dataset. 

However, not every participant in the clinical trial can be linked to the observational dataset. 
For the PCPT-Medicare claims example, some participants might not be willing to share their SSNs or they may be enrolled in health maintenance organization (HMO) and Medicare claims are not applicable to HMO individuals \citep{unger2018using}. 
With incomplete linkages, survival outcomes in the observational dataset might be missing for some participants. 
For a participant censored in the clinical trial, meaning he was not diagnosed with PC within the clinical trial, his survival outcome in the observational dataset would be missing if he is unlinked. On the other hand, if a participant was diagnosed with PC within the clinical trial, his survival outcome has been already observed within the clinical trial and the linkage to the observational follow-up dataset is in fact not necessary. 
This suggests that the missingness of survival outcome depends both on the linkage status and whether a participant was censored in the clinical trial or not. 

To address the missing survival outcomes, a complete-case analysis that only includes linked participants would ignore all the unlinked participants with observed survival outcomes within the clinical trial. However, simply adding those unlinked participants to the complete-case analysis also results in a biased estimate. This would cause the missingness of the survival outcomes to depend on itself and the missingness is then missing not at random (MNAR). 
To properly incorporate those unlinked participants with observed survival outcomes, \gc{we propose two approaches. The first approach, termed non-linked-as-censored (NLAC), 
treats unlinked participants who are censored in the clinical trial as censored, effectively ignoring their survival outcomes in the observational dataset. 
NLAC is a simple idea that works when the Cox model is correctly specified and the linkage satisfies certain conditionally independent assumption. However, when Cox model is mis-specified, we show that NLAC no longer provides unbiased estimates for the population parameters. To gain more robustness against model mis-specification, we develop a second approach based on an inverse probability of linkage weighted (IPLW)  partial likelihood estimator. This approach
involves modeling the linkage probability directly for unlinked participants that are censored in the clinical trial. We show that unbiased estimates can be obtained even if Cox model is mis-specified, under a novel assumption, termed conditional linking at random (CLAR). } 
We further prove the asymptotic normality and consistency of the IPLW estimator. 
Our approach allows inclusion of time-dependent covariates for more flexibility.  
While there has been plenty work \citep{binder1992fitting,robins1993information,lin2000fitting,qi2005weighted} on proving the asymptotic convergence for an inverse probability weighted (IPW) partial likelihood estimator when there are only time-independent covariates,
their proof cannot be easily generalized to the case when there are time-dependent covariates  \citep{breslow2007weighted}. 
 To this end, we establish an IPW type empirical process weak convergence result  that builds on the work in \cite{saegusa2013weighted} and borrow the techniques from \cite{lin1989robust} to extend the theoretical results to include time-dependent covariates. \yingqi{Finally, we develop an augmented inverse probability of linkage weighted (AIPLW) estimator. The estimator is doubly robust in the sense that it provides consistent estimates when either the linkage probability is consistently estimated or the involved regression functions are consistently estimated. }

\emph{Related work.} 
\yingqi{Our work contributes to the expanding body of literature on causal inference methods that integrate data from randomized clinical trials and observational studies, which leverage the strength from each for better power and external validity, see \cite{colnet2024causal} for a comprehensive review.} 
Indeed, there is a growing interest in combining different datasets and studying the treatment effect on long-term outcome in causal inference \citep{rosenman2018propensity,rosenman2020combining,kallus2020role,athey2020combining}.
All these work focuses on using experimental and observational datasets that contain different set of individuals, which is different from our setup. 
IPW  
has also been widely applied in the survival analysis setting \citep{binder1992fitting,robins1994estimation,robins2000correcting,hernan2000marginal,lin2000fitting,qi2005weighted,tsiatis2007semiparametric,breslow2007weighted,saegusa2013weighted}. 
\cite{robins2000correcting} applied inverse probability of censoring weights to estimate Cox model that adjusts for dependent censoring by utilizing data collected on time-dependent prognostic factors. 
IPW has also been applied for Cox models with two-phase stratified sampling under right censoring \citep{binder1992fitting,lin2000fitting,breslow2007weighted}, while \cite{saegusa2013weighted} further studied the problem of two-phase sampling for Cox model under interval censoring with IPW.  
Our approach is different from all previous works as we also allow for time-dependent covariates. 

%

\emph{Outline.}
In Section \ref{sec::BG}, we provide background and notations for our methodological developments on assessing long-term effect through linked dataset. 
\yingqi{In Section \ref{sec::methods}, we introduce the NLAC, IPLW and AIPLW methods and their associated properties}. 
We conduct simulation studies in Section~\ref{sec::simulation} to illustrate the validity of the proposed method.
We apply our approach to the linked SWOG PCPT-Medicare study in Section~\ref{sec::SWOG}.
In Section~\ref{sec::discussion}, we conclude this paper and point out some possible future directions.

\section{Long-term effect assessment through linked data}	\label{sec::BG}
We first consider the oracle setting that all participants from the clinical trial are linked. 
We make a ``no gap'' assumption such that there is no gap between a participant's last recorded date within the clinical trial and the start date of the observational follow-up dataset. 
This ``no gap'' assumption eliminates the possibility of interval censoring in which a participant is diagnosed with the event of interest while not under observation.  For simplicity,  we make this ``no gap'' assumption to focus on the right censoring problem and we discuss how to relax this no gap assumption in Appendix G (supplementary material).

Time is measured since enrollment in the clinical trial. We define $T$ as the failure time,  $C_{1}$ as the censoring time within the clinical trial and $Q = I(T \leq C_1)$
as the censoring indicator for the clinical trial.  
We use $\tau_1$ to denote the end time of clinical trial.  
Possible reasons for censoring in the clinical trial include loss to follow-up and administrative censoring.  In contrast,  the length of 
observational follow-up dataset is often determined by the data availability and vary from person to person.  We set $\tau_2$
with $\tau_2 > \tau_1$ as the common end time for observational follow-up dataset and assume that there are a significant proportion of participants at risk after $\tau_2$. 
Thus, we are interested in estimating the long-term treatment effect on survival up to time $\tau_2$ using data from clinical trial records and observational follow-up.  
Similarly,  we define $C_{2}$ as the censoring time in the observational follow-up dataset.  Possible reasons for censoring in the observational follow-up include short
coverage such that a participant is not covered long enough by the observational dataset,  administrative censoring where a participant is 
event-free and covered by observational follow-up until $\tau_2$. 

We use $C = \max(C_1, C_2)$ to denote the actual censoring time and let $\tilde{T} = \min\{T, C\}$ denote the actual observed time and $\Delta = I(T \leq C)$ be the 
censoring indicator throughout the entire follow-up period. 
Let $\bX \in \R^p$ denote baseline characteristics, such as clinical factors and treatment assignment.  We also use $A$ to denote the treatment assignment when necessary.  
We make the independent censoring assumption that $T$ and $(C_1, C_2)$ are conditionally independent given $\bX$.  Thus,  $T$ is conditionally independent of $C = \max(C_1, C_2)$ given $\bX$. 

\subsection{Cox model and Long-term effect}  \label{sec::sec_Cox} 
We use Cox model to define the long-term effect and we allow for the possibility of model mis-specification. 
Assuming that there are $n$ participants in the clinical trial and they are all linked to the observational follow-up dataset. 
$(\bX_i,  Q_{i}, \tilde{T}_i, \Delta_i)$ for $i = 1, \ldots, n$ are all observed.  
To account for potential different treatment effects between clinical trial and the observational follow-up period,  we consider the following model with a change point at time $\tau_1$ \citep{liang1990Cox,pons2003estimation}
\begin{align}
\lambda(t |\bX; \beta_0, \bbeta_1, \theta) = \lambda_0(t) \exp[(\beta_1 + \theta  I_{t > \tau_1}) A + {\bbeta'_0}^T \bX_{-A}].
\label{eqn::eqn_change}
\end{align}
$\beta_1$ represents the effect of treatment in the clinical trial, while $\theta$ represents the difference of the treatment effect between observational follow-ups and clinical trial.  $\beta_1 + \theta$ represents the long-term treatment effect. 
When $\theta$ = 0,  model \eqref{eqn::eqn_change} reduces to a typical Cox model with no change point. 
For notational simplicity, we use $\exp(\bbeta^T \bX)$ to represent $\exp[(\beta_1 + \theta  I_{t > \tau_1}) A + {\bbeta'_0}^T \bX_{-A}]$ in the subsequent. 
 Following the notation in \cite{lin1989robust},  for the $i$-th individual, let $\lambda_i(t) = \lambda(t | \bX_i)$ be the true hazard function, $N_i(t) = I(\tilde{T}_i \leq t, \Delta_i = 1)$ and $Y_i(t) = I(\tilde{T}_i \geq t)$.  For $k = 0, 1, 2$, define 
\begin{align*}
& \bS_n^{(k)}(t) = \frac{1}{n} \sum_{i=1}^n Y_i(t) \lambda_i(t) \bX_i ^{(k)},  \quad \bs^{(k)}(t) = \E[\bS_n^{(k)}(t)] \\
& \bS_n^{(k)}(\beta, t) = \frac{1}{n} \sum_{i=1}^n Y_i(t) \exp(\bbeta^T \bX_i) \bX_i^{(k)}, \quad \bs^{(k)}(\beta, t) = \E[\bS_n^{(k)}(\bbeta, t)]
\end{align*}
where for a column vector $\mathbf{a}$, denote $\mathbf{a}^{\otimes 2} = \mathbf{a} \mathbf{a}^T$, $\mathbf{a}^{\otimes 1} = \mathbf{a}$ and $\mathbf{a}^{\otimes 0}$ refers to the scalar 1.  
With Cox model as a working model, the parameter of interest is $\bbeta_0^*$
that solves the following equations \citep{andersen1982Cox,lin1989robust}
\begin{equation}
\bU_0(\bbeta) = \E\left[ \Delta \left(\bX - \frac{\bs^{(1)}(\bbeta, \tilde{T})}{\bs^{(0)}(\bbeta, \tilde{T})}\right)\right] = \int^{\tau_2}_0 \bs^{(1)}(t) dt - \int^{\tau_2}_0 \frac{\bs^{(1)}(\bbeta, t)}{\bs^{(0)}(\bbeta, t)} \bs^{(0)}(t) dt .
\label{eq::U0}
\end{equation}
When Cox model is correctly specified with true parameter $\bbeta_0$,  we have $\bbeta_0^* = \bbeta_0$.  
When Cox model is mis-specified, parameter $\bbeta_0^*$ that solves equation \eqref{eq::U0}
is still well-defined and $\bbeta_0^*$ can be interpreted meaningfully for many mis-specified Cox models
\citep{lin1989robust}. 

\subsection{The linkage problem }  \label{sec::sec_assumption}
Now we consider the more realistic setup that not every participant is linked to the observational follow-up dataset. 
We use $L$ to denote the linkage status.  $L = 1$ means that the participant is linked to the observational follow-up dataset and  $L = 0$ means unlinked. 
We classify participants into three classes based on their linkage status $L$ and in-trial censoring indicator $Q$. 


{\bf Class 1: $L=1$.}
This class contains all participants linked to the observational follow-up dataset, 
where we observe the survival outcome: $(L, \tilde{T}, \Delta, Q, \bX)$.


{\bf Class 2: $L=0, Q=1$.} This class represents participants diagnosed with the event of interest within the clinical trial but unlinked to observational follow-up.
We still have complete information $(L,\tilde{T},\Delta, Q, \bX)$. 


{\bf Class 3: $L=0, Q=0$.} This class includes participants who did not experience the event of interest during the clinical trial, and were not linked to observational follow-up dataset.
$C_2$ is missing and thus $C$ is missing,
which causes the missingness of 
$\tilde{T}$ and $\Delta$. 
We only observe $(L, Q, \bX)$.  

\begin{figure}
\center
\includegraphics[height=2in]{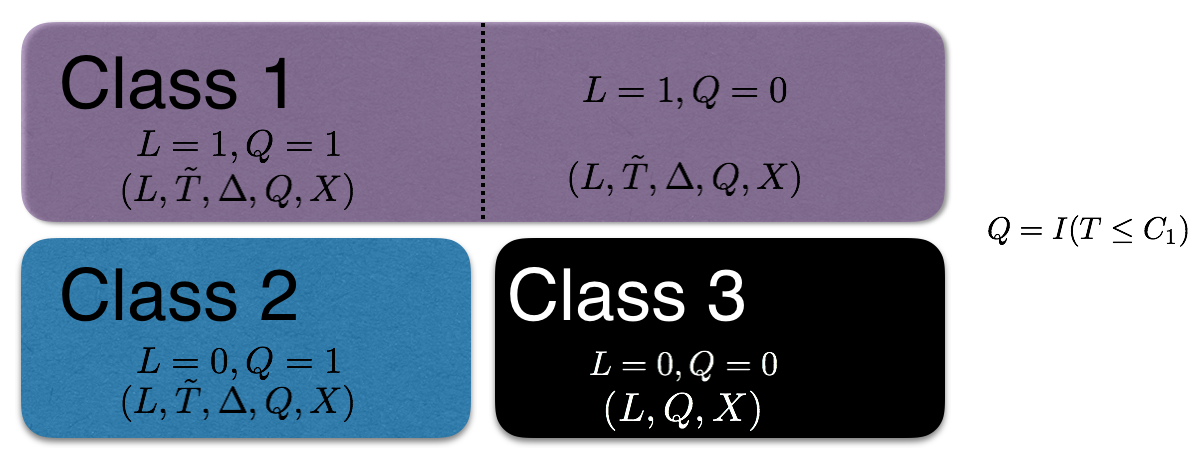}
\caption{A diagram illustrating the three classes of participants and how it is defined via the linkage indicator $L$
and in-trial censoring indicator $Q$. 
}
\label{fig::ex01}
\end{figure}

Figure~\ref{fig::ex01} provides a summary of the three classes defined above. We have completely observed data $(L,\tilde{T}, \Delta, Q, \bX)$ in class 1 and 2, while survival outcomes $\Tilde{T}, \Delta$ are missing in class 3. To deal with the missing $\tilde{T}$ and $\Delta$, we consider two possible approaches. 

\section{Methods}\label{sec::methods}
\subsection{Non-linked-as-censored (NLAC) approach}
\label{sec::NLAC} 
\subsubsection{The assumption and the NLAC approach}  
The first approach, termed 
 \textit{non-linked-as-censored (NLAC)}, treats participants from Class 3 ($L = 0, Q = 0$) as censored and sets their censoring time as $C = C_1$. 
 Then we can fit the Cox regression with all participants.
For participants from Class 3, 
NLAC is a natural idea that simply uses their censoring time in clinical trial $C_1$ as the censoring time for the entire follow-up period. We can view the censoring time as
$$
C_{\text{NLAC}} = L \max(C_1, C_2) + (1 - L) C_1
= L[\max(C_1, C_2) - C_1] + C_1
$$
and NLAC works when $C_{\text{NLAC}} \indep T | \bX$. 
We now present one sufficient linkage assumption for NLAC to work when Cox model is correctly specified.
\gc{\begin{theorem}
    \thlabel{prop::prop_naive_method_linkage_assumption}
When Cox model is correctly specified and $(C_1, C_2) \indep T | \bX$, if the linkage satisfies the following conditional independent assumption: 
$$
\text{\bf (A1)} \quad L \indep T | \bX, C_1, C_2, 
$$ 
NLAC provides consistent and asymptotically normal 
estimates for $\bbeta_0$. 
\end{theorem}}
\gc{
The proof of Theorem \ref{prop::prop_naive_method_linkage_assumption} can be found in Appendix A. On the other hand, when Cox model is misspecified, the population parameter $\bbeta_0^*$ will depend on the actual distribution of the censoring time $C = \max(C_1, C_2)$. Thus, NLAC leads to biased estimates of $\bbeta_0^*$ as the distribution of the censoring time is modified. See Section \ref{sec::motivating_example}  for a toy example. 
}

\subsubsection{Alternative approaches and a motivating example}    \label{sec::motivating_example} 
There are two alternative approaches that practitioners may use, which are
\begin{itemize}
\item Complete-case (CC) analysis that only includes participants that are linked. These corresponds to participants in Class 1 with $L = 1$.
\item Complete-case analysis plus (CC+) that includes not only participants that are linked, but also participants with $Q = 1$.  This corresponds to participants in Class 1 and 2 in Figure \ref{fig::ex01}. These are the participants with $L + Q > 0$. 
\end{itemize}
We simulate data according to a Cox model with hazard function specified by covariates $X_1, X_2, X_3^2$ and a Cox model with covariates $X_1, X_2, X_3$ is fitted. More details are given in Appendix H.5 (supplementary material). 
 Figure~\ref{fig::fig_ci_data_3}
 presents the 95\% confidence intervals for one of the parameters.  Oracle method refers to the approach that all participants in the clinical trial are linked.
Among all four approaches, CC+ gives the most biased estimates. CC gives less biased estimates than CC+.  As anticipated, NLAC also leads to biased estimates.  
Indeed, in order for CC to be consistent, one sufficient condition is linking completely at random (LCAR), i.e., $L\perp (\tilde T, \Delta)$.  This is similar to the missing completely at random (MCAR) \citep{little2019statistical} condition. 
The LCAR condition is a strong condition and may contradict the data\footnote{we can easily test this condition}.
For CC+, the missingness of survival outcomes now also depends on survival outcomes itself as a participant would be included if $L + Q > 0$. Thus, the missingness can be viewed as MNAR and CC+ will always lead to biased estimates. 
\begin{figure}
\centering
\includegraphics[scale = 0.55]{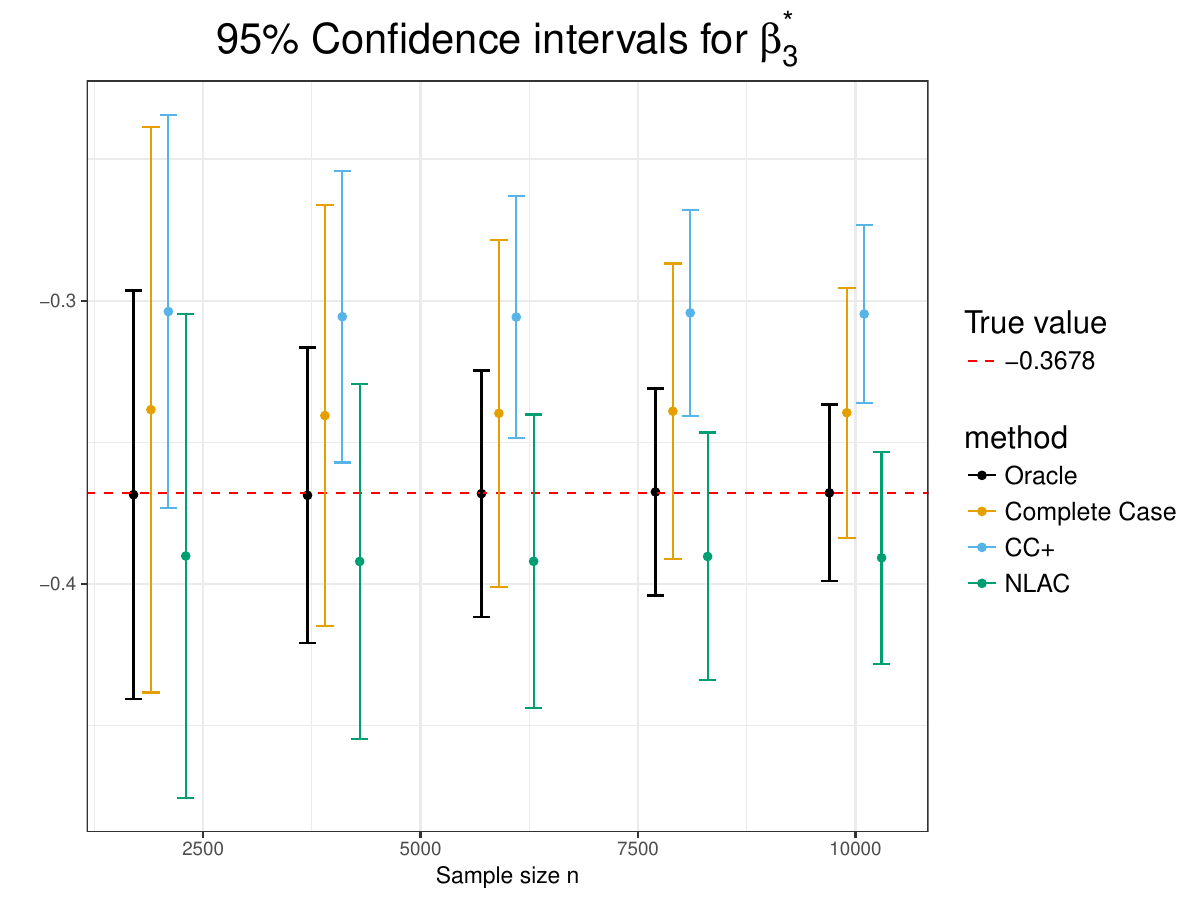}
\caption{
  CC+, CC and NLAC all give inconsistent estimates of $\beta_3^*$.  }
\label{fig::fig_ci_data_3}
\end{figure}

\subsection{Inverse probability of linkage weighted (IPLW) approach}	
\subsubsection{The conditional linking at random (CLAR) assumption}
\gc{ 
To gain robustness against model mis-specification, we propose the second approach, where different from NLAC as for participants in Class 3, we treat their survival outcomes as missing.  
Particularly, we assume the following \textit{conditional linking at random (CLAR)} assumption:
$$
\text{\bf (A2)} L \indep (\tilde{T}, \Delta) | \bX, Q = 0
$$
More specifically, CLAR states that for a participant that is censored in the clinical trial,  their survival outcomes are independent of the linkage status given their baseline covariates, clinical factors and treatment assignment.  
Note that $\bf (A1)$ actually implies that 
$$
\text{\bf (A1)$'$} \ L \indep (\tilde{T}, \Delta) | \bX, C_1, C_2.
$$
So roughly speaking, the main difference between $\bf (A1)$ and CLAR is that we are replacing $C_1$ and $C_2$ by $Q = 0$. }
\gc{\begin{remark}
    We briefly compare CLAR assumption $\bf (A2)$ with the classical MAR type assumption \citep{rubin1976inference, little2019statistical}, which can be written as follows
\begin{equation}
L \indep (\tilde{T}, \Delta) |\bX, Q.
\label{assumption::mar}
\end{equation}
CLAR is actually implied by MAR \eqref{assumption::mar}. 
However, CLAR is restricting the conditional independence to the subpopulation with $Q = 0$, while MAR \eqref{assumption::mar} is assuming the conditional independence for the whole population. To see why this is important, note that MAR \eqref{assumption::mar} also implies that 
\begin{equation}
L \indep (\tilde{T}, \Delta) | Q = 1, \bX.
\label{assumption::mar_q_1}
\end{equation}
Interestingly when $Q = 1$, both $\tilde{T}$ and $\Delta$ are observed, meaning that \eqref{assumption::mar_q_1} might in fact contradict the data. In contrast, with no assumptions for participants with $Q = 1$,  CLAR is non-parametrically identifiable, i.e., they will never contradict the data \citep{robins2000sensitivity}.
Further, CLAR allows the use of more data compared to MAR. Specifically, an IPW estimator under the MAR assumption would only involve participants that are linked, thereby excluding participants from Class 2. 
 Further discussions on potential linkage assumptions are given in Appendix E (supplementary material).
\end{remark}}

\subsubsection{Inverse probability of linkage weighted (IPLW) estimator}

\gc{Since we view $\tilde{T}$ and $\Delta$ as missing for participants in Class 3}, we can not use the classical partial likelihood for Cox model to estimate the parameters.  
We first illustrate our approach with time-independent covariates only. 
We start by writing the regular log-likelihood for Cox model as follows:
\begin{align*}
l_n(\bbeta, \lambda_0) &= \frac{1}{n} \sum_{i=1}^n l(\tilde{T}_i, \Delta_i, \bX_i; \bbeta, \lambda_0) = \frac{1}{n}\sum_{i=1}^n \log \left(\lambda(\tilde{T}_i | \bX_i)^{\Delta_i} S(\tilde{T}_i | \bX_i)\right) 
\end{align*}
where $S(t | \mathbf{x})$ is the conditional survival function for the failure time $T$.  However,  $l_n(\bbeta, \lambda_0)$ is unidentifiable since we do not observe $(\tilde{T}_i, \Delta_i)$ for participants in Class 3 of Figure~\ref{fig::ex01}.  
To resolve the identifiability issue, consider the expected log-likelihood 
$$
\E(l(\bbeta,\lambda_0)) = \E\left[\log\left( \lambda(\tilde{T} | \bX)^{\Delta} S(\tilde{T} | \bX) \right)\right]
$$ 
where $\E$ is the expectation with respect to random variable $(\tilde{T}, \Delta, \bX)$ and $l(\bbeta,\lambda_0) =l(\tilde{T},\Delta, \bX; \bbeta, \lambda_0)$.  By the law of large number,  we have $l_n(\bbeta,\lambda_0) \rightarrow_p \E(l(\bbeta,\lambda_0))$.  

\begin{proposition} \thlabel{prop::prop_ipw_llk}
Under CLAR assumption,  
we have
\begin{equation}
\E(l(\bbeta, \lambda_0)) = \E\left[ \frac{I(L + Q > 0)  l(\bbeta, \lambda_0)}{Q + (1 - Q)P(L = 1 | \bX, Q = 0)} \right]
\label{eq::ipw_whole}
\end{equation}
\end{proposition}
\thref{prop::prop_ipw_llk} shows that $\E(l(\bbeta, \lambda_0))$ can be expressed in the inverse probability of linkage weighted (IPLW) form and the proof can be found in Appendix D (supplementary material).  We assume a logistic regression model for the linkage probability $P(L = 1 | \bX, Q = 0)$ for simplicity such that 
$$
P(L = 1 | \bX, Q = 0; \bgamma_0) = \frac{\exp(\bgamma_0^T \tX)}{1 + \exp(\bgamma_0^T \tX)} = \pi_{\bgamma_0}(\bX)
$$
with $\bgamma_0 \in \R^{p + 1}$ and $\tX = (1,  \bX^T)^T$. In particular, $\bgamma_0$ can be estimated by the maximum likelihood estimator $\hat{\bgamma}_n$.  
Using result \eqref{eq::ipw_whole},  an IPLW estimator of $\E(l(\bbeta, \lambda_0))$ is
\begin{align*}
l_{n}(\bbeta,\lambda_0) & = \frac{1}{n} \sum_{i=1}^n \frac{I(L_i + Q_{i} > 0)}{Q_i + (1 - Q_i) \pi_{\hat{\bgamma}_n}(\bX_i)} [\Delta_i\log\lambda(\tilde{T}_i | \bX_i) + \log S(\tilde{T}_i | \bX_i)]
\end{align*}
with the log-likelihood being weighted by $\hat w_i = I(L_{i} + Q_i > 0)/[ Q_i + (1 - Q_i) \pi_{\hat{\bgamma}_n}(\bX_{i})]$.   
IPLW partial score can then be derived as
\begin{equation}
\label{eqn::weighted_pl}
\hat{\bU}_n(\bbeta) = \frac{1}{n} \sum_{i=1}^n \Delta_i \hat{w}_i \left(\bX_i - \frac{\sum_{j = 1}^{n} \hat{w}_j  I(\tilde{T}_j \geq \tilde{T}_i) \exp(\bX_j^T \bbeta) \bX_j }{\sum_{j = 1}^{n} \hat{w}_j I(\tilde{T}_j \geq \tilde{T}_i) \exp(\bX_j^T \bbeta)} \right),
\end{equation}
which is a sample analog of equation \eqref{eq::U0} and $\hat \bbeta_n$ can be obtained by solving equation \eqref{eqn::weighted_pl} using standard statistical software. Detailed derivations of \eqref{eqn::weighted_pl} can be found in Appendix B (supplementary material). 
In summary,   our method for estimating the regression parameters of Cox model consists of two steps:
\begin{itemize}
\item {\bf Step 1.} We estimate the linkage probability $P(L = 1| Q = 0, \bX;  \bgamma_0) = \pi_{\bgamma_0}(\bX)$ with logistic regression.  $\hat{\bgamma}_n$ can be obtained by the maximum likelihood estimation. 

\item {\bf Step 2.}
An individual with $\Delta_i = 1$ is weighted with weight $\hat w_i$ using the estimated linkage probability $\pi_{\hat \bgamma_n}(\bX_i)$.  More specifically, for participants with $Q_i = 1$, the weight is 1; for participants with $Q_i = 0$ and $L_i = 1$, the weight is ${1}/{\pi_{\hat{\bgamma_n}}(\bX_i)}$.  $\hat \bbeta_n$ is then obtained by solving \eqref{eqn::weighted_pl}. 
\end{itemize} 
 
\begin{remark}
It is possible that the CLAR assumption may not be entirely accurate or only holds approximately.  A sensitivity analysis \citep{little2012prevention} can be conducted to examine how the estimate  changes when we slightly perturb the underlying missing data assumption. 

Particularly, we can perform a sensitivity analysis to assess the
effect of violation of the CLAR assumption using the idea of exponential tilting \citep{kim2011semiparametric, shao2016semiparametric, zhao2017semiparametric}. 
We model the linkage probability as 
\begin{equation}
    \pi_{\bgamma_0, \alpha}(\bX, \tilde{T}) = 
\frac{\exp(\bgamma_0^T \bX+\alpha \tilde{T})}{1 + \exp(\bgamma_0^T \bX+\alpha \tilde{T})},
\label{eq:exponentialtilting}
\end{equation}
where $\alpha$ is the sensitivity parameter that \yingqi{has to be specified by an analyst}. \yingqi{
The parameter $\alpha$ may be interpreted as the effect of observed time $\tilde T$ on the linkage probability. Equation \eqref{eq:exponentialtilting} 
suggests that the value of $\alpha$ is the effect of one unit of time (in our case, a year) on the log-odds of linkage probability. The parameter $\alpha=0$ corresponds to the situation when the CLAR assumption is satisfied. A positive $\alpha$ refers to the assumption that individual with a longer observed time has a higher chance to be linked while a negative $\alpha$ is the opposite--the chance of linkage is lower if the observed time is higher. A larger value of $|\alpha|$ leads to a more severe violation of
the CLAR. }

We begin by selecting a sensitivity parameter vector $\alpha$'s first. Subsequently, for each $\alpha$, utilizing the linkage model \eqref{eq:exponentialtilting}, we can compute the modified weights attributed to $\alpha$ for use in \eqref{eqn::weighted_pl}, along with the corresponding parameter estimates. Finally, we can evaluate the resulting estimates over the range of $\alpha$ values, representing various degrees of violation from the CLAR assumption.
\end{remark}

\subsubsection{Time-dependent covariates}
 In practice, it is common for Cox regression to include time-dependent covariates and we now extend our IPLW method to incorporate the time-dependent covariates. 
 We \yingqi{build} on the work in \cite{lin1989robust} to extend our theoretical results to include time-dependent covariates with the IPLW partial likelihood estimator.  
Let $\bX_i(t) = (\bZ_{1i}^T, \bZ_{2i}(t)^T)^T \in \mathbb{R}^p$ denotes the covariates vector, where $\bZ_{1i} \in \mathbb{R}^{d_1}$ corresponds to the baseline (time-independent) covariates and $\bZ_{2i}(t) \in \mathbb{R}^{d_2}$ corresponds to the time-dependent covariates for $i = 1, \ldots, n$ at time $t$. \yingqi{We assume that $\bX(t)$ is observable for all time points and individuals.} We have $d_1 + d_2 = p$.  $\bZ_{2i}(t)$ can represent covariates that are continuously monitored during the clinical trial and observational follow-up datasets.  For Cox model with a change point \eqref{eqn::eqn_change},  $Z_{2i}(t) = I(t > \tau_1) A$.  
Let $\widebar{\bX}(t) = \{\bX(s):  s \in [0, t]\}$ denotes the history of covariates vector $\bX(s)$,  up to time $t$. 
To incorporate time-dependent covariates into the IPLW partial likelihood, we modify the CLAR assumption (A1) as following:

\begin{itemize}
\item[\bf (D1)] 
The linkage status satisfies that 
$$
P(L = 1 | \tilde{T}, \Delta, Q = 0,  \bar{\bX}(\tau_M)) = P(L = 1 | Q = 0, \bZ_1).
$$
The distribution of $\bZ_1$ is not concentrated on a $(d_1-1)$ dimensional affine subspace of $\R^{d_1}$.
\end{itemize}
(D1) assumes that linkage only depends on time-independent covariates $\bZ_1$ and also ensures the identifiability of $\bgamma_0$ (see, e.g., Example 5.40 of \cite{van2000asymptotic}). We can further relax this assumption such that linkage also depends on $\bZ_2(C_1)$, the value of time-dependent covariates at the censoring time in clinical trial.  For simplicity, we assume that linkage only depends on the baseline (time-independent) covariates. Based on assumption (D1),
we modify the weights as follows:
$$
w_i = \frac{I(L_i + Q_i > 0)}{Q_i + (1 - Q_i)\pi_{\gamma_0}(\bZ_{1i})} \qquad \hat{w}_i = \frac{I(L_i + Q_i > 0)}{Q_i + (1 - Q_i)\pi_{\hat{\gamma}_n}(\bZ_{1i})}
$$ 
with $\bgamma_0 \in \R^{d_1 + 1}$.  The IPLW partial score incorporating time-dependent covariates is now as follows:
\begin{equation}
\hat{\bU}_n(\bbeta) = \frac{1}{n} \sum_{i=1}^n \Delta_i \hat{w}_i \left\{\bX_i(\tilde{T}_i) - \frac{\bS_{n, w}^{(1)}(\bbeta, \tilde{T}_i)}{\bS_{n, w}^{0}(\bbeta, \tilde{T}_i)}\right\} = \frac{1}{n} \sum_{i=1}^n \hat{w}_i \int^{\tau_2}_0 \bX_i(t) d N_i(t)  -\int^{\tau_2}_0 \frac{\bS_{n, w}^{(1)}(\bbeta, t)}{\bS_{n, w}^{(0)}(\bbeta, t)} d \bar{N}(t)
\label{eq::Un}
\end{equation}
with
$
\bS_{n, w}^{(k)}(\bbeta, t) = \frac{1}{n} \sum_{i=1}^n \hat{w}_i  Y_i(t) \exp(\bbeta^T \bX_i(t)) \bX_i(t)^{\otimes k}
$ and $\bar{N}(t) = \frac{1}{n} \sum_{i=1}^n \hat{w}_i N_i(t)$.

We redefine $\bs^{(k)}(\bbeta, t) = \E[Y(t)\exp(\bbeta^T \bX(t))  \bX(t)^{\otimes k}]$ and $\bs^{(k)}(t) = \E[Y(t) \lambda(t | \bar{\bX}(t)) \bX(t)^{\otimes k}]$ where $\lambda(t | \widebar{\bX}(t))$ is the true hazard function for participants with covariates history $\widebar{\bX}(t)$. 
The estimated parameter $\hat{\bbeta}_n$ solves $\hat{\bU}_n(\bbeta) = \bm{0}$ and its population version $\bbeta_0^*$ solves $\bU_0(\bbeta) = 0$ with 
$$
\bU_0(\bbeta) = \E\left[ \Delta \left( \bX(\tilde{T}) - \frac{\bs^{(1)}(\bbeta, \tilde{T})}{\bs^{(0)}(\bbeta, \tilde{T})} \right)\right]  = \int^{\tau_2}_0 \bs^{(1)}(t) dt - \int^{\tau_2}_0 \frac{\bs^{(1)}(\bbeta, t)}{\bs^{(0)}(\bbeta, t)} \bs^{(0)}(t) dt 
$$

 Now we are ready to present the consistency and asymptotic normality of $\hat{\bbeta}_n$.

\begin{theorem}[Asymptotic results of $\hat{\bbeta}_n$]  \thlabel{thm::thm_Cox_right_asn_d_asp}
Let $\hat{\bbeta}_n$ be the solution to the equation $\hat{\bU}_n(\bbeta) = 0$.
Under assumptions (D1) - (D5), we have $\hat{\bbeta}_n \rightarrow_p \bbeta_0^*$ and 
$\sqrt{n}(\hat{\bbeta}_n - \bbeta_0^*) \rightarrow_d N(0, \bSigma_0^{-1}\bSigma_U \bSigma_0^{-1})$ where $\bSigma_0 = -\frac{\partial \bU_0(\bbeta)}{ 
 \partial \bbeta}|_{\bbeta = \bbeta_0^*}$ and the form of $\bSigma_U$ can be found in Theorem B.1 (supplementary material).
\end{theorem}


Additional technical assumptions (D2 - D5) are left in Appendix C (supplementary material). Our proof build on the convergence results for the underlying IPLW empirical process and we give the relevant results in Appendix C (supplementary material). 

{\color{black}
\begin{remark}
There are two approaches to estimate the covariance and construct the associated confidence interval using Theorem \ref{thm::thm_Cox_right_asn_d_asp}.
The first approach is via the bootstrap--we sample with replacement and re-estimate $\bgamma$ and $\bbeta$. 
The second approach is to estimate the asymptotic covariance matrix using the sandwich estimator. This involves estimating both $\bSigma_0$ and $\bSigma_U$.
We can estimate $\bSigma_0$ via $\hat \bSigma_0 = -\frac{\partial \hat{\bU}_n(\bbeta)}{ 
\partial \bbeta}|_{\bbeta = \hat \bbeta_n}$.
The covariance matrix $\bSigma_U$ has 
an explicit form as (see Theorem B.1 in the supplementary material)
\begin{align*}
\bSigma_U &= \Var[\bU(\bbeta_0^*)] + \E \left[\bU(\bbeta_0^*) \bU(\bbeta_0^*)^T \frac{I(Q = 0)[1 - \pi_{\bgamma_0}(\bZ_1)]}{\pi_{\bgamma_0}(\bZ_1)}
\right] \\
& - \bQ_e(\bbeta_0^*)^T \bSigma_{\bgamma_0}^{-1} \bQ_e(\bbeta_0^*),
\end{align*}
where $\bQ_e(\bbeta_0^*) = \E[ I(Q = 0) (1 - \pi_{\bgamma_0}(\bZ_1)) \tZ_1 \bU(\bbeta_0^*)^T]$
and 
$\bSigma_{\bgamma_0}= \E[I(Q = 0) \pi_{\bgamma_0}(\bX)(1 - \pi_{\bgamma_0}(\bX)) \tX \tX^T]$.
$\bSigma_{\bgamma_0}$ can be estimated by the covariance matrix of $\hat \bgamma_n$. For the rest quantities, through a similar proof as \thref{prop::prop_ipw_llk}, we can prove that 
$$
\Var[\bU(\bbeta_0^*)] + \E \left[\bU(\bbeta_0^*) \bU(\bbeta_0^*)^T \frac{I(Q = 0)[1 - \pi_{\bgamma_0}(\bZ_1)]}{\pi_{\bgamma_0}(\bZ_1)}
\right] = \Var\left[ \frac{I(L + Q > 0)}{Q + (1 - Q) \pi_{\bgamma_0}(\bZ_1)}\bU(\bbeta_0^*)\right]
$$
and 
$$
\bQ_e(\bbeta_0^*) = \E\left[ I(Q = 0) \frac{I(L = 1)}{\pi_{\bgamma_0}(\bZ_1)}(1 - \pi_{\bgamma_0}(\bZ_1)) \tZ_1 \bU(\bbeta_0^*)^T\right]
$$
Thus, each quantity above can be estimated empirically by replacing $\bbeta_0^*$ with $\hat \bbeta_n$ and $\bgamma_0$ with $\hat \bgamma_n$.
\end{remark}
}

\yingqi{\subsection{Augmented inverse probability of linkage weighted (AIPLW) approach}  \label{sec::doubly_robust}
We further developed the augmented inverse probability of linkage weighted (AIPLW) estimator for estimating the long-term effect. The augmented IPLW partial score could be written as  
\begin{equation}
\begin{aligned}
\hat{\bU}_{n, AIPLW}  &= \frac{1}{n} \sum_{i=1}^n \left\{
I(Q_i = 1) \Delta_i \left[ \bX_i - \frac{S_{n, DR}^{(1)}(\bbeta, \tilde{T}_i)}{S_{n, DR}^{(0)}(\bbeta, \tilde{T}_i)} \right] 
\right. \\
& \left. + I(Q_i = 0)\left[ \frac{I(L_i = 1)}{\pi_{\hat{\bgamma}_n}(\bX_i)} \Delta_i \left[ \bX_i - \frac{S_{n, DR}^{(1)}(\bbeta, \tilde{T}_i)}{S_{n, DR}^{(0)}(\bbeta, \tilde{T}_i)} \right]     + \left(1 - \frac{I(L_i =  1)}{\pi_{\hat{\bgamma}_n}(\bX_i)}\right) \hat{m}_0(\bX_i)\right]
\right\}.
\end{aligned}
\label{eqn::eqn_aipw_ee}
\end{equation}
The actual forms of $\hat{m}_0(\mathbf{x})$ and $S^{(k)}_{n, DR}(\bbeta, t)$ can be found in Appendix F (supplementary material). On a high level,
three regression functions $m_k(\mathbf{x})$ for $k = 1, 2, 3$ need to be estimated to obtain $\hat{m}_0(\mathbf{x})$ and $S^{(k)}_{n, DR}(\bbeta, t)$. We denote these regression function estimators as $\hat{m}_k(\mathbf{x})$.
}

\yingqi{
Let $\hat{\bbeta}_{AIPW}$ be the solution to $U_{n, AIPW} = \bm{0}$. $\hat{\bbeta}_{AIPW}$ is a doubly robust estimator in the sense that it is a consistent estimator for $\bbeta_0^*$ when either the linkage probability estimator $\pi_{\hat{\bgamma}_n}(\mathbf{x})$ is consistent for $\pi_0(\mathbf{x})$ or the regression function estimators $\hat{m}_k(\mathbf{x})$ are consistent for $m_k(\mathbf{x})$ for $k = 1, 2, 3$. 
However, one limitation is that we need to consistently estimate three regression functions. These three regression functions are themselves variational dependent and congenial parametric modeling can be difficult for all three functions. More details on the illustration for AIPLW estimator could be found in Appendix F. }

\section{Simulation}	\label{sec::simulation}
We now compare the performances of the proposed NLAC and IPLW methods with several other methods,  including complete-case analysis (CC), complete-case analysis plus (CC+) and the oracle method. 
We consider the following data generation settings. 
 The hazard function is $\lambda(t | \bX(t)) = \lambda_0(t) \exp(\beta_1 X_1 + \beta_2 X_2 + \beta_3 X_3(t) \times X_1)$ and $\bbeta_0 = (\beta_1, \beta_2, \beta_3) = (-\ln(4), \ln(1.5), 0.5)$. $X_1$ follows a Bernoulli distribution with probability 0.5, $X_2$ follows a normal distribution with mean and variance both being 1 and $X_3 = I(t \geq \tau_1)$ where $\tau_1$ is the end time of clinical trial.  If we treat $X_1$ as the variable for treatment assignment,  $\beta_3$ now represents the difference between the treatment effect after and before $\tau_1$.  This is the Cox model with a change point at $\tau_1$  \eqref{eqn::eqn_change}. The baseline hazard function is $\lambda_0(t) = 0.06$. $C_1$ is exponentially distributed with rate $0.01 X_1 + 0.03$ and the censoring time $C_2$ is set as $C_1$ plus an exponential random variable with rate $0.05 X_1 + 0.03$.  Further,  we set $\tau_1 = 5$ and $\tau_2 = 16$.  
 Three linkage mechanisms are considered as follows:
\begin{enumerate}
\item[(1)]  $P(L = 1) = 0.5$ and $L$ is independent of all other variables. This is the linking completely at random (LCAR) case. 


\item[(2)]  
$\log \left\{P(L = 1 | \bX, Q = 0) / P(L = 0 | \bX, Q = 0) \right\} = -0.25 + 0.5 X_1 + 0.5 X_2$
 and $P(L = 1 | Q = 1) = 0.5$. Thus the linkage satisfies CLAR assumption. 

\item[(3)]  
$\log \left\{ P(L = 1 | \bX, Q = 0, \tilde{T}, \Delta) / P(L = 0 | \bX, Q = 0, \tilde{T}, \Delta) \right\} = -0.25 + 0.5 X_1 + 0.5 X_2 - 0.01 \tilde{T} - 0.01 \Delta$
and $P(L = 1 | Q = 1) = 0.5$. This is the linking not at random (LNAR($\tilde{T}$)) case.
\end{enumerate}
\begin{table}[H]
\caption{Simulation results when Cox model is correctly specified when $n = 2000$. }
\begin{center}
\begin{tabular}{c c c c c c c c }
\hline
 & &   \multicolumn{3}{c}{Bias (Mean SE)} & \multicolumn{3}{c}{Coverage of 95\% CI}\\
 \hline
 Mechanism & Method              & $\beta_1$ & $\beta_2$ & $\beta_3$ & $\beta_1$ & $\beta_2$ & $\beta_3$ \\
 \hline
LCAR & Oracle & -0.00 (0.108) &  0.00 (0.033)  & -0.00  (0.140) & 0.94   & 0.95  & 0.96 \\
\hline
& CC  & -0.00 (0.153) &  0.00 (0.046)  & -0.01  (0.199) & 0.95   & 0.96  & 0.96  \\
\hline
& CC+ &  0.13 (0.110) & -0.03 (0.038)  & -0.13 ( 0.168)  & $0.77^{\dagger}$   & $0.86^{\dagger}$   & $0.89^{\dagger}$  \\
\hline
& NLAC & -0.00  (0.110) &  0.00 (0.038)  & -0.01 (0.168)  & 0.94   & 0.95   & 0.95  \\
\hline
& IPLW & -0.00   (0.114) & 0.00 (0.041)  & -0.01  (0.170)  & 0.94   & 0.95   & 0.95   \\
\hline
 \hline
CLAR & Oracle & -0.00 (0.108) &  0.00 (0.033)  & -0.00  (0.140) & 0.94   & 0.95  & 0.96 \\
\hline
& CC  & -0.18 (0.151) & -0.06 (0.045)   &  0.19  (0.188) & $0.79^{\dagger}$   & $0.75^{\dagger}$  & $0.83^{\dagger}$  \\
\hline
& CC+  & -0.07  (0.109) & -0.10 (0.038)  &  0.09  (0.157)  & 0.91   & $0.23^{\dagger}$  & 0.91  \\
\hline
& NLAC & -0.00  (0.109) &  0.00 (0.037)  & -0.00  (0.157)  & 0.94   & 0.95   & 0.95  \\
\hline
& IPLW & -0.00   (0.110) & 0.00 (0.040)  & -0.00  (0.159)  & 0.94   & 0.95   & 0.96 \\
\hline
\hline
LNAR ($\tilde{T}$) & Oracle & -0.00 (0.108) &  0.00 (0.033)  & -0.00  (0.140) & 0.94   & 0.95  & 0.96 \\
\hline
& CC  &  -0.19 (0.151) & -0.07 (0.045)   &  0.20  (0.189) & $0.79^{\dagger}$   & $0.69^{\dagger}$  & $0.84^{\dagger}$  \\
\hline
& CC+ &  0.07  (0.109) & -0.11 (0.038)  &  0.09  (0.158)  & 0.92   & $0.16^{\dagger}$  & 0.92  \\
\hline
& NLAC & -0.00  (0.109) &  0.00 (0.037)  & -0.01 (0.158)   & 0.94   & 0.95  & 0.95  \\
\hline
& IPLW & -0.00 (0.111) & -0.00 (0.041)  & -0.01  (0.160)  & 0.94   & 0.95   & 0.95\\
\hline
\hline
LNAR ($C_2$) & Oracle & -0.00 (0.108) &  0.00 (0.033)  & -0.00  (0.140) & 0.94   & 0.95  & 0.96 \\
\hline
& CC  & -0.31 (0.154)   & -0.16 (0.055) &  0.33  (0.229) & $0.47^{\dagger}$   & $0.16^{\dagger}$  & $0.70^{\dagger}$  \\
\hline
& CC+  & -0.12  (0.111)   & -0.22 (0.043) &  0.14  (0.202)  & $0.83^{\dagger}$   & $0.00^{\dagger}$  & $0.88^{\dagger}$  \\
\hline
& NLAC & -0.00  (0.110)   &  0.00 (0.042) &  -0.02  (0.202)  & 0.93   & 0.95  & 0.95  \\
\hline
& IPLW & -0.00   (0.117)   & 0.01 (0.057) & -0.02  (0.211)  & 0.93   & 0.94  & 0.94   \\
\hline
\end{tabular}
\end{center}
\footnotesize{We use $^\dagger$ to highlight settings with coverage below 90\%.}
\label{table::td_s1_lcar}
\end{table}

\begin{table}[H]
\caption{Simulation results for when Cox model is misspecified and $n = 2000$.}
\begin{center}
\begin{tabular}{c c c c c c  }
\hline
&  &  \multicolumn{2}{c}{Bias (Mean SE)} & \multicolumn{2}{c}{Coverage of 95\% CI}\\
 \hline
Mechanism &  Method & $\beta_1$ & $\beta_2$ & $\beta_1$ & $\beta_2$ \\
 \hline
LCAR & Oracle & -0.00 (0.068)   &  0.00 (0.033)  & 0.93   & 0.94 \\
\hline
& CC & -0.00 (0.096)   &  0.00 (0.046)  & 0.93   & 0.96  \\
\hline
& CC+ &  -0.00 (0.081)  & -0.03 (0.038)  & 0.93   & $0.84^{\dagger}$   \\
\hline
& NLAC & -0.08  (0.081)   &  0.00 (0.038)  & $0.82^{\dagger}$   & 0.95  \\
\hline
& IPLW &  -0.00  (0.087)   & 0.00 (0.042)  & 0.94   & 0.95   \\
\hline
 \hline
CLAR & Oracle & -0.00 (0.068)   &  0.00 (0.033)   & 0.93   & 0.94 \\
\hline
& CC & -0.04 (0.087)   &  -0.06 (0.045)  & 0.92   & $0.76^{\dagger}$    \\
\hline
& CC+ &  -0.08 (0.075)  &  -0.10 (0.038)  & $0.81^{\dagger}$   & $0.25^{\dagger}$   \\
\hline
& NLAC & -0.05  (0.076)   &  0.00 (0.037)  & $0.89^{\dagger}$   & 0.95   \\
\hline
& IPLW &  -0.00  (0.079)   & 0.00 (0.040)  & 0.93   & 0.95    \\
\hline
\hline
LNAR($\tilde{T}$) & Oracle & -0.00 (0.068)   &  0.00 (0.033)  & 0.93  & 0.94 \\
\hline
& CC & -0.05 (0.088)   & -0.07 (0.046)  & 0.91   & $0.70^{\dagger}$  \\
\hline
& CC+ &  -0.09 (0.076)  & -0.11 (0.038)  &  $0.78^{\dagger}$   & $0.17^{\dagger}$   \\
\hline
& NLAC &  -0.06 (0.077)   & -0.00 (0.038)  & $0.88^{\dagger}$   & 0.94  \\
\hline
& IPLW &  -0.00  (0.080)   & -0.00 (0.041)  & 0.93   & 0.95   \\
\hline
\end{tabular}
\end{center}
\footnotesize{We use $^\dagger$ to highlight settings with coverage below 90\%.}
\label{table::td_s2_lcar}
\end{table}

Both mechanisms (1) and (2) satisfy our CLAR assumption. Mechanism (2) shows that under the CLAR assumption,  data linkage can still depend on the survival outcomes through the in-trial censoring indicator $Q$. Mechanism (3) slightly violates the CLAR assumption and serves as a case for sensitivity analysis. 

We fit two Cox regressions. The first Cox regression is fitted with covariates $X_1, X_2, X_3(t) \times X_1$ and the second Cox regression is fitted with $X_1$ and $X_2$ only. Thus, Cox regression is correctly specified for the first regression and mis-specified for the second regression.  This suggests that $\bbeta_0^* = \bbeta_0$ for the first Cox regression.  For the misspecified case,  $\bbeta_0^*$ is estimated with the oracle method by computing the averages of 1,000 parameter estimates with sample size $n = 10,000$.  The mis-specified regression omits the time-dependent covariate $X_3(t) \times X_1$ and thus ignores the change-point at $\tau_1$ for the treatment effect.  The corresponding parameter 
$\beta_1^*$ for $X_1$ can be interpreted as an averaging effect for the entire follow-up period and is equal to $-1.10$, between the clinical trial treatment effect $-\ln(4) \approx -1.39$ and the observational follow-up treatment effect $-\ln(4) + 0.5 \approx -0.89$. 

We use the robust variance estimate \citep{lin1989robust} when Cox model is misspecified for all methods other than the IPLW method.  For the IPLW method,  the variance estimate is automatically robust when Cox model is misspecified. When Cox model is correctly specified,  one additional mechanism for linkage is considered as 
\begin{enumerate}
\item[(4)] 
$\log \left\{P(L = 1 | \bX, Q = 0, C_2, \Delta) / P(L = 0 | \bX, Q = 0, C_2, \Delta) \right\} = -0.25 + 0.5 X_1 + 0.5 X_2 - 0.1 C_2 - 0.1 \Delta$
 and $P(L = 1 | Q = 1)  = 0.5$.  We call this linkage mechanism LNAR($C_2$). 
\end{enumerate}
Mechanism (4) is a more serious violation of the CLAR assumption and linkage now depends on the unobserved censoring time $C_2$.
\gc{For all four linkage mechanisms, we use a logistic regression with covariates $X_1$ and $X_2$ to estimate the linkage probability $P(L = 1 | \bX, Q = 0)$.}
The percentages of samples that are not linked and censored in the clinical trial are approximately  39\%,  30\%, 32\%, 51\% for these four mechanisms.  

Simulation results are reported in Tables \ref{table::td_s1_lcar} - \ref{table::td_s2_lcar} under sample size $n = 2000$ and we generate 1,000 repetitions for each simulation setting.
In the table,  bias is the difference of the average of 1,000 parameter estimates and the true parameter value.  Mean standard error (SE) is the average of 1,000 SE estimates.  CI stands for confidence interval. We first discuss the results when Cox model is correctly specified.  When linkage satisfies LCAR,  all methods give consistent estimates and correct coverages for the 95\% confidence intervals except CC+. 
Oracle method gives the smallest variance estimates as each participant is linked.  NLAC gives the second smallest variance estimates.  CC can be viewed as an IPW method with known probability as the weights and it has the largest variance estimates among all methods.  The IPLW method gives smaller variance estimates than CC for two reasons.  First, IPLW method uses more data than CC;  second,  IPLW method uses estimated weights, which is known to be more efficient than IPW method with known probability as weights.  

When linkage satisfies CLAR but not LCAR, only oracle method, NLAC and IPLW method give consistent estimate. 
When CLAR is slightly violated,  CC, CC+ all obtain severely biased estimates and confidence intervals with less than nominal coverages.  NLAC and IPLW methods still perform relatively well in this case.  When linkage depends on the censoring time in observational follow-up $C_2$,  NLAC gives consistent estimates and correct coverage as expected.  For this particular simulation setting,  the proposed IPLW method also works pretty well. 

Next,  we discuss the simulation results when Cox model is misspecified.  When linkage satisfies LCAR,  CC+ and NLAC give inconsistent estimates of the parameters and do not achieve nominal coverage for 95\% confidence intervals.  It is expected that CC would perform well in this case as discussed in Section \ref{sec::motivating_example}.   The IPLW method again obtains smaller variance estimates than CC as more data are fitted and estimated weights improve efficiency.   When the linkage satisfies CLAR but not LCAR,  only oracle method and IPLW method give consistent estimates and correct coverages.  When CLAR is slightly violated,  IPLW approach performs the best among all methods other than the oracle method. 

\yingqi{We also implement the AIPLW estimator and the simulation results are included in Appendix H.3 and H.4 in the supplementary materials. The results showed that the AIPLW estimator provides  unbiased estimates and confidence intervals with nominal coverages when the linkage probability is consistently estimated. When the model for the linkage probability is incorrect, it exhibits competitive performance,    even if some of the regression functions could be slightly misspecified.}

\section{The linked PCPT-Medicare claims study}	\label{sec::SWOG}
We apply the proposed NLAC and IPLW method to the SWOG study that links Medicare claims data to the PCPT data \citep{unger2018using}.  
The PCPT randomly assigned 18,882 eligible men from 1993 to 1997 to finasteride or placebo daily for seven years. 
PCPT clinical records are linked to participants' Medicare claims data according to common social security number, sex and date of birth. 
Medicare claims are available from 1999 to 2011.  
The linkage enables PC and death to be identified by both clinical records and Medicare claims.  The median time from treatment assignment to the end of the linked dataset was 16 years. \gc{We are interested in studying the effect of treatment finasteride on the outcome, which is defined as the time to diagnosis of PC or death, whichever comes first.}

\gc{ We use an `intention to treat' approach in the primary evaluation, where results are analyzed based on participants' initial treatment groups, even if someone did not stick to their assigned treatments. Therefore, it evaluates the effectiveness of  the treatment, recognizing that non-compliance is a manifestation of the patient experience with the treatment, and offers insights into their real-world experiences with the treatment. In PCPT, 16.5\% of the participants were non-compliant, and the most common reasons were side effect \citep{pinsky2021adherence}. We recognize that a data linkage of this sort does not resolve issues of non-compliance with an assigned treatment. However, linking our PCPT trial data to Medicare data maintains the advantages of randomization and `intention to treat' approach allows for an unbiased evaluation of the intervention's long-term effectiveness.  }  

\yingqi{Out of the total 18,882 participants, 2,168 had a gap between the end of the SWOG trial and the start of Medicare claims, with a median gap length of 1.6 years.   We excluded these participants from the analysis.   Additionally, participants with missing covariates were removed, leaving 16,518 participants (finasteride = 8,253; placebo = 8,265) in the final analysis.  Of these, 11,858 (71.8\%) participants were linked to Medicare claims, falling under Case 1. Among this group, 10,041 participants were censored during the trial but were later linked to Medicare claims, demonstrating that the linkage provided additional follow-up information for the majority of participants. Another 411 participants (2.5\%) experienced the event of interest during the clinical trial but were not linked to observational follow-up (Case 2). Finally, 4,249 participants (25.7\%) were censored during the trial and did not have a subsequent link to Medicare claims (Case 3). }

We fit Cox regression with covariates including the prostate-specific antigen (PSA) level at study entry,  race,  body mass index at study entry,  first degree family history of prostate cancer, age at baseline.  Additional covariates were included for logistic regression modeling linkage:  participants' education level, marital status, employment status, type of jobs. 

\gc{The PCPT was designed to identify prostate cancer (PC) in several ways:1) biopsies triggered by concerns from annual digital rectal exams suggestive of PC; 2) biopsies due to elevated PSA levels; 3) end-of-study biopsies in the 7th year, which were recommended for all participants who had not been previously diagnosed with PC. As a result, a significant number of PC cases were identified in the 7-year biopsy window (6.5 - 7.5 years), leading to a noticeable increase in the hazard rate during this period  \citep{unger2018using}. Therefore,Cox regressions with two change points at 6.5 and 7.5 years are fitted to accommodate the design.}  We compared the results of CC,  CC+, NLAC and IPLW method in Table \ref{table::real_data_results}.  Table \ref{table::real_data_results} contains the parameter estimates and 95\% confidence intervals for treatment Finasteride in different time periods\footnote{on the exponential level}.  Overall, the results do not differ too much between all four methods based on the 95\% confidence intervals. 
A key reason might be that the linkage rate was high for the original study \citep{unger2018using} as 75\% of the participants were linked. Further, they examined potential health care utilization differences by arm and other potential biases in \cite{unger2018using} and found no evidence of strong differences. 
This suggests that linkages might be following a LCAR mechanism. 
We also obtained robust variance estimates \citep{lin1989robust} and the corresponding confidence intervals. The results are very similar to the nonrobust ones. In summary, finasteride arm participants had a 30\% decrease in the hazard ratio of prostate cancer (hazard ratio (HR) = 0.75, 95\% confidence intervals (CI) = 0.66 - 0.85) during the first 6.5 years. The effect of finasteride is strongest between 6.5 - 7.5 years (HR = 0.70, 95\% CI = 0.63 - 0.78). The long-term effect of finasteride after the 7.5 years does not seem to increase the risk of PC (HR = 1.06, 95\% CI = 0.99 - 1.14). 
We further fit a Cox regression without any change points and the results are in shown in the supplementary material.  The results again does not differ too much between all four methods. In summary, the long-term effect of finasteride, now estimating the averaging effect over the entire follow-up period, is still beneficial (HR = 0.90, 95\% CI = 0.85 - 0.95). 



\begin{table}[H]
  \caption{SWOG study long-term effect estimation with two change points}
  \begin{center}
  \begin{tabular}{ c c c c }
  \hline
  \hline
   Methods   &  Finasteride (0 - 6.5 years)  & Finasteride (6.5 - 7.5 years) &  Finasteride (7.5 years + )\\
   \hline
  IPLW & \gc{0.747 (0.661 - 0.845)} & \gc{0.702 (0.632 - 0.779)} & \gc{1.064 (0.990 - 1.144)} \\
  CC & \gc{0.741 (0.648 - 0.848)} & \gc{0.700 (0.624 - 0.784)} & \gc{1.047 (0.976 - 1.122)} \\
  CC+ & \gc{0.748 (0.661 - 0.845)} & \gc{0.696 (0.628 - 0.772)} & \gc{1.045 (0.975 - 1.120)} \\
  NLAC & \gc{0.743 (0.657 - 0.840)} & \gc{0.696 (0.627 - 0.773)} & \gc{1.037 (0.967 - 1.113)} \\
  \hline
  \hline
  \end{tabular}
  \end{center}
  \label{table::real_data_results}
\end{table}



\section{Discussion}	\label{sec::discussion}

In this paper,  we consider the problem of long-term effect estimation by fitting a Cox model to a partially linked dataset. We propose three methods, NLAC, IPLW and AIPLW, to estimate the underlying long term effect. In summary, NLAC is a natural and easy to implement method that works well when Cox model is correctly specified and linkage satisfies a conditional independent assumption. On the other hand, IPLW and AIPLW are more robust to model mis-specification for Cox model and relies on a different linkage assumption, the CLAR assumption. \yingqi{ When confidence exists in a correctly specified Cox model and the conditional independent linkage assumption, we recommend using the NLAC method due to its superior efficiency. However, if there is evidence suggesting that the Cox model may be violated, we recommend to consider IPLW or AIPLW methods.
AIPLW has gained popularity with its robustness and efficiency, and our numerical results demonstrate the robustness of AIPLW under model misspecification. However,  \cite{kang2007demystifying} warn against the use of doubly robust estimators when both working models are misspecified. Further discussions highlight that the choice of nuisance parameter estimates can significantly influence the properties of the doubly robust estimators when at least one working model is misspecified \citep{tan2007comment, tsiatis2007comment, robins2007comment}. In our case, constructing the AIPLW estimator requires specifying multiple outcome regression functions. The optimal approach to estimating these nuisance functions remains an open question and is an area for future research.} 
Additionally, we point out a few other potential future directions for this work:

\begin{itemize}

\item {\bf Interval Censoring.}
We have made the ``no gap'' assumption in the current paper to focus on the right censoring problem for simplicity.  However, in practice, it is possible that there might be gaps between the clinical trial and observational follow-up dataset.  Thus,  to fully deal with the problem,  we need to extend our current procedure to the interval censoring case as mentioned in Appendix G (supplementary material).  \cite{saegusa2013weighted} has studied the problem of two-phase sampling for Cox models under interval censoring.  In \cite{zeng2016maximum}, they formulate potentially time-dependent covariates effects on the interval-censored failure time through a broad class of semiparametric transformation models. These models include Cox  model as a special case. These models offer potential utility in addressing gap cases within our application. We defer the extension of the proposed approach for interval censored data to future work.  

\item \gc{{\bf Handling change points.} We believe due to the study design, it is reasonable to assume change points are known in our SWOG data application. However, as reviewer points out, in many studies,  the choice of change points may not be obvious. We can consider estimating the change-point location \citep{gijbels2003estimation}.  Our proposed method can  be extended to accommodate the estimated change point and we leave this as a future work.}

\item \yingqi{{\bf Competing risks.} With the prolonged follow-up, competing risks could occur to prevent the occurrence of the primary event of interest. For example, in the SWOG-Medicare linkage example, if the primary event or outcome of interest is time to diagnosis of prostate cancer, or other non-cancer outcomes such as ischemic and thrombotic events, we should analyze the cumulative incidence of events accounting for potential competing risk of death.  For competing risk data,   \cite{fine1999proportional} developed a proportional subdistribution hazards model to directly assess the effect of a covariate on the cumulative incidence function. We plan to extend the Fine and Gray model to the trial-linkage studies for evaluating long-term effects, which we will address in future work.}

\end{itemize}

%

\if1\blind{
\bigskip
\begin{center}
{\large\bf Acknowledgements}
\end{center}
\gc{We extend our sincere gratitude to the Associate Editor and the three anonymous reviewers for their insightful comments, which have significantly contributed to the improvement of the paper. 
We would like to thank Yikun Zhang for helping with running the simulations for AIPLW experiments and we express our appreciation to  Michael L Leblanc and Catherine M.Tangen for helpful comments.
} 
} \fi

\bigskip
\begin{center}
{\large\bf SUPPLEMENTARY MATERIAL}
\end{center}

\gc{The supplementary material contains the following appendices. 
Appendix A contains the proof for theorem 3.1 and discusses linkage assumptions for NLAC. Appendix B contains the derivation of IPLW partial score. Appendix C and D contains the results for IPLW empirical process theory and the proof for technical results. Appendix E contains more discussions of linkage assumptions. Appendix F discusses the doubly robust AIPLW estimator. Appendix G discusses the relaxation of the ``no gap'' assumption. Appendix H contains more numerical results. Appendix I includes a discussion on extension of the proposed metods to alternative models. }

\bibliographystyle{apalike}
\bibliography{gang}

\clearpage
\let\bibliographystyle\combinedbibliographystyle
\let\bibliography\combinedbibliography

\let\combinedoriginallabel\label
\let\combinedoriginalref\ref
\newcommand{\combinedmaplabel}[1]{%
  \ifnum\pdfstrcmp{#1}{sec::doubly_robust}=0 supp:#1%
  \else\ifnum\pdfstrcmp{#1}{eqn::eqn_aipw_ee}=0 supp:#1%
  \else\ifnum\pdfstrcmp{#1}{table::td_s1_lcar}=0 supp:#1%
  \else\ifnum\pdfstrcmp{#1}{table::td_s2_lcar}=0 supp:#1%
  \else#1%
  \fi\fi\fi\fi
}
\renewcommand{\label}[1]{%
  \edef\combinedmappedlabel{\combinedmaplabel{#1}}%
  \expandafter\combinedoriginallabel\expandafter{\combinedmappedlabel}%
}
\renewcommand{\ref}[1]{%
  \edef\combinedmappedlabel{\combinedmaplabel{#1}}%
  \expandafter\combinedoriginalref\expandafter{\combinedmappedlabel}%
}
\makeatletter
\let\ltx@label\label
\let\caption@x@label\label
\makeatother

\makeatletter
\newcommand{\title}[1]{\gdef\@title{#1}}
\newcommand{\author}[1]{\gdef\@author{#1}}
\newcommand{\maketitle}{%
  \begin{center}
    {\LARGE \@title\par}
    \vskip 1.5em
    {\large \@author\par}
  \end{center}
  \vskip 1.5em
}
\makeatother

\def\spacingset#1{\renewcommand{\baselinestretch}%
{#1}\small\normalsize} 


\if1\blind{
  \title{Supplementary Material}
  \author{
  Gang Cheng, Yen-Chi Chen, Joseph M. Unger, Cathee Till, Ying-Qi Zhao
  }
  \maketitle
} \fi

\if0\blind{
  \begin{center} 
    {\LARGE\bf Supplementary material}
\end{center}
  \medskip
} \fi

\appendix
\spacingset{1.5}

\section{Linkage assumptions for NLAC and proof of Theorem 3.1} \label{appendix::sec_naive_iplw_comparison}
We first present the proof for theorem 3.1.  Recall that for NLAC, the censoring time can be viewed as
$$
C_{\text{NLAC}} = L \max(C_1, C_2) + (1 - L) C_1
= L[\max(C_1, C_2) - C_1] + C_1
$$
\begin{proof} [ of Theorem 3.1
]
We want to prove that $C_{\text{NLAC}} \indep T | \bX$. This independent censoring assumption holds if \text{\bf (A1)}
 holds since \text{\bf (A1)} and $(C_1, C_2) \indep T | \bX$ implies $(L, C_1, C_2) \indep T |  \bX$, which further implies that $C_{\text{NLAC}} \indep T | \bX$. Therefore, \text{\bf (A1)} is a sufficient linkage assumption for NLAC.
\end{proof}

We have presented ${\bf (A1)}$, a sufficient linkage assumption for NLAC. Here, we show that NLAC also works under an assumption that is sufficient for a modified CLAR assumption. For notational simplicity,  we focus on the case with time-independent covariates.

Note that we can modify the CLAR assumption by further conditioning on $C_1$ as 
\[
\text{\bf (A3)} \ L \indep (\tilde{T}, \Delta) | \bX, Q = 0, C_1
\]
given $C_1$ is always observed when $Q = 0$. $\bf (A1)$ actually implies that 
\[
 \ L \indep (\tilde{T}, \Delta) | \bX, C_2, C_1
\]
as $\tilde{T}$ and $\Delta$ are functions of $T, C_1, C_2$. Thus, the main difference between ${\bf (A3)}$ and ${\bf (A1)}$ can be roughly viewed as replacing $C_2$ by $Q = 0$. 

One sufficient assumption for $\text{\bf (A3)}$ is 
$$
\text{\bf (A4)} \ L \indep (T, C_2) | \bX, Q = 0, C_1.
$$
Further, we consider the following assumption:
$$
\text{\bf (A5)} \ L \indep T | (\bX, Q = 0, C_1, C_2).
$$
By the weak union property of conditional independence, $\text{\bf (A4)} \Rightarrow \text{\bf (A5)}$ and 
we now show that NLAC also works under ${\bf (A5)}$.  
\begin{proposition} \thlabel{prop::prop_naive_method_linkage_assumption_clar}
When Cox model is correctly specified and $(C_1, C_2) \indep T | \bX$,  if the following assumption
$$
\text{\bf (N4)} \ L \indep T | (\bX, Q = 0, C_1, C_2)
$$ 
holds, NLAC provides consistent estimates for $\bbeta_0$. 
\end{proposition}

On the other hand,  using the contraction property of conditional independence,  we further have 
$$
\text{\bf (A5)} +  L \indep C_2 | \bX,  Q = 0, C_1\Rightarrow \text{\bf (A4)}
$$
Thus, to conclude, NLAC works under a slightly weaker condition than $\text{\bf (A4)}$ in that the linkage can further depends on the potentially missing $C_2$. However, this does not mean that NLAC works under a weaker assumption than CLAR.

\begin{proof} [ of Proposition \ref{prop::prop_naive_method_linkage_assumption_clar}]
The NLAC method modifies the censoring time $C$ compared to the oracle method where  all participants in the
clinical trial are linked. To prove the consistency of the estimator obtained by the NLAC method,  we only need to prove that the population version of the partial likelihood for the NLAC method has a solution at $\bbeta = \bbeta_0$.  The rest is the same as the consistency proof in \cite{andersen1982Cox}.   For notational simplicity,  we illustrate the proof with time-independent covariates only. 

Recall for NLAC method, we have 
$$
\tilde{T}_j = \begin{cases} \tilde{T}_j & L_j + Q_j > 0 \\
C_{1j} & L_j + Q_j = 0 \end{cases}
$$
and $\Delta_{j} = 0$ if $L_j + Q_j = 0$.  Thus, the partial likelihood for NLAC method solves 
\begin{align*}
\hat{U}_{naive}(\bbeta) = \frac{1}{n} \sum_{i=1}^n I(L_i + Q_i > 0) \Delta_i \left[ \bX_i - \frac{S_{n, w}^{(1)}(\bbeta, \tilde{T}_i)}{S_{n, w}^{(0)}(\bbeta, \tilde{T}_i)} \right] = \bm{0}
\end{align*}
with 
$$
S_{n, w}^{(k)}(\bbeta, t) = \sum_{j=1}^n \left[ I(L_j + Q_j > 0) I(\tilde{T}_j \geq t) \exp(\bX_j^T \bbeta)  + I(L_j + Q_j = 0) I(C_{1j} \geq t) \exp(\bX_j^T \bbeta) 
\right] \bX_j^{\otimes k}
$$
Then, by similar technique in \cite{andersen1982Cox}, we can prove that 
$$
\hat{U}_{naive}(\bbeta) \rightarrow_p U_0(\bbeta) = \E \left[ \Delta I(L + Q > 0) \left[ 
\bX - \frac{S^{(1)}(\bbeta, \tilde{T})}{S^{(0)}(\bbeta, \tilde{T})}
\right]
\right]
$$
with $S^{(k)}(\bbeta, t) = \E\left[ \left\{I(L + Q > 0)I(\tilde{T} \geq t)\exp(\bX^T \bbeta) + I(L = 0) I(Q = 0) I(C_1 \geq t) \exp(\bX^T \bbeta) \right\} \bX^{\otimes k} \right]$. Next, we prove that $\bbeta = \bbeta_0$ solves $U_0(\bbeta) = \bm{0}$.
We first have that
\begin{align*}
\E[ \Delta I(L + Q > 0) \bX] = \E[ \Delta I(Q = 1) \bX] + \E[\Delta I(L = 1) I(Q = 0) \bX]
\end{align*}
Further,  we have 
\begin{align*}
\E[ \Delta I(Q = 1) \bX] & = \E[I(T \leq C) I(T \leq C_1) \bX] = E[I(T \leq C_1) \bX] = \int_0^{\tau_1} \E[P(C_1 \geq t | \bX) f_T(t | \bX)] dt \\
& = \int^{\tau_1}_0 \E[ \bX I(C_1 \geq t) I(T \geq t) \exp(\bX^T \bbeta_0)] \lambda_0(t) dt
\end{align*}
since $\lambda_T(t | \bX) = \lambda_0(t) \exp(\bX^T \bbeta_0)$ and $C_1 \indep T | \bX$. Next, we have 
\begin{align*}
\E[\Delta I(L = 1) I(Q = 0) \bX] & = \E[ I(T \leq C) I(T \geq C_1) I(L = 1) \bX] \\
& = \E[I(T \leq C_2) I(T \geq C_1)I(L = 1) \bX]  \\
& = \E[\bX \pi(\bX, C_1, C_2) I(T \geq C_1) I(T \leq C_2)]
\stepcounter{equation}\tag{\theequation}\label{myeq1}
 \\
& = \E[ \bX \E[ \E[ I(T \geq C_1) I(T \leq C_2) \pi(\bX, C_1, C_2) | T, \bX] | \bX]]
\end{align*}
where $\E[ I(L = 1) | T, Q = 0, C_1, C_2, \bX] = \pi(\bX, C_1, C_2)$. Next, denote $g(T, \bX) = \E[ I(T \geq C_1) I(T \leq C_2) \pi(\bX, C_1, C_2) | T, \bX]$, we further have 
\begin{align*}
g(s, \bX) = \E [I(T \geq C_1) I(T \leq C_2) \pi(X, C_1, C_2) | T = s, \bX] = \E[ I(C_1 \leq s) I(C_2 \geq s) \pi(X, C_1, C_2) | \bX]
\end{align*}
since $(C_1, C_2) \indep T | \bX$. Further, we have
\begin{align*}
& \E[ \Delta I(L = 1) I(Q = 0) \bX]  = \E[\bX \E[ g(T, \bX) | \bX]] = \E\left[ \bX \int^{\tau_2}_0 g(s, \bX) f_T(s | \bX) ds \right] \\
& = \int^{\tau_2}_0 \E\left[ \bX \E[I(C_1 \leq s) I(C_2 \geq s) \pi(\bX, C_1, C_2) | \bX]  f_T(s | \bX) \right] ds \\
& = \int^{\tau_2}_0 \E\left[ \bX I(C_1 \leq s) I(C_2 \geq s) \pi(\bX, C_1, C_2)  P(T \geq s | \bX) \exp(\bX^T \bbeta_0) \right] \lambda_0(s) ds \\
& = \int^{\tau_2}_0 \E\left[ \bX I(C_1 \leq s) I(C_2 \geq s) \pi(\bX, C_1, C_2)  \E[I(T \geq s) | \bX] \exp(\bX^T \bbeta_0) \right] \lambda_0(s) ds \\
& = \int^{\tau_2}_0 \E\left[ \bX \E[I(C_1 \leq s) I(C_2 \geq s) I(T \geq s) \pi(\bX, C_1, C_2) | \bX] \exp(\bX^T \bbeta_0) \right] \lambda_0(s) ds \\
& = \int^{\tau_2}_0 \E \left[ \bX I(C_1 \leq s) I(C_2 \geq s) I(T \geq s) \pi(\bX, C_1, C_2) \exp(\bX^T \bbeta_0) \right] \lambda_0(s) ds
\end{align*}
Thus, together, we have
\begin{align*}
& \E[\Delta I(L + Q > 0) \bX] = \\
& \int^{\tau_1}_0 \E[ \bX \left\{ I(C_1 \geq s) I(T \geq s) + \pi(\bX, C_1, C_2) I(C_1 \leq s) I(C_2 \geq s) I(T \geq s) \right\}  \exp(\bX^T \bbeta_0) ] \lambda_0(s) ds \\ 
& + \int^{\tau_2}_{\tau_1} \E[ \bX \pi(\bX, C_1, C_2) I(C_2 \geq s) I(T \geq s) \exp(\bX^T \bbeta_0)] \lambda_0(s) ds
\end{align*}
as $C_1 \in [0, \tau_1]$.  Similarly, we have
\begin{align*}
& \E \left[ \Delta I(L + Q > 0) \frac{S^{(1)}(\bbeta, T)}{S^{(0)}(\bbeta, T)} \right] = \\
&  \int^{\tau_1}_0 \E[ \left\{ I(C_1 \geq s) I(T \geq s) + \pi(\bX, C_1, C_2) I(C_1 \leq s) I(C_2 \geq s) I(T \geq s) \right\}  \exp(\bX^T \bbeta_0) ] \lambda_0(s) \frac{S^{(1)}(\bbeta, s)}{S^{(0)}(\bbeta, s)}ds \\
& + \int^{\tau_2}_{\tau_1} \E[ \pi(\bX, C_1, C_2) I(C_2 \geq s) I(T \geq s) \exp(\bX^T \bbeta_0)] \lambda_0(s) \frac{S^{(1)}(\bbeta, s)}{S^{(0)}(\bbeta, s)} ds
\end{align*}
Now as long as we can prove that 
\begin{align*}
\E[\Delta I(L + Q > 0) \bX] - \E \left[ \Delta I(L + Q > 0) \frac{S^{(1)}(\bbeta_0, T)}{S^{(0)}(\bbeta_0, T)} \right] = \bm{0}
\end{align*}
Then we are done.  We prove this by proving the following two equalities:
\begin{equation}
\begin{aligned}
& \int^{\tau_1}_0 \E[ \bX \left\{ I(C_1 \geq s) I(T \geq s) + \pi(\bX, C_1, C_2) I(C_1 \leq s) I(C_2 \geq s) I(T \geq s) \right\}  \exp(\bX^T \bbeta_0) ] \lambda_0(s) ds - \\
&  \int^{\tau_1}_0 \E[ \left\{ I(C_1 \geq s) I(T \geq s) + \pi(\bX, C_1, C_2) I(C_1 \leq s) I(C_2 \geq s) I(T \geq s) \right\}  \exp(\bX^T \bbeta_0) ] \lambda_0(s) \frac{S^{(1)}(\bbeta, s)}{S^{(0)}(\bbeta, s)}ds = \bm{0}
\end{aligned}
\label{eqn::naive_population_1}
\end{equation}
and 
\begin{equation}
\begin{aligned}
& \int^{\tau_2}_{\tau_1} \E[ \bX \pi(\bX, C_1, C_2) I(C_2 \geq s) I(T \geq s) \exp(\bX^T \bbeta_0)] \lambda_0(s) ds - \\
& \int^{\tau_2}_{\tau_1} \E[ \pi(\bX, C_1, C_2) I(C_2 \geq s) I(T \geq s) \exp(\bX^T \bbeta_0)] \lambda_0(s) \frac{S^{(1)}(\bbeta, s)}{S^{(0)}(\bbeta, s)} ds = \bm{0}
\end{aligned}
\label{eqn::naive_population_2}
\end{equation}
We first prove equation \eqref{eqn::naive_population_1}.  For $s \in [0, \tau_1]$, we have 
\begin{align*}
S^{(0)}(\bbeta, s) & = \E[ I(\tilde{T} \geq s) \exp(\bX^T \bbeta) [I(Q = 1) + I(L  = 1)I(Q = 0)]  \\
& + I(C_1 \geq s) I(L = 0) I( Q = 0) \exp(\bX^T \bbeta)] \\
& = \E[ I(\tilde{T} \geq s) \exp(\bX^T \bbeta) I(Q = 1) + I(C_1 \geq s) \exp(\bX^T \bbeta) I(Q = 0)]  \\
& + \E[ I(L = 1) I(Q = 0) \exp(\bX^T \bbeta) [I(\tilde{T} \geq s) - I(C_1 \geq s)]]
\end{align*}
Further, 
\begin{align*}
& \E[ I(\tilde{T} \geq s) \exp(\bX^T \bbeta) I(Q = 1) + I(C_1 \geq s) \exp(\bX^T \bbeta) I(Q = 0)] \\
& = \E[ I(T \geq s) I(T \leq C_1) \exp(\bX^T \bbeta)] + \E[I(C_1 \geq s) I(T \geq C_1) \exp(\bX^T \bbeta)] \\
& = \E[ I(T \geq s) I(C_1 \geq s) \exp(\bX^T \bbeta)]
\end{align*}
and
\begin{align*}
& \E[ I(L = 1) I(Q = 0) \exp(\bX^T \bbeta) [I(\tilde{T} \geq s) - I(C_1 \geq s)]] \\
& = \E[ I(L = 1) I(Q = 0) \exp(\bX^T \bbeta) I(T \geq s) I(C_2 \geq s) I(C_1 \leq s)] \\
& = \E[ \pi(\bX, C_1, C_2) I(T \geq s) I(C_2 \geq s) I(C_1 \leq s) \exp(\bX^T \bbeta)] 
\stepcounter{equation}\tag{\theequation}\label{myeq2}
\end{align*}
Similarly,  we can prove that 
\begin{align*}
S^{(1)}(\bbeta, s) = \E[ \bX \left\{ I(C_1 \geq s) I(T \geq s) + \pi(\bX, C_1, C_2) I(C_1 \leq s) I(C_2 \geq s) I(T \geq s) \right\}  \exp(\bX^T \bbeta) ]
\end{align*}
All these results suggest that when $\bbeta = \bbeta_0$, equation \eqref{eqn::naive_population_1} is $\bm{0}$.  When $s \in [\tau_1, \tau_2]$, we have 
\begin{align*}
S^{(0)}(\bbeta, s) & = \E[ I(\tilde{T} \geq s) \exp(\bX^T \bbeta) I(L = 1) I(Q = 0)] \\
& = \E[ I(T \geq s) I(C_2 \geq s) \exp(\bX^T \bbeta) I(T \geq C_1) \pi(\bX, C_1, C_2)] \\
& = \E[ \pi(\bX, C_1, C_2) I(T \geq s) I(C_2 \geq s) \exp(\bX^T \bbeta)] 
\end{align*}
as $C_1 \in [0, \tau_1]$. By the same idea,  we can prove that equation \eqref{eqn::naive_population_2} is $\bm{0}$.  Then we have finished the proof. 
\end{proof} 
Note that we only uses assumption $\text{\bf (A5)}$ in \eqref{myeq1} and \eqref{myeq2} and in fact, we can further assume that
$$
\text{\bf (A6)} \  L \indep T | (\bX, Q = 0, C_1, C_2, \Delta)
$$
and NLAC method still gives consistent estimate in this scenario.  Finally for the case of time-dependent covariates,  we can get similar results under the assumption that
$$
P(T \geq s | C_1, C_2, \bar{\bX}(\tau_M)) = P(T \geq s  | \bar{\bX}(s)) \qquad P(C_k \geq s | T, \bar{\bX}(\tau_M)) = P(C_k \geq s | \bar{\bX}(s)), k = 1, 2.
$$
The proof is overall very similar to the case when there are only time-independent covariates and we omit it. 

\section{Derivation of IPLW partial score} \label{sec::iplw_partial_score}
Assume that there are $n_{00}$ observations with $L = 0$ and $Q = 0$ and define weight $\hat{w}_{i} = I(L_{i} + Q_{i} > 0) / [Q_{i} + (1 - Q_i) \pi_{\hat{\bgamma}_n}(\bX_{i})]$.  Now we are ready to derive the IPLW partial likelihood for estimating $\bbeta_0^*$.  Assume that $\tilde{T}_{(1)} < \tilde{T}_{(2)} < \cdots < \tilde{T}_{(n - n_{00})}$ are the ordered $\tilde{T}_i$'s, and $\bX_{(i)}, L_{(i)}, Q_{(i)}, \Delta_{(i)}$ are the corresponding covariates, linkage indicator,  in-trial censoring indicator and censoring indicator.  Denote $\Lambda_0(t)$ as the cumulative hazard function.  Let $h_i = d\Lambda_0(\tilde{T}_{(i)}) = \Lambda_0(\tilde{T}_{(i)}) - \Lambda_0(\tilde{T}_{(i)}^-)$ and $\Lambda_0(\tilde{T}_{(i)}) = \sum_{j \leq i} h_j$. It is known \citep{van2000asymptotic} that maximization with respect to $\lambda_0$ can be done by maximizing $\bm{h} = (h_1, \ldots, h_{n - n_{00}})$ and thus we only need to consider the case where $l_n(\bbeta,\lambda_0)  = l_n(\bbeta,\bf{h}) $. As a result, the IPLW log-likelihood can be rewritten as:
\begin{align*}
l_n(\bbeta,\lambda_0)  = l_n(\bbeta,\bf{h}) 
&= \frac{1}{n} \sum_{i=1}^n \hat{w}_{(i)} \left[\Delta_{(i)} \log h_i + \Delta_{(i)} \bX_{(i)}^T \bbeta - \exp(\bX_{(i)}^T \bbeta) \sum_{j \leq i} h_j \right].
\end{align*}
Thus, maximizing this with respect to $\bm{h} = (h_1, \ldots, h_{n - n_{00}})$ leads to
$$
\hat{h}_i = \frac{\Delta_{(i)} \hat{w}_{(i)}}{\sum_{j \geq i} \hat{w}_{(j)} \exp\left(\bX_{(j)}^T \bbeta\right)}
$$
Take $\hat {\bm{h}}=(\hat{h}_1,\cdots, \hat h_{n-n_{00}})$ back to the empirical log-likelihood, we obtain the IPLW partial log-likelihood
\begin{align*}
\mathcal{L}_n(\bbeta) & = l_n(\bbeta,\hat {\bm{h}}) = \frac{1}{n}\sum_{i=1}^{n}  \Delta_{(i)} \hat{w}_{(i)} \left(\bX_{(i)}^T \bbeta - \log\left(\sum_{j \geq i} \hat{w}_{(j)} \exp(\bX_{(j)}^T \bbeta)\right) \right)
\end{align*}
which can be further simplified as 
\begin{align*}
\mathcal{L}_n(\bbeta) = \frac{1}{n} \sum_{i=1}^n \Delta_i \hat{w}_i\left(\bX_i^T \bbeta - \log \left(\sum_{j = 1}^n I(\tilde{T}_j \geq \tilde{T}_i) \hat{w}_j \exp(\bX_j^T \bbeta)\right) \right)
\end{align*}
and then the partial score is
\begin{equation}
  \hat{\bU}_n(\bbeta) = \frac{1}{n} \sum_{i=1}^n \Delta_i \hat{w}_i \left(\bX_i - \frac{\sum_{j = 1}^{n} \hat{w}_j  I(\tilde{T}_j \geq \tilde{T}_i) \exp(\bX_j^T \bbeta) \bX_j }{\sum_{j = 1}^{n} \hat{w}_j I(\tilde{T}_j \geq \tilde{T}_i) \exp(\bX_j^T \bbeta)} \right),
  \end{equation}

  To derive the asymptotic distribution of $\hat{\bbeta}_n$,  note that 
  \begin{align*}
  \hat{\bU}_n(\bbeta) - \hat{\bU}_n(\bbeta_0^*) = \left. \frac{\partial \hat{\bU}_n(\bbeta)}{\partial \bbeta} \right \rvert_{\bbeta = \bbeta_0^*} (\bbeta - \bbeta_0^*)+ o_P(\|\bbeta - \bbeta_0^*\|)
  \end{align*}
  Thus, choosing $\bbeta = \hat{\bbeta}_n$ leads to
  $$
  \sqrt{n}(\hat{\bbeta}_n - \bbeta_0^*) \approx \left(-\left. \frac{\partial \hat{\bU}_n(\bbeta)}{\partial \bbeta} \right \rvert_{\bbeta = \bbeta_0^*}\right)^{-1} n^{1/2} \hat{\bU}_n(\bbeta_0^*)
  $$
  As a result,  we just need to prove that $- \left. \frac{\partial \hat{\bU}_n(\bbeta)}{\partial \bbeta} \right \rvert_{\bbeta = \bbeta_0^*} \rightarrow_p -\left. \frac{\partial \bU_0(\bbeta)}{\partial \bbeta} \right \rvert_{\bbeta = \bbeta_0^*} = \bSigma_0$ and $n^{-1/2} \hat{\bU}_n(\bbeta_0^*)$ converges to a normal distribution.  
  We prove the second convergence by showing that $n^{1/2}\hat{\bU}_n(\bbeta_0^*)$ has a weighted asymptotically linear expansion: $n^{1/2} \hat{\bU}_n(\bbeta_0^*) = n^{-1/2}\sum_{i=1}^n \hat{w}_i \bU_i(\bbeta_0^*) + o_p(1)$
  with  
  $$
  \bU_i(\bbeta_0^*) = \int^{\tau_2}_0 \left[ \bX_i(t) - \frac{\bs^{(1)}(\bbeta_0^*, t)}{\bs^{(0)}(\bbeta_0^*, t)}\right]dN_i(t) - \int^{\tau_2}_0 \frac{Y_i(t) \exp({\bbeta_0^*}^T \bX_i(t))}{\bs^{(0)}(\bbeta_0^*, t)} \left[ \bX_i(t) - \frac{\bs^{(1)}(\bbeta_0^*, t)}{\bs^{(0)}(\bbeta_0^*, t)}\right] d\tilde{N}(t) 
  $$
  for $i = 1, \ldots, n$ and $\tilde{N}(t) = \E[N(t)]$.  Similar asymptotic linear expansions have appeared in \cite{lin1989robust} and \cite{lin2000fitting}.   This weighted asymptotic linear  expansion motivates the study of the IPLW empirical measure and processes.  

We start by giving the weighted asymptotic linear expansion of the IPLW partial score $\hat{\bU}_n(\bbeta)$ and subsequently give its asymptotic distribution. 
\begin{theorem}[Asymptotic linear expansion]  \thlabel{thm::thm_Cox_right_asn_d}
Under assumptions (D1) - (D4),
we have the following two results:
\begin{enumerate}
\item For each $\bbeta$, $n^{1/2} \hat{\bU}_n(\bbeta) = n^{-1/2}\sum_{i=1}^n \hat{w}_i \bU_i(\bbeta) + o_p(1)$ such that 
$$
\bU_i(\bbeta) = \int^{\tau_2}_0 \left[ \bX_i(t) - \frac{\bs^{(1)}(\bbeta, t)}{\bs^{(0)}(\bbeta, t)}\right]dN_i(t) - \int^{\tau_2}_0 \frac{Y_i(t) \exp(\bbeta^T \bX_i(t))}{\bs^{(0)}(\bbeta, t)} \left[ \bX_i(t) - \frac{\bs^{(1)}(\bbeta, t)}{\bs^{(0)}(\bbeta, t)}\right] d\tilde{N}(t) 
$$
with $\tilde{N}(t) = \E[ N(t)]$.  
\item $n^{-1/2} \hat{\bU}_n(\bbeta_0^*) \rightarrow_d N(0, \bSigma_U)$ with 
\begin{align*}
\bSigma_U &= \Var[\bU(\bbeta_0^*)] + \E \left[\bU(\bbeta_0^*) \bU(\bbeta_0^*)^T \frac{I(Q = 0)[1 - \pi_{\bgamma_0}(\bZ_1)]}{\pi_{\bgamma_0}(\bZ_1)}
\right] \\
& - \bQ_e(\bbeta_0^*)^T \bSigma_{\bgamma_0}^{-1} \bQ_e(\bbeta_0^*),
\end{align*}
where 
\begin{align*}
    & \bU(\bbeta_0^*) = \int^{\tau_2}_0 \left[ \bX(t) - \frac{\bs^{(1)}(\bbeta, t)}{\bs^{(0)}(\bbeta, t)}\right]dN(t) - \int^{\tau_2}_0 \frac{Y(t) \exp(\bbeta^T \bX(t))}{\bs^{(0)}(\bbeta, t)} \left[ \bX(t) - \frac{\bs^{(1)}(\bbeta, t)}{\bs^{(0)}(\bbeta, t)}\right] d\tilde{N}(t) \\
    & \bQ_e(\bbeta_0^*) = \E[ I(Q = 0) (1 - \pi_{\bgamma_0}(\bZ_1)) \tZ_1 \bU(\bbeta_0^*)^T]
    & 
\end{align*}
and $\bSigma_{\bgamma_0} = \E[I(Q = 0) \pi_{\bgamma_0}(\bX)(1 - \pi_{\bgamma_0}(\bX)) \tX \tX^T]$ is the asymptotic variance of $\hat \bgamma_n$ and can be found in \thref{lem::lem_estimated_weight}. 
\end{enumerate}
\end{theorem}
The proof of \thref{thm::thm_Cox_right_asn_d} can be found in appendix \ref{sec::proofs}.
We first prove that the IPLW partial score can be written in the weighted asymptotic linear expansion form (first assertion).  
Then we use the weak convergence result of the IPLW empirical process from \thref{ipw_weak_converge} to obtain the asymptotic distribution of $n^{1/2} \hat{\bU}_n(\bbeta)$.
Next let $\bA_n(\bbeta) = -\frac{\partial \hat{\bU}_n(\bbeta)}{\partial \bbeta}$ and we can write $\bA_n(\bbeta)$ as 
\begin{equation}
\bA_n(\bbeta) = -\frac{\partial \hat{\bU}_n(\bbeta)}{\partial \bbeta} = \int^{\tau_2}_0 \left\{ \frac{\bS_{n, w}^{(2)}(\bbeta, t)}{\bS_{n, w}^{(0)}(\bbeta, t)} - \left(\frac{\bS_{n, w}^{(1)}(\bbeta, t)}{\bS_{n, w}^{(0)}(\bbeta, t)}\right)^{\otimes 2}\right\} d\bar{N}(t)
\label{eq::An}
\end{equation}
and similarly define its population version $\bA(\bbeta)$
\begin{equation}
\begin{aligned}
\bA(\bbeta) &= -\frac{\partial \bU_0(\bbeta)}{\partial \bbeta} = \int^{\tau_2}_0 \left\{ \frac{\bs^{(2)}(\bbeta, t)}{\bs^{(0)}(\bbeta, t)} - \left(\frac{\bs^{(1)}(\bbeta, t)}{\bs^{(0)}(\bbeta, t)}\right)^{\otimes 2}\right\} d \tilde{N}(t) \\
& = \int^{\tau_2}_0 \left\{ \frac{\bs^{(2)}(\bbeta, t)}{\bs^{(0)}(\bbeta, t)} - \left(\frac{\bs^{(1)}(\bbeta, t)}{\bs^{(0)}(\bbeta, t)}\right)^{\otimes 2}\right\} \bs^{(0)}(t) dt
\end{aligned}
\end{equation}
Under assumption (D5),  $\bA(\bbeta_0^*) = \bSigma_0$ is positive definite and we later prove that $\bA_n(\bbeta) \rightarrow_p \bA(\bbeta)$ for $\bbeta$ in a compact set $\mathbb{B}$ that contains $\bbeta_0^*$. 
Together with the asymptotic normality of the IPLW partial score, we can derive the asymptotic normality of the estimator $\hat \bbeta_n$.

\section{Empirical process theory for an IPLW process}  \label{sec::appendix_ipw_process}

{\bf Assumptions.}
\begin{itemize}
\item[\bf (D2)] The time-dependent covariates $\bX_i(t)$ have bounded total variation such that $\|\bX_i(0)\|_1 + \int_0^{\tau_2} \|d \bX_i(t)\|_1 \leq B$ for 
a fixed constant $B > 0$. 
\item[\bf (D3)]
$P(L = 1 | Q = 0, \bZ_1) = \pi_{\bgamma_0}(\bZ_1) \geq \delta > 0$ for all possible values of $\bZ_1$. 
\item[\bf (D4)] 
The failure time and censoring time satisfy 
$$
P(T \geq s | C_1, C_2,  \bar{\bX}(s)) = P(T \geq s | \bar{\bX}(s))
$$
for $s \in [0, \tau_2]$ and $P(\tilde{T} \geq \tau_2) > 0$. 
\item[\bf (D5)]
 $\bSigma_0 = \int^{\tau_2}_0 \left\{ \frac{\bs^{(2)}(\bbeta_0^*, t)}{\bs^{(0)}(\bbeta_0^*, t)} - \left(\frac{\bs^{(1)}(\bbeta_0^*, t)}{\bs^{(0)}(\bbeta_0^*, t)} \right)^{\otimes 2} \right\} \bs^{(0)}(t) dt$ is positive definite. 
\end{itemize}
(D2) assumes that time-dependent covariates have bounded variation \citep{bilias1997towards}. (D3) is a standard positivity assumption for IPW type approach. (D4) is an independent censoring assumption \citep{bilias1997towards} and basically requires that a positive fraction of participants are still at-risk after the end of the observational dataset.  (D5) is a standard assumption for Cox models \citep{andersen1982Cox,lin1989robust} that ensures the uniqueness of $\bbeta_0^*$.

To study the asymptotic properties of  $\hat \bbeta_n$, we need to generalize empirical process theory
into an IPW scenario.
We first introduce an IPLW empirical measure 
$$
\P^{\pi}_n = \frac{1}{n} \sum_{i=1}^n \frac{I(L_i + Q_i > 0)}{Q_i + (1 - Q_i) \pi_{\bgamma_0}(\bX_i)} \delta_{\bX_i, \tilde{T}_i, \Delta_i} = \frac{1}{n} \sum_{i=1}^n w_i \delta_{\bX_i, \tilde{T}_i, \Delta_i}
$$ 
where $\delta_{\bX_i, \tilde{T}_i, \Delta_i}$ is the Dirac measure placing unit mass on $(\bX_i, \tilde{T}_i, \Delta_i)$ and $w_i = I(L_i + Q_i > 0) / [Q_i + (1 - Q_i)\pi_{\bgamma_0}(\bX_i)]$ such that  $\P^{\pi}_n f = \frac{1}{n} \sum_{i=1}^n w_i f(\bX_i, \tilde{T}_i, \Delta_i)$ for a function $f = f(\bx, \tilde{t}, \delta)$.  
In practice, $\pi_{\bgamma_0}(X_i)$ is unknown,  so we introduce the IPLW empirical measure with estimated weight
$$
\P^{\pi, e}_n = \frac{1}{n} \sum_{i=1}^n \frac{I(L_i + Q_i > 0)}{Q_i + (1 - Q_i) \pi_{\hat{\bgamma}_n}(\bX_i)} \delta_{\bX_i, \tilde{T}_i, \Delta_i} = \frac{1}{n} \sum_{i=1}^n \hat{w}_i \delta_{\bX_i, \tilde{T}_i, \Delta_i}
$$
with $\bgamma_0$ replaced by $\hat{\bgamma}_n$. Finally, we denote $\P_0$ as the probability measure corresponding to the true distribution such that $\P_0 f = \E[ f(\bX, \tilde{T}, \Delta)]$.  Note the usual empirical measure $\P_n = \frac{1}{n} \sum_{i=1}^n \delta_{\bX_i, \tilde{T}_i, \Delta_i}$ is  unobserved due to the missingness in Class 3 of
Figure 1 in the main paper.
The  IPLW empirical measure  leads to the IPLW empirical processes $\G^{\pi}_n = \sqrt{n}(\P^{\pi}_n - \P_0)$ and $\G^{\pi, e}_n = \sqrt{n}(\P^{\pi, e}_n - \P_0)$.  
It turns out that our IPLW empirical measure and empirical process also enjoy similar asymptotic properties as the usual empirical measure and empirical processes.  For any $\phi: {\cal F} \rightarrow \mathbb{R}$, we write $\| \phi(f)\|_{{\cal F}} = \sup_{f \in {\cal F}} | \phi(f)|$.
We say that ${\cal F}$ is $\P$-Glivenko-Cantelli if and only if $\sup_{f \in {\cal F}} | (\P_n - \P) f| = o_P(1)$. 
We first prove that Glivenko-Cantelli property also holds for the IPLW empirical process. 
\begin{proposition}[IPLW uniform convergence]\thlabel{thm::thm_gc}
Suppose that ${\cal F} = \{f(\bx,  \tilde{t}, \delta)\}$ is $\P_0$-Glivenko-Cantelli with an integrable envelope function $F$ such that $\P_0 F < \infty$.  Under assumption (A1-3), 
$$
\|\P^{\pi}_n - \P_0\|_{\cal F} \rightarrow_{P^*} 0.
$$
If $\hat{\bgamma}_n \rightarrow_p \bgamma_0$,  then
$
\|\P^{\pi, e}_n - \P_0\|_{\cal F} \rightarrow_{P^*} 0
$ also holds. 
\end{proposition}
The proof can be found in Appendix \ref{sec::proofs}.   
Recall that $\tX = (1, \bX^T)^T$ and $\tx = (1, \bx^T)^T$. Now, to definite weak convergence, 
let $X_n$ be a bounded process and $X$ be a bounded process whose finite-dimensional laws correspond to the finite dimensional projections of a tight Borel law on $\ell^{\infty}({\cal F})$.  We say that $X_n \rightsquigarrow X$ in $\ell^{\infty}({\cal F})$ if and only if $\E^* H(X_n) \rightarrow \E H(X)$ for all $H \in C_b(\ell^{\infty}({\cal F}))$, where $C_b(\ell^{\infty}({\cal F}))$ denotes all bounded continuous functions on $\ell^{\infty}({\cal F})$ \citep{van1996weak,van2000asymptotic}.  
The next theorem states the weak convergence result for IPLW empirical process. 
\begin{proposition}[IPLW weak convergence]  \thlabel{ipw_weak_converge}
Under assumption (A1-3),  suppose that ${\cal F} = \{f(\bx, \tilde{t}, \delta\}$ is $\P_0$-Donsker with an integrable envelope function $F$ such that $\P_0 F < \infty$, then  
\begin{align*}
& \G_n^{\pi} \rightsquigarrow \G(g_1 \cdot) \\
& \G_n^{\pi, e} \rightsquigarrow \G^e = \G(g_1 \cdot - g_2 \bQ_e(\cdot)^T \bg_3)
\end{align*}
in $\ell^{\infty}({\cal F})$ where $g_1(l, q, \bx) = \frac{I(l + q > 0)}{q + (1 - q) \pi_{\bgamma_0}(\bx)}$, $g_2(l, q, \bx) = I(q = 0)[l - \pi_{\bgamma_0}(\bx)]$,  $\bg_3(\bx) = \bSigma_{\bgamma_0}^{-1} \tx$ and $\bQ_e(f) = \E[I(Q = 0)(1 - \pi_{\bgamma_0}(\bX)) f(\bX, \tilde{T}, \Delta) \tX]$.  $\G$ is the $\P_0$-Brownian bridge process, indexed by ${\cal F}$.  
\end{proposition}

\section{Proofs}  \label{sec::proofs}

\begin{proof}[ of Proposition 3.3]
For simplicity, denote $P(L = 1 | Q = 0, \bX) = \pi_0(\bX)$. 
First, we have
\begin{align*}
\E\left[ \frac{I(L = 1) l(\bbeta, \lambda_0)}{\pi_0(\bX)} \middle | Q = 0 \right]  & = \E\left[ \E\left[ \frac{I(L = 1) l(\bbeta, \lambda_0)}{\pi_0(\bX)} \middle| Q = 0, \bX, \tilde{T}, \Delta \right] \middle | Q = 0 \right] \\
& = \E\left[ \frac{ l(\bbeta, \lambda_0) }{ \pi_0(\bX)}
\E\left[ I(L = 1) \middle | Q = 0, \bX, \tilde{T}, \Delta \right]
 \middle | Q = 0 \right] \\
 & = \E\left[ l(\bbeta, \lambda_0) | Q = 0 \right]
\end{align*}
The second to last equality holds as $l(\beta, \lambda_0)$ is a function of $\tilde{T}, \bX, \Delta$.  The last equality holds by assumption (A1).  
To prove equation (5) in Proposition 3.3, we have 
\begin{align*}
& \E \left[ \frac{I(L + Q > 0)}{Q + (1 - Q)\pi_0(\bX)} l(\bbeta, \lambda_0)
\right] \\
& = \E\left[ \frac{I(L = 1)}{\pi_0(\bX)} l(\bbeta, \lambda_0)
\middle | Q = 0 \right] P(Q = 0) + \E \left[ l(\bbeta, \lambda_0) \middle | Q = 1
\right] P(Q = 1) \\
& = \E\left[ l(\bbeta, \lambda_0) | Q = 0 \right] P(Q = 0) + \E \left[ l(\bbeta, \lambda_0) \middle | Q = 1
\right] P(Q = 1) = \E(l(\bbeta, \lambda_0))
\end{align*}
\end{proof}

We first give the asymptotic distribution of $\hat{\bgamma}_n$,  the estimates of the logistic regression parameter $\bgamma_0$. 
\begin{lemma}  \thlabel{lem::lem_estimated_weight}
Under assumptions (A1), (A2) and (A3),  $\hat{\bgamma}_n$ is consistent for $\bgamma_0$ and 
$$
\sqrt{n}(\hat{\bgamma}_n - \bgamma_0) \rightarrow_d N(0, \bSigma^{-1}_{\bgamma_0})
$$
where $\bSigma_{\bgamma_0} = \E\left[I(Q = 0) \frac{\exp(\tX^T \bgamma_0) \tX \tX^T }{(1 + \exp(\tX^T \bgamma_0))^2} \right] = \E[I(Q = 0) \pi_{\bgamma_0}(\bX)(1 - \pi_{\bgamma_0}(\bX)) \tX \tX^T]$. 
\end{lemma}
\begin{proof}[ of \thref{lem::lem_estimated_weight}]
First, the score for logistic regression converges as
\begin{align*}
\bS_n(\bgamma) = \frac{1}{n} \sum_{i=1}^n I(Q_i = 0) \left[L_i \tX_i - \frac{\exp(\tX_i^T \bgamma) \tX_i}{1 + \exp(\tX_i^T \bgamma)} \right] \rightarrow_p \E \left[I(Q = 0) \left(L \tX - \frac{\exp(\tX^T \bgamma) \tX}{1 + \exp(\tX^T \bgamma)} \right) \right] = \bS(\bgamma)
\end{align*}
and 
\begin{align*}
\bS(\bgamma_0) &= \E \left[ I(Q = 0) \left(L\tX - \frac{\exp(\tX^T \bgamma_0) \tX}{1 + \exp(\tX^T \bgamma_0)} \right) \right] \\
& = \E \left[  
\E \left[ I(Q = 0) \left(
L \tX - \P(L = 1 | Q = 0; \bX) \tX
\right)
 | Q,  \bX
\right]
\right] \\
& = \E \left(
I(Q = 0)
\left[\E[L | Q, \bX] \tX - \P(L = 1 | Q = 0, \bX)\tX 
\right]
\right) = 0
\end{align*}
since $\E[L | Q, \bX] = I(Q = 0) \P(L = 1 | Q = 0; \bX) + I(Q = 1) \P(L = 1 | Q = 1; \bX)$. 
Further,  by the law of large numbers,
\begin{align*}
\nabla \bS_n(\bgamma) = -\frac{1}{n} \sum_{i=1}^n I(Q_i = 0) \frac{\exp(\tX_i^T \bgamma) \tX_i \tX_i^T }{\left[ 1 + \exp(\tX_i^T \bgamma)\right]^2} \rightarrow_p -\E\left[  I(Q = 0) \frac{\exp(\tX^T \bgamma) \tX \tX^T}{(1 + \exp(\tX^T \bgamma))^2} \right] = -\bSigma_{\bgamma}
\end{align*}
We assume that $\bSigma_{\bgamma}$ is positive-definite, thus $\bgamma_0$ is the unique solution for $\bS(\bgamma) = \bm{0}$.  Next we verify the uniform convergence condition: 
$$
\sup_{\bgamma} \| \bS_n(\bgamma) - \bS(\bgamma) \|  \rightarrow_p 0
$$
Denote $\bphi_{\bgamma}(l, q, \bx) = I(q = 0)\left[l - \pi_{\bgamma}(\bx)\right] \tx$ and $\bS_n(\bgamma) - \bS(\bgamma) = (\P_n - \P_0) \bphi_{\bgamma}$.  The function class $\{\bphi_{\bgamma}(l, q, \bx):  \bgamma\}$ forms a VC-subgraph class by Lemma 2.6.15, 2.6.18 of \cite{van1996weak}.  Thus, by Theorem 5.9 of \cite{van2000asymptotic},  we have 
$\hat{\bgamma}_n \rightarrow_p \bgamma_0$. 

For asymptotic normality of $\hat{\bgamma}_n$, note that 
\begin{align*}
\|\bphi_{\bgamma_1}(L, Q, \bX) - \bphi_{\bgamma_2}(L, Q, \bX) \| \leq 
\left\| 
\left[
 \frac{\exp(\tX^T \bgamma_1)}{1 + \exp(\tX^T \bgamma_1)} - \frac{\exp(\tX^T \bgamma_2)}{1 + \exp(\tX^T \bgamma_2)}
  \right]\tX
\right\|
  \leq \frac{1}{4} \|\tX \tX^T\| \|\bgamma_1 - \bgamma_2\|
\end{align*}
Under the assumption that $\bX$ is bounded, we have $\E \|\tX \tX^T\| < \infty$.  $\E \bphi_{\bgamma}(L, Q, \bX)$ is differentiable at $\bgamma_0$ with derivative $\bSigma_{\bgamma_0}$.   By Theorem 5.21 of \cite{van2000asymptotic},  we concluded that 
$$
\sqrt{n}(\hat{\bgamma}_n - \bgamma_0) = \bSigma_{\bgamma_0}^{-1} \frac{1}{\sqrt{n}} \sum_{i=1}^n \bphi_{\bgamma_0}(L_i, Q_i, \bX_i) + o_P(1) \rightarrow_d N(0,  \bSigma_{\bgamma_0}^{-1})
$$
\end{proof}

\begin{proof}[ of \thref{thm::thm_gc}]
We start with bounding the difference between the IPLW empirical measure $\P_n^{\pi}$ and the usual empirical measure $\P_n$. 
By triangle inequality,
$$
\|\P^{\pi}_n - \P_0\|_{\cal F} \leq \|\P_n - \P_0\|_{\cal F} + \left\|\frac{1}{n}\sum_{i=1}^n \left(\frac{I(L_i + Q_i > 0)}{Q_i + (1 - Q_i)\pi_{\gamma_0}(\bX_i)} - 1 \right) \delta_{\bX_i, \tilde{T}_i, \Delta_i}\right\|_{\cal F}
$$
The first term is $o_{P^*}(1)$ since ${\cal F}$ is Glivenko-Cantelli. For the second term,  define function $g(l, q, \bx) = \frac{I(l + q > 0)}{q + (1 - q)\pi_{\bgamma_0}(\bx)} - 1$, then consider the function class ${\cal F}^* = \left\{ \left[g \cdot f\right](l, q, \bx, \delta, \tilde{t}):  f \in {\cal F}  \right\}$ where $f = f(\bx, \tilde{t}, \delta)$.  By assumption (A3),  $g(l, q, \bx)$ is bounded and ${\cal F}$ has an integrable envelope function $F$.  These two together imply that ${\cal F}^*$ is $\P_0$-Glivenko-Cantelli by the Glivenko-Cantelli Preservation theorem \citep{van2000preservation}. Then,  for any $f \in {\cal F}$, 
\begin{align*}
\P_0 g f & = \E \left[ \E[g(L, Q, \bX) f(\bX, \tilde{T}, \Delta) | Q]  \right] \\
& = \E[g(L, Q, \bX) f(\bX, \tilde{T}, \Delta) | Q = 0] \P(Q = 0) + \E[g(L, Q, \bX) f(\bX, \tilde{T}, \Delta) | Q = 1] \P(Q = 1) \\
& = \E[0 * f (\bX, \tilde{T}, \Delta)| Q = 1] \P(Q = 1) + \E \left[\left(\frac{I(L = 1)}{\pi_{\bgamma_0}(\bX)} - 1 \right) f(\bX, \tilde{T}, \Delta) \middle | Q = 0 \right] \P(Q = 0) \\
& = \E \left[ f(\bX, \tilde{T}, \Delta) \E\left(\frac{I(L = 1)}{\pi_{\bgamma_0}(\bX)} - 1 \middle | Q = 0, \bX, \tilde{T}, \Delta \right)  \middle | Q = 0 \right] \P(Q = 0) = 0
\end{align*}
The last equality is due to the assumption (A1).  Thus,  the second term could be rewritten as 
$$
\left\|\frac{1}{n} \sum_{i=1}^n \delta_{\bX_i, \tilde{T}_i, \Delta_i,  L_i, Q_i}\right\|_{{\cal F}^*} = \|\P_n - \P_0\|_{{\cal F}^*}
$$
which is again $o_{P^*}(1)$. 

 Now consider $\P_n^{\pi, e}$.   Since $\hat{\bgamma}_n \rightarrow_p \bgamma_0$, it suffices to consider a small compact neighborhood $\mathbb{K} \subset \R^{p + 1}$ of $\bgamma_0$.  Let $\xi_{\bgamma}(\bx, q) = \frac{q + (1 - q)\pi_{\bgamma_0}(\bx)}{q + (1 - q) \pi_{\bgamma}(\bx)}$. Since $\bX$ is bounded and $\pi_{\bgamma}$ is continuous in $\bX$,  $\xi_{\bgamma}(\bX, Q)$ is also bounded.  Lemma 2.6.15 and 2.6.18 of \cite{van1996weak} then imply that $\{\xi_{\bgamma}(\bx, q):  \bgamma \in \mathbb{K} \}$ is a VC-subgraph class.  Next the Glivenko-Cantelli Preservation theorem \citep{van2000preservation} implies that
$$
{\cal G} = \left \{ h_{\bgamma, f} (q, \bx, \tilde t, \delta) = \underbrace{\frac{q + (1 - q) \pi_{\bgamma_0}(\bx)}{q + (1 - q)\pi_{\bgamma}(\bx)}}_{\xi_{\bgamma}(\bx, q)} f(\bx, \tilde t,\delta):  f \in {\cal F}, \bgamma \in \mathbb{K} \right\}
$$
is a $\P_0$-Glivenko-Cantelli class as ${\cal F}$ has an integrable envelope function and $\xi_{\bgamma}(x, q)$ is bounded.  Then recognizing that 
\begin{align*}
\P_n^{\pi, e} f & = \frac{1}{n} \sum_{i=1}^n \frac{I(L_i + Q_i > 0)}{Q_i + (1 - Q_i)\pi_{\bgamma_0}(\bX_i)} \left\{\frac{Q_i + (1 - Q_i)\pi_{\bgamma_0}(\bX_i)}{Q_i + (1 - Q_i)\pi_{\hat{\bgamma}_n}(\bX_i)} \right\} f(\bX_i, \tilde{T}_i, \Delta_i) = \P_n^{\pi} \xi_{\hat\bgamma_n}  f
\end{align*}
We have 
\begin{align*}
\|\P^{\pi, e}_n - \P_0\|_{\cal F} &= \sup_{f \in {\cal F}} \left\| \P_n^{\pi} \xi_{\hat \bgamma_n} f - \P_0 f \right\| \leq \sup_{f \in {\cal F}} \left\| \P_n^{\pi} \xi_{\hat \bgamma_n}f - \P_0 \xi_{\hat \bgamma_n} f \right\| + \sup_{f \in {\cal F}} \left \| \P_0 \xi_{\hat \bgamma_n} f - \P_0 f  \right\|
\end{align*}
Notice the first term is $o_{P^*}(1)$ since ${\cal G}$ is $\P_0$-Glivenko-Cantelli.  For the second term, we have 
\begin{equation}
\begin{aligned}
\frac{1}{Q + (1 - Q) \pi_{\hat{\bgamma}_n}(\bX)} - \frac{1}{Q + (1 - Q) \pi_{\bgamma_0}(\bX)} = I(Q = 0) \left(1 - \frac{1}{\pi_{\bgamma^*}(\bX)} \right) \tX^T (\hat{\bgamma}_n - \bgamma_0)
\end{aligned}
\label{eqn::eqn_ipw_diff}
\end{equation}
where $\bgamma^*$ is some convex combinations of $\bgamma_0$ and $\hat{\bgamma}_n$ and $\bgamma^* \rightarrow_p \bgamma_0$.  This implies that 
\begin{equation}
\begin{aligned}
\xi_{\hat{\bgamma}_n}(\bX, Q) - 1 = \xi_{\hat{\bgamma}_n}(\bX, Q) - \xi_{\bgamma_0}(\bX, Q) =  I(Q = 0) \pi_{\bgamma_0}(\bX) \left(1 - \frac{1}{\pi_{\bgamma^*}(\bX)} \right) \tX^T (\hat{\bgamma}_n - \bgamma_0)
\end{aligned}
\label{eqn::eqn_gamma_diff}
\end{equation}
and the second term can be rewritten as following:
\begin{align*}
& \sup_{f \in {\cal F}} \left \|  \E \left \{f(\bX, \tilde{T},\Delta) I(Q = 0) \pi_{\bgamma_0}(\bX) \left(1 - \frac{1}{\pi_{\bgamma^*}(\bX)} \right) \tX^T (\hat{\bgamma}_n - \bgamma_0) \right\}
 \right \|\\
& = \sup_{f \in {\cal F}} \left \|  \left\{f(\bX, \tilde{T},\Delta) I(Q = 0) \pi_{\bgamma_0}(\bX) \left(1 - \frac{1}{\pi_{\bgamma^*}(\bX)} \right) \tX^T  \right\}
\right \| \| \hat{\bgamma}_n - \bgamma_0\| 
\end{align*}
 Further, since $\pi_{\bgamma}(\bX)$ is bounded away from zero for $\bgamma \in \mathbb{K}$ and $\|\P_0\|_{\cal F} < \infty$,  the second term is now determined by $\|\hat{\bgamma}_n - \bgamma_0\|$, which is $o_P^*(1)$. 
\end{proof}
The proof of \thref{ipw_weak_converge} relies on the following lemma.

\begin{lemma} \thlabel{ipw_equi_cont_gamma}
Let $\zeta_{\bgamma}(l ,q, \bx) = \frac{I(l + q > 0)}{q + (1 - q)\pi_{{\bgamma}}(\bx)}$
and ${\cal F} = \{f(\bx, \tilde{t}, \delta\}$ be a $\P_0$-Glivenko-Cantelli class with an integrable envelope function $F$ such that $\P_0 F < \infty$. Then 
$$
\sup_{f \in {\cal F}} \left\| \sqrt{n}(\P_n - \P_0)\left(
\zeta_{\hat \bgamma_n} f - \zeta_{\bgamma_0}f \right) \right\| = o_P^*(1)
$$
\end{lemma}
\begin{proof} 
First, recall that $\xi_{\bgamma}(q, x) = \frac{q + (1 - q) \pi_{\bgamma_0}(\bx)}{q + (1 - q) \pi_{\bgamma}(\bx)}$,  then we have that 
\begin{align*}
& \P_n \zeta_{\hat \bgamma_n} f = \P_n^{\pi} \xi_{\hat \bgamma_n} f; \quad \P_n \zeta_{\bgamma_0} f = \P_n^{\pi} \xi_{\bgamma_0} f
\end{align*}
Further,  since 
\begin{align*}
\P_0 \zeta_{\bgamma} &= \E_0 \left[ \frac{I(L + Q > 0)}{Q + (1 - Q) \pi_{\bgamma}(\bX)} \right] \\
& = \E_0 \left[ \E_0 \left[ \frac{I(L + Q > 0)}{Q + (1 - Q) \pi_{\bgamma}(\bX)} \middle | Q, \bX\right] \right] = \E_0 \left[ \frac{Q + (1 - Q) \pi_{\bgamma_0}(\bX)}{Q + (1 - Q) \pi_{\bgamma}(\bX)}\right] = \P_0 \xi_{\bgamma}
\end{align*}
We further have that 
\begin{align*}
& \P_0 \zeta_{\hat \bgamma_n} f = \P_0 \xi_{\hat \bgamma_n} f;  \quad \P_0 \zeta_{\bgamma_0} f = \P_0 \xi_{\bgamma_0} f
\end{align*}
Given above results and equation \eqref{eqn::eqn_gamma_diff}, we have 
\begin{align*}
& \sup_{f \in {\cal F}} \left\| \sqrt{n} (\P_n - \P_0)(\zeta_{\hat \bgamma_n} f - \zeta_{\bgamma_0} f) \right\| = \sup_{f \in {\cal F}} \left\| \sqrt{n} (\P_n^{\pi} - \P_0)(\xi_{\hat \bgamma_n} f - \xi_{\bgamma_0} f) \right\| \\
& \leq  \sup_{f \in {\cal F}} \left \| (\P_n^{\pi} - \P_0) \left( I(Q = 0) \pi_{\bgamma_0}(\bX) \left(1 - \frac{1}{\pi_{\bgamma^*}(\bX)} \right) f(\bX, \tilde{T}, \Delta) \tX^T
 \right) \right\| \left \|\sqrt{n} (\hat{\bgamma}_n - \bgamma_0)\right\|
\end{align*}
where $\bgamma^*$ is a point lies between $\hat \bgamma_n$ and $\bgamma_0$.
Note that the consistency of $\hat{\bgamma}_n$ implies that $\bgamma^*\overset{p}{\rightarrow}\bgamma_0$ and $\bgamma^*$ will fall into a compact small neighborhood $\mathbb{K}$ around $\bgamma_0$ with probability 1. As a result, 
${\cal G} = \{I(q = 0)\pi_{\bgamma_0}(\bx) \left(1 - \frac{1}{\pi_{\bgamma}(\bx)}\right)\tx^T:   \bgamma\in \mathbb{K} \}$ forms a VC-subgraph class by Lemma 2.6.15 and 2.6.18 of \cite{van1996weak}.  Then by the Glivenko-Cantelli Preservation theorem \citep{van2000preservation},  we have ${\cal F}_1 = \{f g: f \in {\cal F}, g \in {\cal G} \}$ is a $\P_0$-Glivenko-Cantelli class with an integrable envelope function. Then 
\begin{align*}
\sup_{f \in {\cal F}} \left\| \sqrt{n} (\P_n - \P_0)(\zeta_{\hat \bgamma_n} f - \zeta_{\bgamma_0} f) \right\| \leq \|\P_n^{\pi} - \P_0\|_{{\cal F}_1} \sqrt{n}\|\hat \bgamma_n - \bgamma_0\|
\end{align*}
By \thref{thm::thm_gc} and \thref{lem::lem_estimated_weight},  $\sqrt{n}\|\hat{\bgamma}_n - \bgamma_0\| = O_P^*(1)$ and $\|\P_n^{\pi} - \P_0\|_{{\cal F}_1} = o_P^*(1)$. 
\end{proof}

\begin{proof}[ of \thref{ipw_weak_converge}]
Let $\zeta_{\bgamma}(l,  q, \bx) = \frac{I(l + q > 0)}{q + (1 - q)\pi_{{\bgamma}}(\bx)}$.
This implies that $\zeta_{\bgamma_0}(L,Q, \bX) = g_1(L,Q, \bX)$, where $g_1$ is defined in \thref{ipw_weak_converge}.
Moreover, under assumption (A3) and $\|\P_0\|_{\cal F} < \infty$, ${\cal F}' = \{g_1 \cdot f:  f \in {\cal F} \}$ is a Donsker Class by Example 2.10.10 of \cite{van1996weak}. 

By the definition that $\G_n^{\pi} = \sqrt{n}(\P_n^{\pi} - \P_0)$, we have
$$
\G_n^{\pi}f = \sqrt{n}(\P_n - \P_0)\zeta_{\bgamma_0}f =  \sqrt{n}(\P_n - \P_0)g_1 f.
$$ 
Thus, the usual Donsker theorem (see, e.g., Section 19.2 of \citealt{van2000asymptotic}) implies that $\G^{\pi}_n \rightsquigarrow \G(g_1 \cdot)$ in $l^{\infty}({\cal F})$ and 
$$
\Var(\G (g_1 f)) = \Var(f(\bX, \tilde{T}, \Delta)) + \E\left[f(\bX, \tilde{T}, \Delta)^2 \frac{I(Q = 0)[1 - \pi_{\bgamma_0}(\bX)]}{\pi_{\bgamma_0}(\bX)}\right]
$$
Next,  
\begin{equation}
\begin{aligned}
\G_n^{\pi, e} f - \G_n^{\pi} f & = \sqrt{n}(\P_n^{\pi, e} - \P_0)f - \sqrt{n}(\P_n^{\pi} - \P_0)f \\
& = \G_n \left(\zeta_{\hat {\bgamma}_n}- \zeta_{ {\bgamma}_0} \right)f  +  \sqrt{n} \P_0 \left( \zeta_{\hat {\bgamma}_n} - \zeta_{ {\bgamma}_0}\right)f 
\end{aligned}
\label{eq::WC::01}
\end{equation}
By \thref{ipw_equi_cont_gamma},  $\G_n \left(\zeta_{ \hat{\gamma}_n} - \zeta_{ {\gamma}_0} \right)f$ is $o_P^*(1)$,
which bounds the first term. 
For the second term,  note that 
\begin{align*}
\zeta_{\hat \bgamma_n} - \zeta_{\bgamma_0} = - I(L = 1) I(Q = 0) \left( \frac{1}{\pi_{\bgamma_0}(\bX)} - 1 \right) \tX^T (\hat{\bgamma}_n - \bgamma_0) + o(\|\hat{\bgamma}_ n - \bgamma_0 \|^2)
\end{align*}
and $\E[I(L = 1) | Q = 0, \bX] = \pi_{\bgamma_0}(\bX)$, we can show that  
\begin{align*}
\sqrt{n} \P_0 \left( \zeta_{\hat {\bgamma}_n} - \zeta_{ {\bgamma}_0}\right)f 
= - \left\{\E \left[ I(Q = 0) (1 - \pi_{\bgamma_0}(\bX)) f(\bX, \tilde{T}, \Delta) \tX^T \right]\right\} \sqrt{n}(\hat{\bgamma}_n - \bgamma_0) +o_P^*(1) .
\end{align*}
%
%
%
Together,  equation \eqref{eq::WC::01} and \thref{lem::lem_estimated_weight} implies that
\begin{align*}
\G_n^{\pi, e} f &= \underbrace{\frac{1}{\sqrt{n}}\sum_{i=1}^n \left[g_1(L_i, Q_i, \bX_i) f(\bX_i, \tilde{T}_i, \Delta_i) - \P_0 f \right]}_{ \G_n^{\pi} f}\\
& \underbrace{- \frac{1}{\sqrt{n}} \E \left[ I(Q = 0) (1 - \pi_{\bgamma_0}(\bX)) f(\bX, \tilde{T}, \Delta) \tX^T \right] \bSigma_{\bgamma_0}^{-1}\sum_{i=1}^n I(Q_i = 0)[L_i - \pi_{\bgamma_0}(\bX_i)] \bX_i + o_P^*(1)}_{=\sqrt{n} \P_0 \left( \zeta_{\hat {\bgamma}_n} - \zeta_{ {\bgamma}_0}\right)f} \\
& = \G_n \left[ g_1 \cdot f -  g_2 \bQ_e(f)^T g_3\right] + o_P^*(1)
\end{align*}
where $g_1(l, q, \bx) = \frac{I(l + q > 0)}{q + (1 - q) \pi_{\bgamma_0}(\bx)}$, $g_2(l, q, \bx) = I(q = 0)[l - \pi_{\bgamma_0}(\bx)]$,  $\bg_3(\bx) = \bSigma_{\bgamma_0}^{-1} \tx$ and $\bQ_e(f) = \E[I(Q = 0)(1 - \pi_{\bgamma_0}(\bX)) f(\bX, \tilde{T}, \Delta) \tX]$.  This proves the finite dimensional convergence of $\G_n^{\pi, e}$.   Next, we prove the asymptotic equicontinuity of $\G_n^{\pi, e}$.   Define $\rho(f, g) = \P_0(f - g)^2$ and ${\cal F}_\delta = \{f - g: f, g \in {\cal F},  \rho(f - g) \leq \delta\}$.  First,  ${\cal F}$ is totally bounded by the metric $\rho$ given \cite[Problem 2.1.2]{van1996weak} and $\|\P_0\|_{{\cal F}} < \infty$.   Next,
\begin{align*}
\|\G^{\pi, e}_n\|_{{\cal F}_\delta} \leq \|\G_n^\pi\|_{{\cal F}_\delta} + \|\G_n^{\pi, e} - \G_n^\pi\|_{{\cal F}_\delta}
\end{align*} 
For the first term on the right-hand side, $\G^\pi_n$ is asymptotically equicontinuous with respect to $\rho$. 
For the second term, 
\begin{align*}
\|\G_n^{\pi, e} - \G_n^{\pi}\|_{{\cal F}_\delta} \leq \|\bQ_e^T\|_{{\cal F}_\delta} \sqrt{n} \|\hat{\bgamma}_n - \bgamma_0\| + o_P^*(1)
\end{align*}
and 
\begin{align*}
\|\bQ_e^T\|_{{\cal F}_\delta} = \sup_{h \in {\cal F}_\delta} \left\|  \E [I(Q = 0) (1 - \pi_{\bgamma_0}(\bX)) h(\bX, \tilde{T}, \Delta) \tX^T
\right\| \leq C \sup_{h \in {\cal F}_\delta}( \P (h^2))^{1/2} \leq C \delta
\end{align*}
by the definition of $h$ and some constant $C > 0$.  Thus, we have $\lim_{\delta \searrow 0} \|\G_n^{\pi, e}\|_{{\cal F}_\delta} = 0$.   Thus, by Theorem 1.5.7 of \cite{van1996weak}, $\G^{\pi, e}_n \rightsquigarrow \G(g_1\cdot - g_2 \bQ_e(\cdot)^T \bg_3)$ in $l^\infty({\cal F})$  and 
\begin{align*}
\Var(\G(g_1\cdot f - g_2 \bQ_e(f)^T \bg_3)) & = \Var(f(\bX, \tilde{T}, \Delta)) \\
& + \E\left[f(\bX, \tilde{T}, \Delta)^2 \frac{I(Q = 0)[1 - \pi_{\bgamma_0}(\bX)]}{\pi_{\bgamma_0}(\bX)}\right] - \bQ_e(f)^T\bSigma_{\bgamma_0}^{-1}\bQ_e(f),
\end{align*}
which completes the proof.
\end{proof}
Before proving \thref{thm::thm_Cox_right_asn_d}, we first introduce a useful lemma.

\begin{lemma} \thlabel{lem::uniform_lemma}
Under assumption (D2-3),  for a small compact set $\mathbb{B}$ that contains $\bbeta_0^*$, 
$$
\sup_{t \in [0, \tau_2], \bbeta \in \mathbb{B}} \|\bS_{n, w}^{(k)}(\bbeta, t) - \bs^{(k)}(\bbeta, t)\| = o_P(1)
$$
for $k = 0, 1, 2$. 
\end{lemma}
\begin{proof}[ of \thref{lem::uniform_lemma}]
By assumption (D2),  $\bX(t)$ can be written as the difference of two nondecreasing processes of $t \in [0, \tau_2]$, then by Example 2.11.16 of \cite{van1996weak} and the fact that $\bX(t)$ is bounded,  $\G_n \bX \rightsquigarrow \G_0$ in $l^{\infty}([0, \tau_2])$.  Now we can view ${\cal G} = \{\bx(t):  t \in [0, \tau_2]\}$ as a function class $\{ f_t(x) = \bx(t);  t \in [0, \tau_2]\}$.  ${\cal G}$ is a $\P_0$-Donsker class.  Next, $\{\bbeta:  \bbeta \in \mathbb{B}\}$ is trivially a $\P_0$-Donsker class and $\{\bx(t)^T \bbeta:  t \in [0, \tau_2], \bbeta \in \mathbb{B}\}$ is also a Donsker class by theorem 2.10.6 of \cite{van1996weak}.  Similarly,  we can prove that $\{y(t) \exp(\bbeta^T \bx(t)) \bx(t)^{\otimes k}:  \bbeta \in \mathbb{B}, t \in [0, \tau_2]\}$ is also a Donsker class for $k = 0, 1, 2$  as $\{y(t):  t \in [0, \tau_2]\}$ is also a $\P_0$-Donsker class by theorem 2.11.16 of \cite{van1996weak}.  Then the result can be proved by \thref{thm::thm_gc}.  
\end{proof}

\begin{proof}[ of \thref{thm::thm_Cox_right_asn_d}]
The proof of the asymptotic linear expansion is inspired by \cite{lin1989robust}.
By \thref{lem::uniform_lemma}, we have
$$
\sup_{\bbeta \in \mathbb{B}, t \in [0, \tau_2]} \|\bS_{n, w}^{(k)}(\bbeta, t) - \bs^{(k)}(\bbeta, t)\| = o_P(1) \qquad k = 0, 1, 2
$$
for a compact set $\mathbb{B}$ that contains $\bbeta_0^*$. 
Then we can decompose the partial score in equation 
(8) of the main paper
as follows:
\begin{equation}
\begin{aligned}
\sqrt{n} \hat{\bU}_n(\bbeta) &= \frac{1}{\sqrt{n}} \sum_{i=1}^n \Delta_i \hat{w}_i \left\{\bX_i(\tilde{T}_i) -  \frac{\bS_{n, w}^{(1)}(\bbeta, \tilde{T}_i)}{\bS_{n, w}^{0}(\bbeta, \tilde{T}_i)}\right\} \\
&= \frac{1}{\sqrt{n}} \sum_{i=1}^n \hat{w}_i \int^{\tau_2}_0 \bX_i(t) d N_i(t)  -\sqrt{n}\int^{\tau_2}_0 \frac{\bS_{n, w}^{(1)}(\bbeta, t)}{\bS_{n, w}^{(0)}(\bbeta, t)} d \bar{N}(t)\\
& = \underbrace{\frac{1}{\sqrt{n}} \sum_{i=1}^n \hat{w}_i \int^{\tau_2}_0 \bX_i(t) dN_i(t)}_{(I)} -\underbrace{\sqrt{n} \int^{\tau_2}_0 \frac{\bs^{(1)}(\bbeta, t)}{\bs^{(0)}(\bbeta, t)} d[\bar{N}(t) - \tilde{N}(t)]}_{(II)} \\
& \underbrace{- \sqrt{n} \int^{\tau_2}_0 \frac{\bS_{n, w}^{(1)}(\bbeta, t)}{\bS_{n, w}^{(0)}(\bbeta, t)} d \tilde{N}(t)}_{(III)} - \underbrace{\sqrt{n} \int^{\tau_2}_0 \left[
\frac{\bS_{n, w}^{(1)}(\bbeta, t)}{\bS_{n,w}^{(0)}(\bbeta, t)} - \frac{\bs^{(1)}(\bbeta, t)}{\bs^{(0)}(\bbeta, t)}
 \right] d[ \bar{N}(t) - \tilde{N}(t) ]}_{(IV)}
\end{aligned}
\label{eqn::eqn_u_decomposition}
\end{equation}

Term (I) and (II) are already linear expansions so we do not need to conduct any further derivation.
In what follows, we will first show that term (IV) is $o_P(1)$ and then argue that term (III)
has an asymptotic linear expansion.

{\bf Term (IV).}
By assumption (D4),  for large enough $n$, both $\bS_{n, w}^{(0)}(\bbeta, t)$ and $\bs^{(0)}(\bbeta, t)$ are bounded away from 0, so \thref{lem::uniform_lemma} implies that
$$
\sup_{\bbeta \in \mathbb{B}, t \in [0, \tau_2]} \left \| \frac{\bS_{n, w}^{(1)}(\bbeta, t)}{\bS_{n, w}^{(0)}(\bbeta, t)} - \frac{\bs^{(1)}(\bbeta, t)}{\bs^{(1)}(\bbeta, t)}
\right\| = o_P(1)
$$
Moreover,
by \thref{ipw_weak_converge},  $n^{1/2}(\bar{N}(\tau_2) - \tilde{N}(\tau_2))$ converges to a mean zero normal random variable. Together, this implies that (IV) is $o_P(1)$.  

{\bf Term (III).}
Next,  we have that 
\begin{align*}
\frac{\bS_{n,w}^{(1)}(\bbeta, t)}{\bS_{n, w}^{(0)}(\bbeta, t)} &= \frac{\bS_{n, w}^{(1)}(\bbeta, t)}{\bs^{(0)}(\bbeta, t)\left[1 + \frac{\bS_{n, w}^{(0)}(\bbeta, t) - \bs^{(0)}(\bbeta, t)}{\bs^{(0)}(\bbeta, t)}\right]} = \frac{\bS_{n,w}^{(1)}(\bbeta, t)}{\bs^{(0)}(\bbeta, t)} \left[ 1 - \frac{\bS_{n, w}^{(0)}(\bbeta, t) - \bs^{(0)}(\bbeta, t)}{\bs^{(0)}(\bbeta, t)} + o_P(1) \right] \\
& = \frac{1}{\bs^{(0)}(\bbeta, t)} \left[ \bS_{n, w}^{(1)}(\bbeta, t) - \frac{\bs^{(1)}(\bbeta, t)}{\bs^{(0)}(\bbeta, t)} \{\bS_{n, w}^{(0)}(\bbeta, t) - \bs^{(0)}(\bbeta, t)\}\right] + o_P(1) 
\end{align*}
Thus, 
$$
n^{1/2} \int^{\tau_2}_0 \frac{\bS_{n, w}^{(1)}(\bbeta, t)}{\bS_{n, w}^{(0)}(\bbeta, t)} d \tilde{N}(t) = n^{1/2} \int^{\tau_2}_0 \frac{1}{\bs^{(0)}(\bbeta, t)} \left[
\bS_{n, w}^{(1)}(\bbeta, t) - \frac{\bs^{(1)}(\bbeta, t)}{\bs^{(0)}(\bbeta, t)}\{\bS_{n, w}^{(0)}(\bbeta, t) - \bs^{(0)}(\bbeta, t)\} d \tilde{N}(t)
\right] + o_P(1)
$$
Taking above equation back to equation \eqref{eqn::eqn_u_decomposition}, we get the desired asymptotic linear expansion of $\hat{\bU}_n(\bbeta)$: $\sqrt{n} \hat{\bU}_n(\bbeta) = n^{-1/2}\sum_{i=1}^n \hat{w}_i \bU_i(\bbeta) + o_p(1)$ with 
$$
\bU_i(\bbeta) = \int^{\tau_2}_0 \left[ \bX_i(t) - \frac{\bs^{(1)}(\bbeta, t)}{\bs^{(0)}(\bbeta, t)}\right]dN_i(t) - \int^{\tau_2}_0 \frac{Y_i(t) \exp(\bbeta^T \bX_i(t))}{\bs^{(0)}(\bbeta, t)} \left[ \bX_i(t) - \frac{\bs^{(1)}(\bbeta, t)}{\bs^{(0)}(\bbeta, t)}\right] d\tilde{N}(t) 
$$
Finally,  to apply \thref{ipw_weak_converge},  note that $\pi_{\bgamma_0}(\bX)$ needs to be replaced by $\pi_{\bgamma_0}(\bZ_1)$ and 
$\bU_i(\bbeta)$ can be viewed as a function $\eta_{\bbeta}(\bar{\bX}(\tilde{T}), \tilde{T}, \Delta)$. 
As a result,  we have $\sqrt{n} \hat{\bU}_n(\bbeta) = n^{-1/2}\sum_{i=1}^n \hat{w}_i \bU_i(\bbeta) + o_p(1) = \G^{\pi, e}_n \eta_{\bbeta} + o_P(1)$. 
\thref{ipw_weak_converge} implies
the asymptotic normality of $\sqrt{n} \hat{\bU}_n(\bbeta_0^*)$\footnote{we only need the pointwise convergence result at $\bbeta = \bbeta_0^*$}, which completes the proof.  
\end{proof}

\begin{proof}[ of 
Theorem 3.5
]
Based on the fact that $\hat \bU_n(\hat \bbeta_n) = 0$ and $\bU_0(\bbeta_0^*)=0$,  we have 
$$
\sqrt{n}(\hat{\bbeta}_n - \bbeta_0^*) = \bA_n(\bbeta_0^*)^{-1} n^{1/2} \widehat{\bU}_n(\bbeta_0^*) + o_P(\sqrt{n}\|\hat{\bbeta}_n - \bbeta_0^*\|),
$$
where $\bA_n$ is defined in equation \eqref{eq::An}.   To prove that $\hat{\bbeta}_n \rightarrow_p \bbeta_0^*$, we adopt the same strategy as Lemma 3.1 in \cite{andersen1982Cox}.  From \thref{thm::thm_Cox_right_asn_d},  we obtained that $\frac{1}{n} \hat{\bU}_n(\bbeta) = \frac{1}{n} \sum_{i=1}^n \hat{w}_i \bU_i(\bbeta) + o_P(n^{-1/2})$.  Then by \thref{thm::thm_gc}, 
\begin{align*}
\frac{1}{n} \hat{\bU}_n(\bbeta) \rightarrow_p \E\left\{ \int^{\tau_2}_0 \left[ \bX(t) - \frac{\bs^{(1)}(\bbeta, t)}{\bs^{(0)}(\bbeta, t)}
 \right] dN(t) - \int^{\tau_2}_0 \left[ \bX(t) - \frac{ \bs^{(1)}(\bbeta, t)}{ \bs^{(0)}(\bbeta, t)}\right] \frac{Y(t) \exp(\bbeta^T \bX(t))}{\bs^{(0)}(\bbeta, t)} d \tilde{N}(t) \right\}
\end{align*}
It is not hard to prove that above expectation equals to $\bU_0(\bbeta) = \E\left[ \Delta \left( \bX(\tilde{T}) - \frac{\bs^{(1)}(\bbeta, \tilde{T})}{\bs^{(0)}(\bbeta, \tilde{T})} \right)\right]$.   By assumption (D5), $\bA(\bbeta_0^*) = -\frac{\partial \bU_0(\bbeta)}{\partial \bbeta} \big |_{\bbeta =\bbeta_0^*}$ is positive definite and by a similar argument as Lemma 3.1 in \cite{andersen1982Cox},  we have $\hat{\bbeta}_n \rightarrow_p \bbeta_0^*$.   We now prove the uniform convergence of $\bA_n(\bbeta)$ to $\bA(\bbeta)$ in a small compact subset $\mathbb{B}$ that contains $\bbeta_0^*$. 
\begin{equation}
\begin{aligned}
& \sup_{\bbeta \in \mathbb{B}}\|\bA_n(\bbeta) - \bA(\bbeta)\|  \leq \int^{\tau_2}_0 \sup_{\bbeta \in \mathbb{B}}\left\| \frac{\bS_{n, w}^{(2)}(\bbeta, t)}{\bS_{n, w}^{(0)}(\bbeta, t)} - \left(\frac{\bS_{n, w}^{(1)}(\bbeta, t)}{\bS_{n, w}^{(0)}(\bbeta, t)}\right)^{\otimes 2}\right \| |d(\bar{N}(t) - \tilde{N}(t))| \\
& + \int^{\tau_2}_0  \sup_{\bbeta \in \mathbb{B}} \left\| \left \{ \frac{\bS_{n, w}^{(2)}(\bbeta, t)}{\bS_{n, w}^{(0)}(\bbeta, t)} - \left(\frac{\bS_{n, w}^{(1)}(\bbeta, t)}{\bS_{n, w}^{(0)}(\bbeta, t)}\right)^{\otimes 2}\right\} - \left\{ \frac{\bs^{(2)}(\bbeta, t)}{\bs^{(0)}(\bbeta, t)} - \left(\frac{\bs^{(1)}(\bbeta, t)}{\bs^{(0)}(\bbeta, t)}\right)^{\otimes 2}\right\} \right\|d \tilde{N}(t)
\end{aligned}
\label{eqn::difference_p_s_d}
\end{equation}
For the first term of \eqref{eqn::difference_p_s_d},  let 
$$
h_n(t) = \sup_{\bbeta \in \mathbb{B}}\left\| \frac{\bS_{n, w}^{(2)}(\bbeta, t)}{\bS_{n, w}^{(0)}(\bbeta, t)} - \left(\frac{\bS_{n, w}^{(1)}(\bbeta, t)}{\bS_{n, w}^{(0)}(\bbeta, t)}\right)^{\otimes 2}\right \|$$ and 
$$
h(t) = \sup_{\bbeta \in \mathbb{B}}\left\| \frac{\bs^{(2)}(\bbeta, t)}{\bs^{(0)}(\bbeta, t)} - \left(\frac{\bs^{(1)}(\bbeta, t)}{\bs^{(0)}(\bbeta, t)}\right)^{\otimes 2}\right \|
$$. We can replace $h_n(t)$ by $h(t)$ such that 
$$
\int^{\tau_2}_0 |h_n(t)| |d(\bar{N}(t) - \tilde{N}(t))| = \int^{\tau_2}_0 |h(t)| |d(\bar{N}(t) - \tilde{N}(t))| + o_P(1)
$$
by \thref{lem::uniform_lemma}.  By \thref{thm::thm_gc}, $\sup_{t \in [0, \tau_L]}|\bar{N}(t) - \tilde{N}(t)| = o_P(1)$.   Thus the first term is $o_P(1)$.  The second term is also $o_P(1)$ by \thref{lem::uniform_lemma}.  By the uniform convergence of $\bA_n(\bbeta)$ to $\bA(\bbeta)$, we then have $\bA_n(\bbeta_0^*) \rightarrow_p \bA(\bbeta_0^*)$. 
Then, by slutsky's theorem, we obtained that 
$$
\sqrt{n}(\hat{\bbeta}_n - \bbeta_0^*) \rightarrow_d \mathbf{N}(0, \bSigma_0^{-1} \bSigma_U \bSigma_0^{-1})
$$
\end{proof}

\section{Linkage assumptions} \label{sec::linkage_assumption_for_naive_method}
%
We now discuss some other potential assumptions for linkage. 
For an alternative approach,  we might assume that 
$$
L \indep (T, C) | \bX.
$$
However,  the IPW type method for this assumption suffers from the same issue as the complete-case analysis in that unlinked participants that are diagnosed with PC within the clinical trial will not be included in analysis. 

As our main goal is to deal with the missing survival outcome $T$ and $C_2$ and the missingness only happens when a participant is not linked and censored in the clinical trial, an alternative approach would be to directly model $P(L = 0, Q = 0 | \bX, T, C_1, C_2)$ and a MAR type assumption would be 
$$
P(L = 0, Q = 0 | \bX, T, C_1, C_2) = P(L = 0, Q = 0 | \bX)
$$
since only $\bX$ is always observed.  However,  this MAR assumption would never hold as we always have $Q = 0$ when $T \geq C_1$.  Thus,  we choose to model linkage alone as the CLAR assumption.

\section{More details on augmented inverse probability of linkage weighting estimation}  \label{sec::doubly_robust}
Now we propose an augmented inverse probability of linkage weighting (AIPLW) estimator for estimating the long-term effect. We first need to define several outcome regression functions:
\begin{align*}
m_0(\bX) & = \E\left[  \Delta \left(\bX - \frac{s^{(1)}(\tilde{T}, \bbeta)}{s^{(0)}(\tilde{T}, \bbeta)} \right)
\middle| \bX, Q = 0, L = 1  \right]  \\
& =  \underbrace{\E\left[ \Delta \middle| \bX, Q = 0, L = 1\right] \bX}_{m_1(\bX)} -
\underbrace{\E \left[ \Delta \frac{s^{(1)}(\tilde{T}, \bbeta)}{s^{(0)}(\tilde{T}, \bbeta)} \middle | \bX,  Q = 0, L = 1 \right]}_{m_2(\bX)} \\ 
m_3(\bX, t) & = \E\left[ I(\tilde{T} \geq t) \exp(\bX^T \bbeta) | \bX, Q = 0,  L = 1 \right]  = \E\left[ I(\tilde{T} \geq t) | \bX, Q = 0,  L = 1 \right] \exp(\bX^T \bbeta)
\end{align*}
The augmented IPW partial score is then as following:
\begin{equation}
\begin{aligned}
U_{n, AIPLW}  &= \frac{1}{n} \sum_{i=1}^n \left\{
I(Q_i = 1) \Delta_i \left[ \bX_i - \frac{S_{n, DR}^{(1)}(\bbeta, \tilde{T}_i)}{S_{n, DR}^{(0)}(\bbeta, \tilde{T}_i)} \right] 
+ I(Q_i = 0) \left[ \frac{I(L_i = 1)}{\pi_{\hat{\gamma}_n}(\bX_i)} \Delta_i  \left[ \bX_i - \frac{S_{n, DR}^{(1)}(\bbeta, \tilde{T}_i)}{S_{n, DR}^{(0)}(\bbeta, \tilde{T}_i)} \right]  \right. \right. \\
& \left. \left.  + \left(1 - \frac{I(L_i =  1)}{\pi_{\hat{\gamma}_n}(\bX_i)}\right) \hat{m}_0(\bX_i)\right]
\right\} 
\end{aligned}
\label{eqn::eqn_aipw_ee}
\end{equation}
with 
\begin{align*}
S^{(k)}_{n, DR}(\bbeta, t) & = \frac{1}{n} \sum_{i=1}^n \left\{ I(Q_i = 1) I(\tilde{T}_i \geq t) \exp(\bX_i^T \bbeta)  + I(Q_i = 0) \left[
\frac{I(L_i = 1)}{\pi_{\hat{\gamma}_0}(\bX_i)}I(\tilde{T}_i \geq t) \exp(\bX_i^T \bbeta)  \right. \right. \\
& \left. \left. +  \left(1 - \frac{I(L_i =  1)}{\pi_{\hat{\gamma}_n}(\bX_i)}\right) \hat{m}_3(\bX_i, t)
\right] \right\} \bX_i^{\otimes k}
\end{align*}
where $\hat{m}_k(\mathbf{x})$ are certain estimators of $m_k(\mathbf{x})$ for $k = 1, 2, 3$ \footnote{We use $\hat m_3(\bx)$ to represent $\hat m_e(\bx, t)$ for simplicity} such that $\hat{m}_0(\mathbf{x}) = \hat{m}_1 (\mathbf{x}) - \hat{m}_2(\mathbf{x})$.  Let $\hat{\bbeta}_{AIPW}$ be the solution to $U_{n, AIPW} = \bm{0}$. $\hat{\bbeta}_{AIPW}$ is a doubly robust estimator in the sense that either $\pi_{\hat{\gamma}_n}(\mathbf{x})$ being consistent for $\pi_0(\mathbf{x}) = P(L = 1 | \mathbf{x},  Q = 0)$ or $\hat{m}_k(\mathbf{x})$ being consistent for $m_k(\mathbf{x})$ for $k = 1, 2, 3$ will guarantee that $\hat{\bbeta}_{AIPW}$ is a consistent estimator for $\bbeta_0^*$. Informally, to prove the doubly-robust property of $\hat{\bbeta}_{AIPW}$,  we need to prove that the population version of the estimating equation \eqref{eqn::eqn_aipw_ee} have the same root $\bbeta_0^*$ as the IPLW estimation equation.  

We first argue that $S_{n, DR}^{(k)}(\bbeta, t)$ is a doubly-robust estimator of $s_{(k)}(\bbeta, t)$ in the above sense.  It is not hard to see that $S_{n, DR}^{(k)}(\bbeta, t)$ converges to $s_{DR}^{(k)}(\bbeta, t)$ with
\begin{align*}
s^{(k)}_{DR}(\bbeta, t) = & \E \left\{  \left(
I(Q = 1) I(\tilde{T} \geq t) \exp(\bX^T \bbeta) + I(Q = 0) \left[ \frac{I(L = 1)}{\pi_{\gamma_0}(\bX)} I(\tilde{T} \geq t) \exp(\bX^T \bbeta) \right. \right. \right.\\
& \left. \left. \left. + \left(1 - \frac{I(L = 1)}{\pi_{\gamma_0}(\bX)} \right) m_3^*(\bX, t) \right] \right) \bX^{\otimes k}
\right\}
\end{align*}
where $\hat{m}_3(\mathbf{x}, t) \rightarrow_p m_3^*(\mathbf{x}, t)$ and $\pi_{\hat{\gamma}_0}(\mathbf{x}) \rightarrow \pi_{\gamma_0}(\mathbf{x})$.  $m_3^*(\mathbf{x}, t)$ is not necessarily $m_3(\mathbf{x}, t)$ and similarly $\pi_{\gamma_0}(\mathbf{x})$ might not be $\pi_0(\mathbf{x})$. Then when $\pi_0(\mathbf{x}) = \pi_{\gamma_0}(\mathbf{x})$ and $m_3^*(\mathbf{x}, t) \neq m_3(\mathbf{x}, t)$,  we have
\begin{align*}
s_{DR}^{(k)}(\bbeta, t)  &= \E \left\{
\left[I(Q = 1) I(\tilde{T} \geq t) \exp(\bX^T \bbeta) + I(Q = 0) \frac{I(L = 1)}{\pi_0(\bX)}I(\tilde{T} \geq t) \exp(\bX^T \bbeta) \right] \bX^{\otimes k}
\right\} \\
& = s^{(k)}(\bbeta, t) 
\end{align*} 
according to Proposition 4.2 and assumption (A1).  On the other hand,  when $\pi_{\gamma_0}(\mathbf{x}) \neq \pi_0(\mathbf{x})$ and $m_3^*(\mathbf{x}, t) = m_3(\mathbf{x}, t)$,  we can rewrite $s_{DR}^{(k)}(\bbeta, t)$ as 
\begin{align*}
& \E \left\{  \left(I (Q = 1) I(\tilde{T} \geq t) \exp(\bX^T \bbeta) + \underbrace{I(Q = 0) \frac{I(L = 1)}{\pi_{\gamma_0}(\bX)} \left(
I(\tilde{T} \geq t) \exp(\bX^T \beta) - m_3(\bX, t) \right)}_{\mathbf{I}} \right. \right. \\
& \left. \left. + I(Q = 0)m_3(\bX, t) \right) \bX^{\otimes k} \right\}
\end{align*}
and term $\mathbf{I}$ is 0 by law of total expectation.  Thus,  again we have $s_{DR}^{(k)}(\bbeta, t) = s^{(k)}(\bbeta, t)$. 
Similarly,  it is not hard to prove that the population version of the above estimating equation \eqref{eqn::eqn_aipw_ee}
\begin{align*}
& U_{AIPW}  = \E \left\{
I(Q = 1) \Delta \left[ \bX - \frac{s_{DR}^{(1)}(\bbeta, \tilde{T})}{s_{DR}^{(0)}(\bbeta, \tilde{T})} \right] 
+ I(Q = 0) \left[ \frac{I(L = 1)}{\pi_{\gamma_0}(\bX_i)} \Delta \left[ \bX - \frac{s_{DR}^{(1)}(\bbeta, \tilde{T})}{s_{DR}^{(0)}(\bbeta, \tilde{T})} \right]  + \right. \right. \\
& \left. \left.  \left(1 - \frac{I(L =  1)}{\pi_{\gamma_0}(\bX)}\right) m_0^*(\bX)\right]
\right\} 
\end{align*}
where $\hat{m}_0(\mathbf{x}) \rightarrow m_0^*(\mathbf{x})$. Then when we have $\pi_{\gamma_0}(\mathbf{x}) = \pi_0(\mathbf{x})$,   $U_{AIPW}$ becomes
$$
\E \left\{ \left[I(Q = 1) + I(Q = 0) \frac{I(L = 1)}{\pi_{\gamma_0}(\bX)} \right] \Delta \left[ \bX - \frac{s^{(1)}(\bbeta, \tilde{T})}{s^{(0)}(\bbeta, \tilde{T})} \right] \right\}
$$
which is the same as the IPLW estimating equation according to 
Proposition 3.3
and assumption (A1).  On the other hand,  if we have  $\pi_{\gamma_0}(\mathbf{x}) \neq \pi_0(\mathbf{x})$ and $m_k^*(\mathbf{x}) = m_k(\mathbf{x})$ for $k = 1, 2, 3$, then using the same argument as above, we have 
 that $U_{AIPW}$ becomes
\begin{align*}
E\left\{ \Delta   \left[ \bX - \frac{s^{(1)}(\bbeta, \tilde{T})}{s^{(0)}(\bbeta, \tilde{T})} \right]
\right\}
\end{align*}
which is just the original partial likelihood estimating equation for Cox model.  Thus,  $\hat{\bbeta}_{AIPW}$ is doubly-robust. 

\subsection{Difficulty of construction of doubly-robust estimator} \label{sec::subsec::aipw_difficulty}
Based on \eqref{eqn::eqn_aipw_ee},  we need to estimate $m_k(\mathbf{x})$ for $k = 1, 2, 3$ to solve for $\hat{\bbeta}_{AIPW}$.  However, there are major difficulties with estimating these three regression functions. 
To estimate $m_k(\mathbf{x})$, we need to estimate
\begin{align*}
m_1(\bX) &= \E[\Delta |\bX; Q = 0, L = 1]\bX = P(T \leq C | X,  Q = 0, L = 1) \bX \\
m_2(\bX) & = \E \left[\Delta \frac{s^{(1)}(\tilde{T}, \bbeta)}{s^{(0)}(\tilde{T}, \bbeta)} \middle | \bX, Q = 0, L = 1 \right] \\
m_3(\bX, t) & = P(\tilde{T} \geq t | \bX; Q = 0, L = 1) \exp(\bX^T \bbeta)
\end{align*}
We have a couple modeling strategies.  We use $m_1(\mathbf{x})$ as an example to illustrate the modeling details. For the first modeling strategy,  we can try to estimate $m_k(\mathbf{x})$ through modeling the distribution of $C_1, C_2, T$.  More specifically, for $P(T \leq C |\bX,  Q = 0, L = 1)$,  it is not hard to get that 
\begin{align*}
P(T \leq C | \bX, Q = 0; L = 1) = \frac{P(T \leq C_2, T \geq C_1 | \bX, L = 1)}{P(T \geq C_1 | \bX, L = 1)}
\end{align*}
For the numerator,  we have 
\begin{align*}
P(T \leq C_2, T \geq C_1 | \bX, L = 1) = \int^{\tau_2}_0 P(C_2 \geq s,  C_1 \leq s | T = s, \bX, L = 1) f_T(s | \bX, L = 1) ds
\end{align*}
Thus,  we need to model the joint distribution of $C_1$ and $C_2$ given $T, \bX$ and $L = 1$.   Note here with the independent censoring assumption,  we have $(C_1,  C_2) \indep T | \bX$.  Together with assumption (A1),  it is not clear if we also have $(C_1,  C_2) \indep T | \bX, L = 1$.  Next, even if we make the further assumption that $(C_1, C_2) \indep T | \bX, L = 1$. It still requires us to model the joint distribution of $C_1$ and $C_2$ given $\bX$ and $L = 1$.  Another thing is that this also requires us to model the distribution of failure time $T$ given $\bX$ and $L = 1$, which need careful modeling to avoid model conflict with the Cox model for $T$ given $\bX$ alone.   Similar modelings are required for estimating $m_2(\bX)$ and $m_3(\bX)$. 

For a second modeling strategy,  we can directly model $P(\Delta = 1 | \bX, Q = 0, L = 1)$ with a logistic regression as $\Delta$ is binary variable. Similarly,  we can estimate $m_3(\bX)$ by directly modeling the distribution of the observed time $\tilde{T}$ given $\bX$, $Q = 0$ and $L = 1$ through Cox regression.  To estimate $m_2(\bX)$,  we need to either model the distribution of $\tilde{T}$ given $\Delta = 1$, $Q = 0$ and $L = 1$ or model the distribution of $\Delta$ given $T, \bX, Q = 0$ and $L = 1$.  Take all things into consideration,  very careful modelings need to be carried out to ensure model congeniality and whether such models exist is not clear to us.

Finally,  to avoid the potential model congeniality issue,  non-parametric estimation technique might be applied. However, nonparametric estimation in general suffers from the curse of dimensionality issue, which might require a very large number of samples to get a good estimate.

\section{Relaxation of the ``no gap'' assumption}  \label{sec::sec_gap_relaxation}
So far we have made the ``no gap'' assumption to focus on the right censoring problem.  Now we consider relaxations of this assumption as it is quite common that a participant might not be under observation for some time in practice.  
This allows for the possibility of interval censoring as a participant might be diagnosed with the event of interest during the gap when he is not under observation.  Further,  this creates a situation that we have both right censored and interval-censored data, which is also known as partly interval-censored data \citep{turnbull1976empirical}.  

Partly interval-censored data for Cox regression has been studied in \cite{kim2003maximum},\cite{cai2003hazard} and the estimation is more difficult than right-censored data.  For simplicity,  we do not deal with interval-censoring in the current paper and leave that to future work.  Instead we consider an alternative approach that transforms the interval-censored data to right-censored data. 
This approach is in the same spirit as the NLAC approach.  
However, one has to be careful with the transformation. We first discuss an intuitive but problematic approach.

\subsection{A problematic approach}
For illustration, we consider the oracle setting such that each participant is linked to the observational follow-up datasets.  For participants that are known to be interval-censored
during the gap between clinical trial and observational follow-up,  we treat such participants as being right censored at the last recorded date of clinical trial.  Thus,  we transform the partly interval-censored problem to a right-censored only problem.  On the other hand,  for participants with gaps,  it is also possible that they might not be interval-censored.  It is then tempting to use their survival information in the observational dataset, i.e,  failure time $T$ or the censoring time $C_2$.  However,  this approach is problematic as this would lead to biased estimates. To see the effect of bias with this approach empirically,  we conducted a simulation study\footnote{The detailed simulation setting is provided in the appendix \ref{sec::simulation_results_ti_gap}}. with approximately 4.5\% of the participants being interval-censored.  The coverage of the 95\% confidence interval for parameter $\beta_1$ is only about 65\% with $n = 10,000$ and 1,000 repetitions. 

To see why we cannot use the survival information in the observational follow-up dataset for a participant with gap and not interval-censored,  we need to think about the corresponding censoring distribution. Considering participants with gaps, 
effectively the censoring time $C$ is set as 
$$C = \begin{cases} C_1  &  \text{if } C_1 < T < C_1 + U \\ \max(C_1, C_2) &  \text{if } C_1 + U \leq T \end{cases},
$$
where $U$ is a random variable for the length of the gap between the clinical trial and the start of observational follow-up dataset. Thus,  it is clear that the censoring time $C$ now depends on the failure time $T$ and violates the independent censoring assumption.

\subsection{A remedy } 
We now propose a remedy approach that properly transforms the partly interval-censored data to right-censored data and conventional statistical software can then be applied to estimate the parameter for Cox models. 
Again we consider the oracle setting that each participant is linked to the observational follow-up dataset. 
For participants with gaps and censored in the clinical trial,  we simply view them as being right censored at the last recorded date of clinical trial. Thus, we set $C = C_1$ whenever there is a gap between a participant's last recorded date in clinical trial and the start time of observational follow-up.  
Let $G$ denotes whether gap exists for a participant.  
Equivalently, the censoring time $C$ is set as 
$$
C = I(G = 1) C_1 + I(G = 0) \max(C_1, C_2).
$$ 
Thus,  similar to the linkage assumption for the NLAC method,  above proposed method works if $G \indep T | \bX, C_1, C_2$, under the oracle setting that each participant is linked.

For the more practical setting with incomplete linkages,  we can similarly apply the methods developed in the current paper.  To be more specific, for NLAC, the censoring time $C$ can be written as 
$$
C = I(L = 1)[ I(G = 1) C_1 + I(G = 0) \max(C_1, C_2)] + I(L = 0) C_1
$$
Thus,  NLAC again sets the censoring time as $C_1$ for participants that are not linked and censored in the clinical trial. 
One sufficient condition for NLAC to work is 
$$
L \indep T | \bX, C_1, C_2, G,
$$
which is similar to assumption $\text{\bf (A1)}$. 
For the IPLW method,  we might modify the CLAR assumption as 
$$
P(L = 1 | \bX, Q = 0,  \tilde{T}, \Delta, G) = P(L = 1 | \bX, Q = 0, G).
$$
We present relevant simulation results in Appendix \ref{sec::simulation_results_ti} due to space limit.  The limitation of this remedy approach is similar to NLAC:  it only works when Cox model is correctly specified;  when Cox model is mis-specified,  our proposed approach will no longer work as we modify the censoring times.  

\section{More simulation results}   \label{sec::simulation_results_ti} 
\subsection{Revisit the motivating example}
We first revisit the motivating example in Section 3.1.2.  
Figure~\ref{fig::fig_motivating_example}
shows that our proposed IPLW method gives both consistent estimates and correct coverages for the 95\% confidence intervals. On the other hand, CC+, CC and NLAC all give below nominal coverages and inconsistent estimates. 
\begin{figure}[h]
  \centering
  \begin{subfigure}[t]{0.45\textwidth} 
  \includegraphics[width=\textwidth]{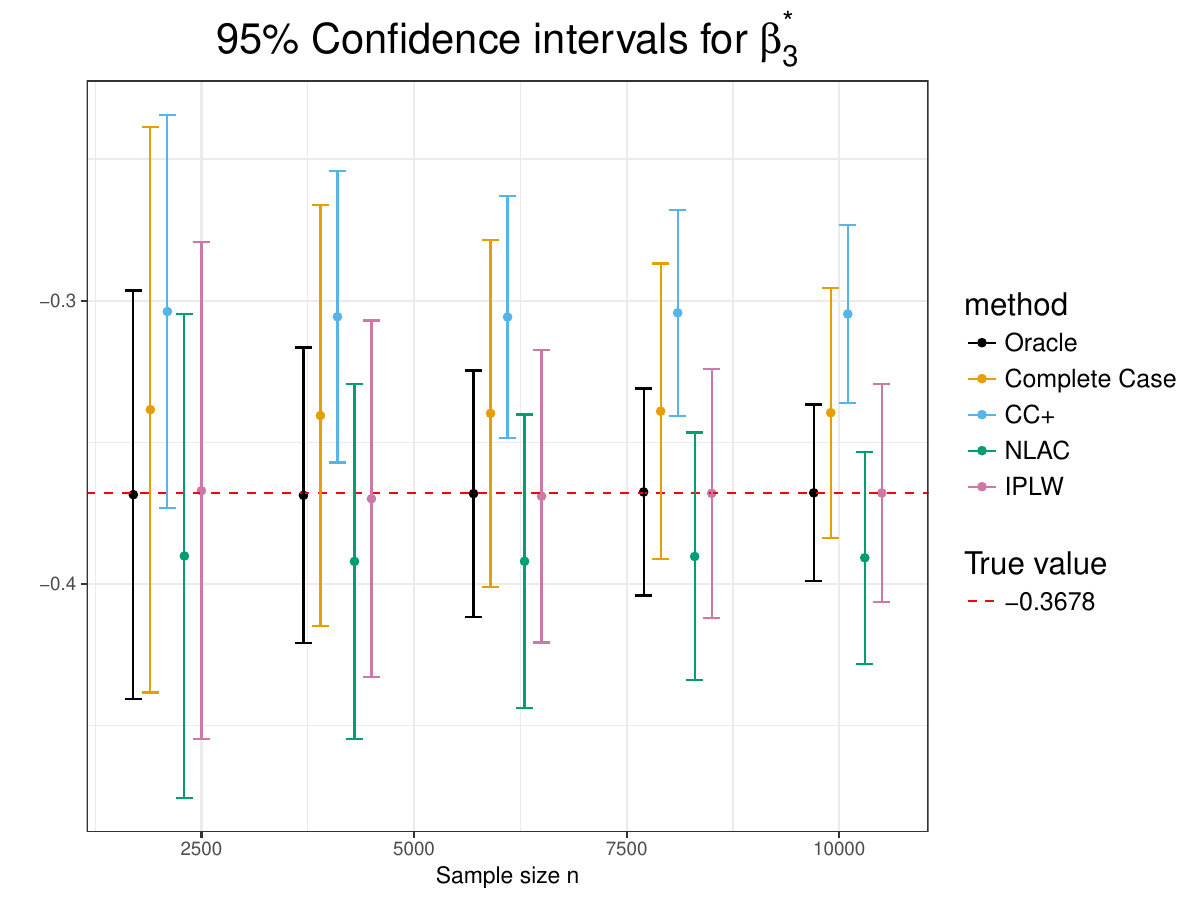}
  \label{fig::fig_motivating_example_ci}
  \end{subfigure}
  \hfill
  \begin{subfigure}[t]{0.45\textwidth} 
    \includegraphics[width=\textwidth]{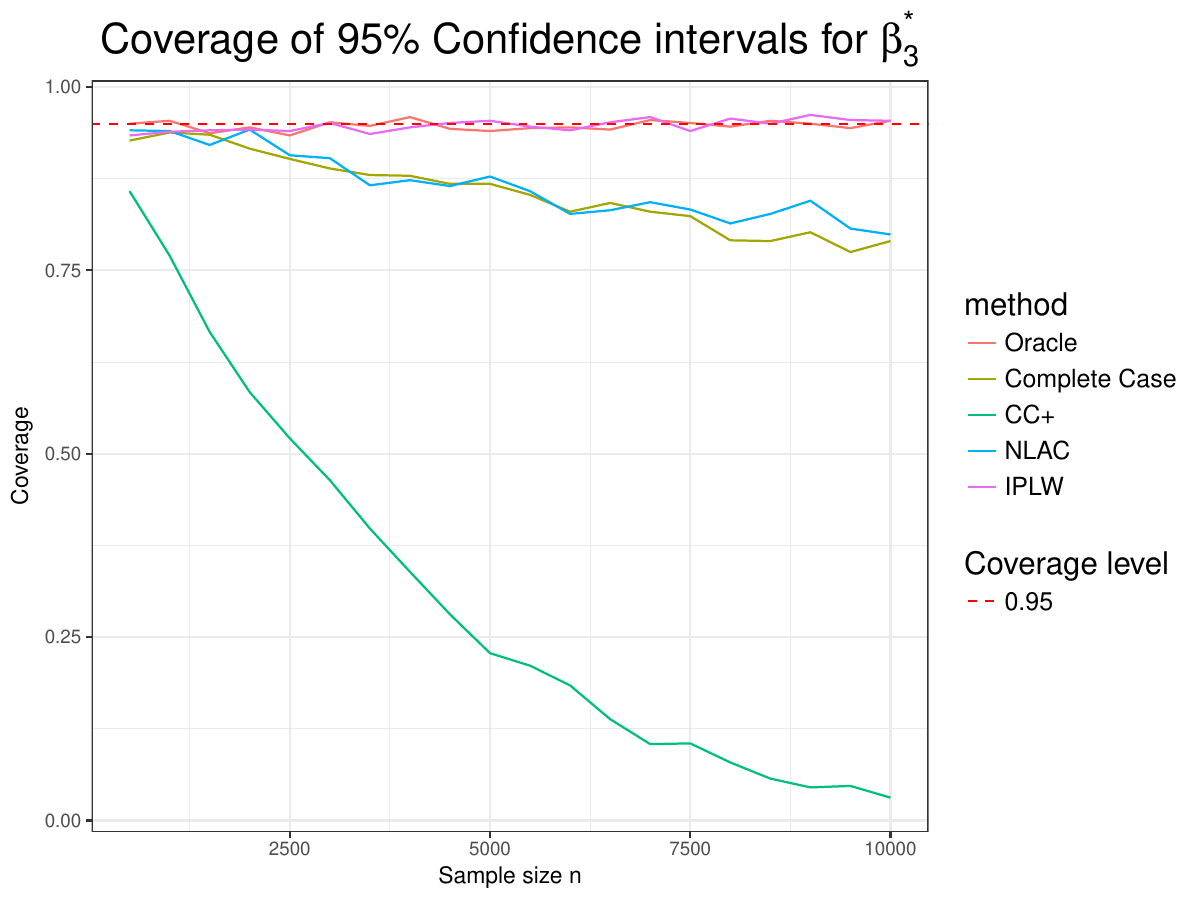}
    \label{fig::fig_motivating_example_coverage}
  \end{subfigure}
  \caption{Revisiting the motivating example in Section 3.1.2.
  Our approach and the oracle approach are the only methods leading to a valid confidence interval.}
  \label{fig::fig_motivating_example}
  \end{figure}

\subsection{IPLW under MAR}
For the simulation setting in Section 4 of the main paper, 
we also present the results when linkage satisfies the missing at random (MAR) assumption as 
$$
\log \{P(L = 1 | \bX) / P(L = 0 | \bX)\}=  -0.25 + 0.5 X_1 + 0.5 X_2.
$$
The results showed favorable performances of both NLAC and IPLW. Though as expected, NLAC yielded slightly biased estimates when the Cox model was not accurately specified.

\begin{table}[H]
\caption{Simulation results when Cox model is correctly specified and linkage satisfies MAR under $n = 2000$. }
\begin{center}
\begin{tabular}{
 | c 
 | c 
 | c
 |  c 
 |  c 
 |  c 
 | c 
 |
}
\hline
 &   \multicolumn{3}{c |}{Bias (Mean SE)} & \multicolumn{3}{c|}{Coverage of 95\% CI}  \\
 \hline
 Method              & $\beta_1$ & $\beta_2$ & $\beta_3$ & $\beta_1$ & $\beta_2$ & $\beta_3$ \\
 \hline
 Oracle & -0.00 (0.108) &  0.00 (0.033)  & -0.00  (0.140) & 0.94   & 0.95  & 0.96 \\
\hline
CC  & -0.00 (0.131) &  0.00 (0.043)  & -0.01  (0.173) & 0.95   & 0.96  & 0.96  \\
\hline
 CC+ &  -0.07 (0.109) & -0.10 (0.038)  & 0.09 ( 0.157)  & $0.92$   & $0.24^{\dagger}$   & $0.91$  \\
\hline
 NLAC & -0.00  (0.109) &  0.00 (0.037)  & -0.00 (0.157)  & 0.94   & 0.95   & 0.94 \\
\hline
 IPLW & -0.00   (0.110) & 0.00 (0.040)  & -0.00  (0.159)  & 0.94   & 0.94   & 0.95   \\
\hline
\hline
\end{tabular}
\end{center}
\footnotesize{We use $^\dagger$ to highlight settings with coverage below 90\%.}
\label{table::td_s1_lcar}
\end{table}

\begin{table}[H]
\caption{Simulation results for when Cox model is misspecified and linkage satisfies MAR under $n = 2000$.}
\begin{center}
\begin{tabular}{
 | c 
 | c 
 |c
 | c 
 | c 
|
}
\hline
  &  \multicolumn{2}{c|}{Bias (Mean SE)} & \multicolumn{2}{c|}{Coverage of 95\% CI}  \\
 \hline
  Method & $\beta_1$ & $\beta_2$ & $\beta_1$ & $\beta_2$ \\
 \hline
 Oracle & -0.00 (0.068)   &  0.00 (0.033)  & 0.93   & 0.94 \\
\hline
 CC & -0.01 (0.084)   &  0.00 (0.043)  & 0.94   & 0.96  \\
\hline
 CC+ &  -0.08 (0.075)  & -0.10 (0.038)  & $0.80^{\dagger}$   & $0.24^{\dagger}$   \\
\hline
 NLAC & -0.05  (0.076)   &  0.00 (0.037)  & $0.88^{\dagger}$   & 0.95  \\
\hline
 IPLW &  -0.00  (0.079)   & 0.00 (0.040)  & 0.95   & 0.95   \\
\hline
\hline
\end{tabular}
\end{center}
\footnotesize{We use $^\dagger$ to highlight settings with coverage below 90\%.}
\label{table::td_s2_lcar}
\end{table}

\subsection{Simulation results for AIPLW} \label{sec::aipw_simulation_results}
In this section, we illustrate the AIPLW method through a  case when there are time-independent covariates only. We consider the following data generation mechanism. The hazard function is $\lambda(t | x) = \lambda_0(t) \exp(\beta_1 X_1 + \beta_2 X_2 + \beta_3 X_3)$ and $\bbeta_0 = (\beta_1, \beta_2, \beta_3)^T = (-\log(2), \log(2), 0.2)^T$.  $X_1$ is a Bernoulli variable that takes value 1 with probability 0.5.  $X_2$ is a normal random variable with mean -1 and standard deviation $1$.  $X_3$ is a normal random variable with mean $1$ and standard deviation 2. The baseline hazard function is $\lambda_0(t) = 0.15$. The censoring time in clinical trial $C_1$ is an exponential random variable with rate $0.5 * X_1 + 0.05$.  The censoring time $C_2$ is defined as $C_1$ plus an exponential random variable with rate $0.8 * X_1 + 0.03$.   Further,  we set $\tau_1 = 7$ and $\tau_2 = 16$.  Two linkage mechanisms are under consideration, namely CLAR and missing at random (MAR).  
The CLAR is the same as in Section 4 of the main paper
and the MAR is the same as in Appendix H.2.

We obtain AIPLW estimate $\hat \bbeta$ by solving \eqref{eqn::eqn_aipw_ee}. As we discussed in Appendix \ref{sec::doubly_robust},  we need to estimate $m_k(\bx)$ for $k = 1, 2, 3$ and the linkage probability $\pi_0(\bx)$ for the AIPLW method.  We use a correctly specified logistic regression function to estimate the linkage probability.   For the regression functions, we first estimate $m_3(\bx, t)$ to obtain $S^{(k)}_{n, DR}(\bbeta, t)$.  This requires us to estimate $P(\tilde{T} \geq t | \bX, Q = 0, L = 1)$, the conditional survival probability of the observed time $\tilde{T}$. We use a Cox regression to estimate this probability.  Next, for $m_1(\bx)$, we again use a logistic regression to estimate $\E[ \Delta | \bX, Q = 0, L = 1]$ as $\Delta$ is binary.  

The estimation for function $m_2(\bx) = \E \left[ \Delta \frac{s^{(1)}(\tilde{T}, \bbeta)}{s^{(0)}(\tilde{T}, \bbeta)} \middle | \bx,  Q = 0, L = 1 \right]$ is trickier as we do not observe $\frac{s^{(1)}(\tilde{T}, \bbeta)}{s^{(0)}(\tilde{T}, \bbeta)}$.   Instead, we use $ \Delta \frac{S^{(1)}_{n, DR}(\tilde{T}, \bbeta)}{S^{(0)}_{n, DR}(\tilde{T}, \bbeta)}$ as the outcome and fit a linear regression model with $\bX$ as covariates to obtain an estimate of $m_2(\bx)$. 

After we obtain all the estimates, we then use Newton-Raphson's algorithm to solve equation \eqref{eqn::eqn_aipw_ee}.  Note as we discussed in \ref{sec::subsec::aipw_difficulty}, the regression functions estimates $\hat m_k(\bx)$ are likely misspecified due to congeniality issues and we only have the linkage probability correctly specified.  However, this still provides a consistent and asymptotically normal estimates $\hat \bbeta$ for the AIPLW methods due to the doubly-robust property.   We obtain the standard error for $\hat \bbeta$ using bootstrap with 100 bootstrap repetitions. 
Now we present the simulation results as follows:

\begin{table}[H]
\scriptsize
\caption{Simulation results for AIPLW with linkage mechanism CLAR and MAR.}
\begin{center}
\begin{tabular}{
 | c 
 | c 
 |c
 |  c 
 | c 
 | c 
 | c 
 | c 
 | c 
 |
}
\hline
 & &   & \multicolumn{3}{c|}{ Bias (Mean SE)} & \multicolumn{3}{c|}{ Coverage of 95\% CI} \\
 \hline
 Mechanism & Method              & n         & $\beta_1$ & $\beta_2$ & $\beta_3$ & $\beta_1$ & $\beta_2$ & $\beta_3$ \\
 \hline
CLAR & Oracle       & 2000     & -0.00  (0.095) & 0.00 (0.038)   & 0.00  (0.018)  & 0.93  & 0.95  & 0.95 \\
\hline
& CC                  & 2000     & -0.15  (0.135) & -0.16 (0.059)   & -0.01 (0.026)  & $0.82^{\dagger}$   & $0.23^{\dagger}$  & 0.93  \\
\hline
& CC+               & 2000     &  -0.00  (0.104) & -0.24 (0.044)  & -0.04 (0.020)  & 0.95   & $0.00^{\dagger}$  & $0.58^{\dagger}$  \\
\hline
& NLAC               & 2000      &  -0.00  (0.104) &  0.00 (0.043)  & 0.00 (0.020)   & 0.93   & 0.96  & 0.95  \\
\hline
& IPLW                  & 2000      & -0.00    (0.116) & 0.00 (0.055)  & 0.00  (0.025)  & 0.94   & 0.95  & 0.94  \\
\hline
& AIPLW                  & 2000      & -0.00    (0.117) & 0.00 (0.060)  & 0.00  (0.027)  & 0.94  & 0.95  & 0.95  \\
\hline
\hline
MAR & Oracle           & 2000     & -0.00  (0.095) & 0.00 (0.038)   & 0.00  (0.018)  & 0.93   & 0.95  & 0.95 \\
\hline
& CC                      & 2000     & -0.00 (0.139) & 0.01 (0.063)   & 0.00  (0.029)  & 0.94   & $0.94$  & 0.94  \\
\hline
& CC+                  & 2000     &  -0.00  (0.104) & -0.24 (0.044)  & -0.04 (0.020)  & 0.94   & $0.00^{\dagger}$  & $0.56^{\dagger}$  \\
\hline
& NLAC                 & 2000      &  -0.00  (0.104)  &  0.00 (0.043)   & 0.00 (0.020)   & 0.94   & 0.96  & 0.94  \\
\hline
& IPLW                   & 2000      &  -0.00   (0.116)  & 0.00 (0.055)  & 0.00  (0.025)  & 0.94    & 0.95  & 0.93  \\
\hline
& AIPLW                   & 2000      &  0.00   (0.116)  & 0.01 (0.058)  & 0.00  (0.026)  & 0.94    & 0.95  & 0.95  \\
\hline
\end{tabular}
\end{center}
\footnotesize{We use $^\dagger$ to highlight settings with coverage below 90\%.}
\label{table::ti_s2_aipw_clar_mar}
\end{table}

\subsection{Mis-specified linkage model under CLAR}
In this section, we present simulation results for both IPLW and AIPLW estimators when the linkage models are mis-specified under the CLAR mechanism.  The data generation setting is the same as in Appendix 
\ref{sec::aipw_simulation_results}.  However, we estimate the linkage probability using a logistic regression model with $X_2$ alone as the covariate.  Thus, the linkage probability is mis-specified in this case as we left out $X_1$. 
\begin{table}[H]
\caption{Simulation results when linkage model is misspecified under CLAR with $n = 2000$. }
\begin{center}
\begin{tabular}{
 |c 
 |c 
 |c
 |  c 
 | c 
 | c 
 | c 
 |
}
\hline
 &   \multicolumn{3}{c}{Bias (Mean SE)} & \multicolumn{3}{c}{Coverage of 95\% CI}  \\
 \hline
 Method              & $\beta_1$ & $\beta_2$ & $\beta_3$ & $\beta_1$ & $\beta_2$ & $\beta_3$ \\
 \hline
\hline
 IPLW & -0.14   (0.121) & 0.02 (0.052)  & -0.03  (0.025)  & $0.79^{\dagger}$   & 0.94   & 0.94   \\
\hline
AIPLW & -0.00   (0.116) & 0.00 (0.058)  & 0.00  (0.026)  & $0.94$   & 0.95   & 0.96   \\
 \hline
\hline
\end{tabular}
\end{center}
\footnotesize{We use $^\dagger$ to highlight settings with coverage below 90\%.}
\label{table::ti_aipw_ipw_linkage_misspecified}
\end{table}
We can see that in this case IPLW obtains biased estimates for $\beta_1$ and the confidence intervals does not reach the 95\% nominal coverages.  On the other hand, AIPLW estimator is able to obtain unbiased estimates for all three parameters and also yield confidence intervals with nominal coverages.

\subsection{Simulation results when there are gaps} \label{sec::simulation_results_ti_gap}
We now present the simulation results when there are gaps between the clinical trial and the observational follow-up dataset.  We consider the following simulation scenario.  The hazard function is $\lambda(t | \bX; \bbeta_0) = \lambda_0(t) \exp(\beta_1 X_1 + \beta_2 X_2 + \beta_3 X_3)$ and $\bbeta_0 = (\beta_1, \beta_2, \beta_3)^T = (-\log(2), \log(2), 0.2)^T$.  $X_1$ is a Bernoulli variable that takes value 1 with probability 0.5.  $X_2$ is a normal random variable with mean -1 and standard deviation 1.  $X_3$ is a normal random variable with mean $1$ and standard deviation 2.  The baseline hazard function is $\lambda_0(t) = 0.15$.  The censoring time in clinical trial $C_1$ is uniformly distributed between 0 and 3.5.  

For each patient,  the probability for a gap between clinical trial and the observational follow-up dataset to exist is 0.5 and the length of the gap $U$ is set as a uniform random variable between 1 and 2.  Thus, the starting time of the observational follow-up period for a participant is set as $C_1$ plus $U$. 
$C_2$ is set as start time of the observational follow-up time plus an exponential random variable with rate $0.8 * X_1 + 0.03$.  We further set $\tau_1 = 3.5$ and $\tau_2 = 16$\footnote{$C_1$ and $C_2$ will be administratively censored by $\tau$ or $\tau_M$}.   The percentage of interval-censored patient is approximately 4.5\%.  This is the scenario we used in Appendix \ref{sec::sec_gap_relaxation}. 


We consider the same three mechanisms for linkage to the medicare data as the simulations with time-dependent covariates in Section 4 of the main paper.
The only difference is that for LCAR,  we have $P(L = 1) = 0.4$.  When Cox model is correctly specified, one additional mechanism for linkage is considered as
\begin{align*}
& P(L = 1 | \bX, C_2,  \Delta; Q = 0) = \frac{\exp(-0.25 + 0.5 * X_1 + 0.5 * X_2 - 0.1 * C_2 - 0.1 * \Delta)}{1 + \exp(-0.25 + 0.5 * X_1 + 0.5 * X_2 - 0.1 * C_2 - 0.1 * \Delta)} \\ & P(L = 1 | Q = 1)  = 0.5.
\end{align*}
This leads to a more serious violation of the CLAR assumption (A1) and linkage now depends on the censoring time in clinical trial $C_2$ conditional on $\bX$. As expected, NLAC should still work under this linkage mechanism as linkage does not depend on the failure time $T$. 
 We again consider sample sizes $n = 500, 1,000, 1,500, \ldots, 10,000$.  For each simulation setting, we generate 1,000 repetitions.
The simulation results are given in 
Table ~\ref{table::ti_s1_lcar} to Table ~\ref{table::ti_s1_nlar}. 
As the results are similar to the simulation studies in the main text, we omit the discussion here. 

\begin{table}[H]
\caption{Simulation results for linkage mechanism (LCAR) and (CLAR).}
\begin{center}
\begin{tabular}{cc c c c c c c c }
\hline
 & &   & \multicolumn{3}{c}{Bias (Mean SE)} & \multicolumn{3}{c}{Coverage of 95\% CI}\\
 \hline
 Mechanism & Method              & n         & $\beta_1$ & $\beta_2$ & $\beta_3$ & $\beta_1$ & $\beta_2$ & $\beta_3$ \\
 \hline
LCAR & Oracle       & 500     & -0.01  (0.219) & 0.01 (0.095)   & 0.00  (0.045)  & 0.96   & 0.94  & 0.95 \\
&                        & 2000   & -0.00  (0.108) & 0.00 (0.047)   & 0.00  (0.022)  & 0.94   & 0.95  & 0.94 \\
\hline
& CC                  & 500     & -0.03  (0.355) & 0.02 (0.155)   & 0.01  (0.073)  & 0.97   & 0.94  & 0.96  \\
&                        & 2000   & -0.01  (0.172) & 0.00  (0.075)  & 0.00  (0.036)  & 0.96   & 0.96  & 0.94  \\
\hline
& CC+               & 500     &  0.11   (0.241) & -0.06 (0.111)  & -0.02 (0.054)  & 0.92   & 0.91  & 0.934  \\
&                         & 2000   &  0.12   (0.118) & -0.07 (0.054)  & -0.02 (0.026)  & $0.81^{\dagger}$  & $0.72^{\dagger}$  & $0.86^{\dagger}$  \\
\hline
& NLAC               & 500      &  -0.01  (0.240) &  0.01 (0.112)  & 0.00 (0.054)   & 0.96   & 0.95  & 0.95  \\
&                          & 2000    &  -0.00  (0.118) &  0.00 (0.055)  & 0.00 (0.026)   & 0.94   & 0.95  & 0.94  \\
\hline
& IPLW                  & 500      & -0.02    (0.271) & 0.02 (0.125)  & 0.01  (0.060)  & 0.95   & 0.92  & 0.93  \\
&                          & 2000    & -0.01    (0.134) & 0.00 (0.063)  & 0.00  (0.030)  & 0.95   & 0.95  & 0.93   \\
\hline
\hline
CLAR & Oracle           & 500     & -0.01  (0.219) & 0.01 (0.095)   & 0.00  (0.045)  & 0.96   & 0.94  & 0.95 \\
&                            & 2000   & -0.00  (0.108) & 0.00 (0.047)   & 0.00  (0.022)  & 0.94   & 0.95  & 0.94 \\
\hline
& CC                      & 500     & -0.18  (0.320) & -0.14 (0.153)   & 0.00  (0.070)  & 0.92   & $0.83^{\dagger}$  & 0.95  \\
&                            & 2000   & -0.16  (0.156) & -0.15  (0.074)  & -0.00  (0.034)  & $0.80^{\dagger}$   & $0.45^{\dagger}$  & 0.95  \\
\hline
& CC+                  & 500     &  -0.08   (0.239) & -0.22 (0.117)  & -0.02 (0.054)  & 0.95   & $0.51^{\dagger}$  & 0.94  \\
&                            & 2000   &  -0.07   (0.118) & -0.23 (0.057)  & -0.02 (0.027)  & $0.89^{\dagger}$   & $0.03^{\dagger}$  & $0.88^{\dagger}$  \\
\hline
& NLAC                 & 500      &  -0.01  (0.238)  &  0.01 (0.114)   & 0.00 (0.054)   & 0.95   & 0.95  & 0.95  \\
&                           & 2000    &  -0.00  (0.117)   &  0.00 (0.056)  & 0.00 (0.027)   & 0.95   & 0.95  & 0.94  \\
\hline
& IPLW                   & 500      &  -0.02   (0.263)  & 0.02 (0.134)  & 0.01  (0.063)  & 0.94    & 0.91  & 0.93  \\
&                          & 2000    &  -0.00   (0.130)  & 0.01 (0.069)  & 0.00  (0.032)  & 0.95    & 0.95  & 0.94   \\
\hline
\end{tabular}
\end{center}
\footnotesize{We use $^\dagger$ to highlight settings with coverage below 90\%.}
\label{table::ti_s1_lcar}
\end{table}

\begin{table}[H]
\caption{Simulation results for linkage mechanism (LNAR($\tilde{T}$)) and (LNAR($C_2$)). }
\begin{center}
\begin{tabular}{c c c c c c c c c }
\hline
 & &   & \multicolumn{3}{c}{Bias (Mean SE)} & \multicolumn{3}{c}{Coverage of 95\% CI}\\
 \hline
 Mechanism & Method              & n         & $\beta_1$ & $\beta_2$ & $\beta_3$ & $\beta_1$ & $\beta_2$ & $\beta_3$ \\
 \hline
LNAR($\tilde{T}$) & Oracle           & 500     & -0.01  (0.219) & 0.011 (0.095)   & 0.00  (0.045)  & 0.96   & 0.94  & 0.95 \\
&                            & 2000   & -0.00  (0.108) & 0.00 (0.047)   & 0.00  (0.022)  & 0.94   & 0.95  & 0.94 \\
\hline
& CC                      & 500     & -0.20  (0.320) & -0.16 (0.154)   & 0.00  (0.071)  & 0.92   & $0.80^{\dagger}$  & 0.95  \\
&                            & 2000   & -0.19  (0.156) & -0.17  (0.074)  & -0.01  (0.034)  & $0.77^{\dagger}$   & $0.38^{\dagger}$  & 0.94  \\
\hline
& CC+                  & 500     &  -0.10   (0.239) & -0.24 (0.118)  & -0.02 (0.055)  & 0.94   & $0.47^{\dagger}$  & 0.94  \\
&                            & 2000   &  -0.1   (0.118) & -0.24 (0.058)  & -0.02 (0.027)  & $0.87^{\dagger}$   & $0.02^{\dagger}$  & $0.86^{\dagger}$  \\
\hline
& NLAC                 & 500      &  -0.01  (0.238)  &  0.01 (0.114)   & 0.00 (0.055)   & 0.95   & 0.95  & 0.95  \\
&                          & 2000    &  -0.01  (0.118)   &  0.00 (0.056)  & -0.00 (0.027)   & 0.94   & 0.95  & 0.94  \\
\hline
& IPLW                   & 500      &  -0.02   (0.265)  & 0.01 (0.137)   & 0.01  (0.064)  & 0.94    & 0.92  & 0.92  \\
&                           & 2000    &  -0.01   (0.131)  & -0.00 (0.070)  & -0.00 (0.032)  & 0.94    & 0.95  & 0.94   \\
\hline
\hline
LNAR($C_2$) & Oracle                & 500     & -0.01  (0.219) & 0.01 (0.095)   & 0.00  (0.045)  & 0.96   & 0.94  & 0.95 \\
&                            & 2000   & -0.00  (0.108) & 0.00 (0.047)   & 0.00  (0.022)  & 0.94   & 0.95  & 0.94 \\
\hline
& CC                      & 500     & -0.63  (0.340) & -0.27 (0.176)   & -0.02  (0.082)  & $0.55^{\dagger}$   & $0.63^{\dagger}$  & 0.95  \\
&                           & 2000   & -0.60  (0.165) & -0.28  (0.083)  & -0.02  (0.039)  & $0.04^{\dagger}$   & $0.09^{\dagger}$  & 0.91  \\
\hline
& CC+                  & 500     &  -0.46   (0.248) & -0.35 (0.127)  & -0.05 (0.060)  & $0.54^{\dagger}$   & $0.24^{\dagger}$  & $0.87^{\dagger}$  \\
&                           & 2000   &  -0.44   (0.122) & -0.36 (0.061)  & -0.05 (0.029)  & $0.04^{\dagger}$   & $0.00^{\dagger}$  & $0.62^{\dagger}$  \\
\hline
& NLAC                 & 500      &  -0.01  (0.246)  &  0.01 (0.123)   &  0.00 (0.059)   & 0.96   & 0.95  & 0.95  \\
&                          & 2000    &  -0.01  (0.121)   &  0.00 (0.060)  &  0.00 (0.029)   & 0.94   & 0.95  & 0.95  \\
\hline
& IPLW                   & 500      &  -0.04   (0.316)  & 0.04 (0.180)   &  0.02  (0.088)  & 0.94    & $0.87^{\dagger}$  & $0.89^{\dagger}$  \\
&                          & 2000    &  -0.01   (0.157)  & 0.01 (0.101)   &  0.01 (0.047)  & 0.94    & 0.91  & 0.91   \\
\hline
\end{tabular}
\end{center}
\footnotesize{We use $^\dagger$ to highlight settings with coverage below 90\%.}
\label{table::ti_s1_nlar}
\end{table}

\subsection{Simulation setting in Section 3.1.2
}  \label{appendix::simulation_setting}
 The hazard function is $\lambda(t | x) = \lambda_0(t) \exp(\beta_1 X_1 + \beta_2 X_2 + \beta_3 X_3^2)$ and $\bbeta_0 = (\beta_1, \beta_2, \beta_3)^T = (-\log(2), \log(2), 0.2)^T$.  $X_1$ is a Bernoulli variable that takes value 1 with probability 0.5.  $X_2$ is a normal random variable with mean -1 and standard deviation $1$.  $X_3$ is a normal random variable with mean $X_2$ and standard deviation 2. The baseline hazard function is $\lambda_0(t) = 0.05$. The censoring time in clinical trial $C_1$ is exponentially distributed with rate $0.1 * X_1 + 0.05$.  The censoring time $C_2$ is defined as $C_1$ plus an exponential random variable with rate $0.8 * X_1 + 0.03$.   Further,  we set $\tau_1 = 3$ and $\tau_2 = 16$.   


\subsection{SWOG study long-term effect estimation with no change points}
We further fit a Cox regression without any change points and the results are shown below. 
\begin{table}[H]
  \caption{SWOG study long-term effect estimation with no change points}
  \begin{center}
  \begin{tabular}{ c c }
  \hline
  \hline
   Methods   &  Finasteride\\
   \hline
  IPLW & \gc{0.901 (0.855 - 0.951)} \\
  CC & \gc{0.901 (0.854 - 0.951)} \\
  CC+ & \gc{0.886 (0.841 - 0.933)}  \\
  NLAC & \gc{0.912 (0.871 - 0.955)}  \\
  \hline
  \hline
  \end{tabular}
  \end{center}
  \label{table::real_data_results_with_no_change_point}
\end{table}

\section{Additional Future Directions}	\label{sec::further}

\begin{itemize}

\item{\bf Alternative working models.} While the hazard ratio is usually used to quantify treatment effects,  it is most meaningful when the hazard functions are approximately proportional over time, which may not be true in current clinical practice. For example, cancer immunotherapy, which works through activating the immune system, has gained substantial attraction in recent years. However, delayed effects may arise with these therapies since they do not act directly on a patient's tumor, and the  proportional hazards assumption  is often violated.  To address this, alternative measures have been proposed to assess treatment effects that are robust even in the presence of violations of the proportional hazards assumption. One such measure is milestone survival, which quantifies the probability of survival beyond clinically relevant time points.
In this case, a flexible $\tau$-year specific generalized linear model ($\tau$-GLM)  \citep{uno2007evaluating} can be fitted.  Let $T$ denote the failure time. Assume that the $\tau$-GLM working model is imposed $P(Y_\tau = 1 | X) = g(\beta^T X)$, where $Y_\tau = I(T\leq \tau)$ and $g(\cdot)$ is a known smooth function. Let $\bar\beta$ denote the unique solution to  $$ U(\beta) = E[X\{Y_\tau - g(\beta^T X)\}] = 0.  $$  When the model is correctly specified, $\bar\beta$ is the true model parameter.  Under mild model mis-specification,  the resulting risk score ${\bar\beta}^T X$ can be shown to  approximately maximize  a weighted area under the receiver operating characteristic curve among all functions of $X$ for classifying  $Y_\tau$, indicating that $\bar\beta$ also has a meaningful interpretation in this case. We plan to explore the potential generalization of the proposed methods to accommodate these alternative models in the future.

\item {\bf Missing covariates.}
Another direction that we will be exploring is the case of missing covariates \citep{tsiatis2007semiparametric}.
Missing covariates is a common issue in medical research. 
When part of $\bX$ is missing, CLAR will no longer be enough to identify the underlying parameter
since the linkage probability cannot be computed for every individual.
In this case, we have to impose additional assumptions on the missingness of $\bX$.
However, such assumption has to be carefully chosen so that it will not conflict with the assumption on the linkage. 

\item {\bf Beyond CLAR and sensitivity analysis.}
CLAR may not hold in certain situations. 
For instance, if the data being linked is from another study, in which
the time to event variable $T$ may influence the chance that someone participates, then CLAR may no longer be true.
In this case,  we need to model the linkage probability that depends on $T$, which could be seen as a sensitivity analysis \citep{little2012prevention, little2019statistical} on perturbing CLAR.  How to analyze the data in this case is left as a future work.

\end{itemize}

\bibliographystyle{apalike}
\bibliography{gang}

\end{document}